\documentclass[11pt]{article}
\usepackage[margin=1in]{geometry}
\usepackage[colorlinks=true,
            linkcolor=blue,
            citecolor=blue,
            urlcolor=blue]{hyperref}

\usepackage{subcaption}

\usepackage{fixmath}
\usepackage{bm}
\usepackage{amsbsy}
\usepackage{color}
\usepackage{verbatim}
\usepackage{multirow}
\usepackage{amssymb}
\usepackage{amsthm}
\usepackage{amsmath}
\usepackage{array}
\usepackage{mathtools}

\allowdisplaybreaks

\usepackage{graphicx}
\usepackage{tikz}
\usetikzlibrary{positioning}
\usepackage{xcolor}

\usepackage{thm-restate}
\usepackage{algpseudocode}
\usepackage{algorithm}
\usepackage{algorithmicx}

\newtheorem{theorem}{Theorem}
\newtheorem{lemma}{Lemma}
\newtheorem{corollary}{Corollary}
\newtheorem{definition}{Definition}
\newtheorem{assumption}{Assumption}

\newtheorem{claim}{Claim}

\newcommand{\size}[1]{\ensuremath{|#1|}}
\newcommand{\ceil}[1]{\ensuremath{\lceil#1\rceil}}
\newcommand{\floor}[1]{\ensuremath{\lfloor#1\rfloor}}

\newcommand{\Ceil}[1]{\ensuremath{\left\lceil#1\right\rceil}}
\newcommand{\Floor}[1]{\ensuremath{\left\lfloor#1\right\rfloor}}

\newcommand{\lra}[1]{\ensuremath{(#1)}}

\newcommand{\lrc}[1]{\ensuremath{\{#1\}}}

\newcommand{\lrA}[1]{\ensuremath{\left(#1\right)}}

\newcommand{\lrC}[1]{\ensuremath{\left\{#1\right\}}}

\def\OPT{\mbox{OPT}}
\def\MST{\mbox{MST}}
\def\C{\mathcal{C}}
\def\M{\mathcal{M}}

\def\T{\mathcal{T}}

\def\OPT{\mbox{OPT}}
\def\MST{\mbox{MST}}
\def\C{\mathcal{C}}
\def\M{\mathcal{M}}

\def\ra{\chi}

\title{Improved Approximation Algorithms for Capacitated Vehicle Routing with Fixed Capacity}

\author
{
Jingyang Zhao\\
University of Electronic Science and Technology of China\\
\texttt{jingyangzhao1020@gmail.com}
\and
Mingyu Xiao\\
University of Electronic Science and Technology of China\\
\texttt{myxiao@uestc.edu.cn}
}

\date{}

\begin{document}

\maketitle

\begin{abstract}
The Capacitated Vehicle Routing Problem (CVRP) is one of the most extensively studied problems in combinatorial optimization. Based on customer demand, we distinguish three variants of CVRP: unit-demand, splittable, and unsplittable. In this paper, we consider $k$-CVRP in general metrics and on general graphs, where $k$ is the vehicle capacity. All three versions are APX-hard for any fixed $k\geq3$. Assume that the approximation ratio of metric TSP is $\frac{3}{2}$. We present a $(\frac{5}{2}-\Theta(\frac{1}{\sqrt{k}}))$-approximation algorithm for the splittable and unit-demand cases, and a $(\frac{5}{2}+\ln2-\Theta(\frac{1}{\sqrt{k}}))$-approximation algorithm for the unsplittable case. Our approximation ratio is better than the previous results when $k$ is less than a sufficiently large value, approximately $1.7\times10^7$.

For small values of $k$, we design independent and elegant algorithms with further improvements. For the splittable and unit-demand cases, we improve the approximation ratio from $1.792$ to $1.500$ for $k=3$, and from $1.750$ to $1.500$ for $k=4$. For the unsplittable case, we improve the approximation ratio from $1.792$ to $1.500$ for $k=3$, from $2.051$ to $1.750$ for $k=4$, and from $2.249$ to $2.157$ for $k=5$. The approximation ratio for $k=3$ surprisingly achieves the same value as in the splittable case. Our techniques, such as EX-ITP -- an extension of the classic ITP method, have the potential to improve algorithms for other routing problems as well.
\end{abstract}

\maketitle

\section{Introduction}
In the Capacitated Vehicle Routing Problem (CVRP), we are given an undirected complete graph $G=(V\cup\{v_0\}, E)$ with edge weights $w$ satisfying the symmetric and triangle inequality properties. The $n$ nodes in $V=\{v_1,\dots,v_n\}$ represent $n$ customers and each customer $v_i$ has a demand $d_i\in\mathbb{Z}_{\geq 1}$. A vehicle with a capacity of $k\in\mathbb{Z}_{\geq 1}$ is initially located at the depot $v_0$. A tour is a walk that begins and ends at the depot and the sum of deliveries to all customers in it is at most $k$. The distance of a tour is the sum of the weights of all edges in the tour. In CVRP, we wish to find a set of tours to satisfy the demand of every customer with minimum total distance. We use $k$-CVRP to denote the problem where the capacity $k$ is a fixed integer. In the \emph{unsplittable} version of the problem, each customer's demand can only be delivered by a single tour. In the \emph{splittable} version, each customer's demand can be delivered by more than one tour. If each customer's demand is 1, it is called the \emph{unit-demand} version.

Since CVRP was raised by Dantzig and Ramser~\cite{dantzig1959truck} in 1959, it has become a very famous problem with numerous applications in combinatorial optimization. It has been widely studied in both theory and application. Readers can refer to a survey~\cite{toth2014vehicle} for its applications and fast solvers in practice. In theory, it is a rather rich problem in approximation algorithms~\cite{BompadreDO06,blauth2022improving,cvrptree5,abs220205691,Cohen-AddadFKL20,abs-2203-15627}.

When $k=1$ or $k=2$, $k$-CVRP can be solved in polynomial time~\cite{AsanoKTT97+}. However, for each fixed $k\geq3$, the problem becomes APX-hard, even for the unit-demand case~\cite{AsanoKTT97}. A classic algorithm known as \emph{Iterated Tour Partitioning (ITP)} was introduced about 40 years ago~\cite{HaimovichK85}. ITP is not only efficient in practice but also achieves strong theoretical approximation guarantees. Assuming we are given an $\alpha$-approximation algorithm for metric TSP, for both splittable $k$-CVRP and unit-demand $k$-CVRP, ITP achieves an approximation ratio of $\alpha+1-\frac{\alpha}{k}$~\cite{HaimovichK85}. For unsplittable $k$-CVRP, a variant of ITP called UITP achieves an approximation ratio of $\alpha+2-\frac{2\alpha}{k}$ for even $k$ and $\alpha+2-\frac{\alpha}{k}$ for odd $k$~\cite{altinkemer1987heuristics}.

For metric TSP, there is a well-known $\frac{3}{2}$-approximation algorithm~\cite{christofides1976worst,serdyukov1978some}. 
Although a recent breakthrough by Karlin, Klein, and Oveis Gharan~\cite{KarlinKG21,DBLP:conf/ipco/KarlinKG23} has improved the approximation ratio to \(\frac{3}{2} - \varepsilon\), where \(\varepsilon \approx 10^{-36}\), this improvement is too small to significantly affect the analysis for \( k \)-CVRP. Thus, for ease of comparison, we may continue to assume the approximation ratio \(\alpha = \frac{3}{2}\) for metric TSP in our algorithms.

Due to its simplicity and versatility, ITP has been adapted for various other vehicle routing problems~\cite{Mathieu021}. However, for general metrics, there have been few improvements over ITP in approximating \( k \)-CVRP.

One interesting improvement was done by Bompadre \emph{et al.}~\cite{BompadreDO06} about 20 years ago. 
For any \(\alpha \geq 1\), they improved the approximation ratio by a term of \(\frac{1}{3k^3}\) for all three versions of \(k\)-CVRP. Specifically, for \(\alpha = \frac{3}{2}\), the improvement was \(\frac{1}{4k^2}\) for the splittable and unit-demand cases, and \(\frac{1}{3k^2}\) for the unsplittable case. These results are still the best-known approximation ratios for many small values of \(k\).
Recently, one significant progress was done by Blauth \emph{et al.}~\cite{blauth2022improving}. They improved the approximation ratio to $\alpha+1-\varepsilon$ for the splittable and unit-demand cases, and to $\alpha+2-2\varepsilon$ for the unsplittable case, where $\varepsilon$ is a small value related to $\alpha$, with $\varepsilon>\frac{1}{3000}$ when $\alpha=\frac{3}{2}$.
While this provides a slight improvement in the constant part of the approximation ratio, it does not outperform the approximation ratio in~\cite{BompadreDO06} for small values of \(k\).
Friggstad \emph{et al.}~\cite{uncvrp} proposed two further improvements for unsplittable $k$-CVRP. The first is an $(\alpha+1.75)$-approximation algorithm using a combinatorial method, while the second is an $(\alpha+\ln 2+\frac{1}{1-\delta})$-approximation algorithm based on LP rounding, with a running time of \(n^{O(\frac{1}{\delta})}\). They also showed that both approximation ratios can be further improved by a small constant \(\varepsilon'\) through the method from~\cite{blauth2022improving}.

For small $k$, $k$-CVRP has garnered independent interest~\cite{4cvrp,HassinR00,BompadreDO06}. For unit-demand 3-CVRP, ITP achieves an approximation ratio of $2$. Based on the $\frac{25}{33}$-approximation algorithm for MAX TSP on general graphs~\cite{HassinR00}, Bazgan \emph{et al.}~\cite{BazganHM05} proposed a $1.990$-approximation algorithm, marking the first known improvement. 
Leveraging the best-known $\frac{4}{5}$-approximation algorithm for MAX TSP~\cite{DudyczMPR17}, the approximation ratio of their algorithm can be further improved to $1.934$.
For unit-demand 4-CVRP, ITP achieves an approximation ratio of 2.125. Anily and Bramel~\cite{4cvrp} proposed an approximation algorithm for \emph{Capacitated TSP with Pickups and Deliveries}, demonstrating that their algorithm can be applied to unit-demand \(k\)-CVRP. Their approach achieves an improved approximation ratio only for the case \(k=4\), with an approximation ratio of 1.750, which remains the best-known result.

In addition to the aforementioned algorithms, another approach for solving \(k\)-CVRP with \(k = O(1)\) is to reduce it to the more general \emph{minimum weight \(k\)-set cover problem}~\cite{chvatal1979greedy,hassin2005better,gupta2023local}. By treating each feasible tour as a \(k\)-set, this reduction applies to both unit-demand \(k\)-CVRP and unsplittable \(k\)-CVRP. For splittable \(k\)-CVRP, it can be shown that the problem is equivalent to unit-demand \(k\)-CVRP when \(k = n^{O(1)}\). Since there are \(n^{O(k)}\) feasible tours, the reduction remains polynomial for \(k = O(1)\). 
Recently, Gupta \emph{et al.}~\cite{gupta2023local} improved the approximation ratio for the minimum weight \(k\)-set cover problem to \(\min\{H_k - \frac{1}{8k}, H_k - \sum_{i=1}^{k}\frac{\log i}{8ki}\}\), where \(H_k = \sum_{i=1}^k\frac{1}{i}\) is the \(k\)-th harmonic number. For \(k = 3\), this achieves an approximation ratio of 1.792, surpassing all previous results for both splittable 3-CVRP and unit-demand 3-CVRP. Furthermore, this algorithm yields the best-known approximation ratio for unsplittable \(k\)-CVRP when \(k \leq 9\).

In this paper, we focus on approximation algorithms for \(k\)-CVRP with bounded \(k\). As highlighted in~\cite{practicalcvrp,BompadreDO06}, many practical problems involve small values of \(k\). 
For example, a logistics company raised the need to transport newly produced cars using a truck with a capacity of at most six cars, where the capacity \(k = 6\). 
In benchmark datasets, the capacities of the instances are also typically in the range of hundreds or thousands~\cite{bench}. These examples demonstrate the practical relevance and significance of studying \(k\)-CVRP with bounded \(k\).

\medskip
\noindent\textbf{Other Related Work.} Although there are little improvements on the general metrics, a huge number of contributions have been made for special graph classes. We list some of them.

Consider unit-demand $k$-CVRP. In $\mathbb{R}^2$, a PTAS is known for $k=O(\log\log n)$~\cite{HaimovichK85}, $k=O(\frac{\log n}{\log\log n})$ or $k=\Omega(n)$~\cite{AsanoKTT97+}, and $k\leq 2^{\log^{f(\varepsilon)}n}$~\cite{AdamaszekCL10}. In $\mathbb{R}^l$, where $l$ is fixed, a PTAS is known for $k=O(\log^{\frac{1}{l}}n)$~\cite{KhachayD16}.
Das and Mathieu~\cite{DasM15} proposed a QPTAS in $\mathbb{R}^2$ for arbitrary $k$ with a running time of $n^{\log^{O(1/\varepsilon)}n}$. The running time was improved to $n^{O(\log^6 n/\varepsilon^5)}$ by Jayaprakash and Salavatipour~\cite{JayaprakashS22}, where they also showed a QPTAS for graphs of bounded treewidth, bounded doubling metrics, or bounded highway dimension with arbitrary $k$.

Consider unit-demand $k$-CVRP with $k=O(1)$. Becker \emph{et al.}~\cite{BeckerKS17,BeckerKS19} showed a QPTAS for planar and bounded genus graphs, a PTAS in bounded highway dimension, and an $O(n^{tw\cdot Q})$ time exact algorithm for graphs with treewidth $tw$.
A PTAS in planar graphs was given in~\cite{BeckerKS19}.
Cohen-Addad~\emph{et al.}~\cite{Cohen-AddadFKL20} showed an EPTAS for graphs of bounded-treewidth, bounded highway dimension, or bounded genus metrics, and the first QPTAS for minor-free metrics. 
Filtser and Le~\cite{abs-2203-15627} proposed a better EPTAS for planar graphs that runs in \emph{almost linear} time, and an EPTAS for minor-free metrics.

For the case of trees, splittable CVRP is NP-hard~\cite{LabbeLM91}, and unsplittable CVRP is NP-hard to approximate better than $\frac{3}{2}$~\cite{GoldenW81}. For splittable CVRP on trees, Hamaguchi and Katoh~\cite{cvrptree1} proposed a $\frac{3}{2}$-approximation algorithm. The approximation ratio was improved to $\frac{\sqrt{41}-4}{4}$ by Asano \emph{et al.}~\cite{cvrptree2} and to $\frac{4}{3}$ by Becker~\cite{cvrptree3}. Becker and Paul~\cite{cvrptree4} proposed a $(1,1+\varepsilon)$-bicriteria polynomial time scheme where the capacity of each tour is allowed to be violated by an $\varepsilon$ fraction. Finally, a PTAS was given by Mathieu and Zhou~\cite{cvrptree5}.
For the unsplittable CVRP on trees, Labb{'{e}} \emph{et al.}~\cite{LabbeLM91} proposed a 2-approximation algorithm, which was  improved to a tight approximation ratio of $\frac{3}{2}+\varepsilon$ by Mathieu and Zhou~\cite{abs220205691}.

Grandoni \emph{et al.}~\cite{abs-2209-05520} gave a $(2+\varepsilon)$-approximation algorithm for unplittable CVRP in $\mathbb{R}^2$. Dufay \emph{et al.}~\cite{abs-2210-03811} showed a $(1.692+\varepsilon)$-approximation algorithm for distance-constrained CVRP on trees.
M{\"o}mke and Zhou~\cite{momke2022capacitated} showed an 1.95-approximation algorithm for graphic CVRP.

\subsection{Our Contributions}
We present several approximation algorithms for $k$-CVRP, improving the best-known approximation ratios for any $k\leq 1.7\times 10^7$. 
This threshold is quite large, and to our knowledge, the capacity of all practical and artificial instances does not exceed this value.
We summarize our contributions to splittable (including unit-demand) $k$-CVRP and unsplittable $k$-CVRP separately below.

For splittable $k$-CVRP and unit-demand $k$-CVRP, we have the following contributions.

\begin{enumerate}
\item[1.] Based on a new concept of \emph{home-edges}, which are edges incident to the depot in an optimal solution, we obtain tighter lower bounds for $k$-CVRP based on the minimum weight spanning tree and several different kinds of cycle covers (Section~\ref{SEC3}). 
Our analysis shows that if the lower bound used for the connection part of ITP in~\cite{altinkemer1987heuristics,HaimovichK85} is tight, the optimal solution weight (denoted by \(\OPT\)) will be dominated by home-edges, enabling the computation of a Hamiltonian cycle with cost \(\OPT\) and a cycle cover with zero cost in polynomial time. 
Since the upper bound on the cost of the Hamiltonian cycle in~\cite{altinkemer1987heuristics,HaimovichK85} is $\frac{3}{2}\OPT$, we can obtain significant improvements for this case. Otherwise, since the lower bound for the connection part is not tight, we obtain improvements as well.
Note that our lower bounds are also suitable for the unsplittable case.

\item[2.] The classic ITP is based on a given Hamiltonian cycle of the graph. We extend ITP to an algorithm based on any cycle cover, called EX-ITP. One advantage is that an optimal cycle cover is polynomially computable~\cite{hartvigsen1984extensions,schrijver2003combinatorial}. Based on EX-ITP with our new lower bounds, we can quickly improve the approximation ratio from 1.792 to 1.500 for 3-CVRP and from 1.750 to 1.667 for 4-CVRP. Additionally, based on a good structural property related to perfect matchings, we can surprisingly improve the approximation ratio to $1.500$ for 4-CVRP (Section~\ref{SEC5}).
Then, to obtain improvements for larger $k$, we also consider mod-$k$-cycle covers (the length of each cycle in it is divisible by $k$). We show that given an $\alpha$-approximation algorithm for metric TSP, by making a trade-off among two Hamiltonian cycles and a cycle cover with ITP and EX-ITP, we can achieve an approximation ratio of $\alpha+1-\frac{\alpha}{k}-\Theta(\frac{1}{k})$, improving the previous approximation ratio of $\alpha+1-\frac{\alpha}{k}-\max\{\Omega(\frac{1}{k^3}),\varepsilon\}$~\cite{BompadreDO06,blauth2022improving} for small $k$ (Section~\ref{sec-split-initial})\footnote{With a more careful analysis, the previous results may be further improved slightly. 
However, the improvement is small and not mentioned in any published paper.}.
Our precise approximation ratio is presented in Theorem~\ref{res-intitial-split}, and some numerical values for specific $k$ under $\alpha=\frac{3}{2}$ are shown in Table~\ref{res-1} (the lines for Section~\ref{sec-split-initial}). Note that if $\alpha$ is smaller, we obtain better results.

\item[3.] The best approximation ratio for metric TSP is still about $\frac{3}{2}$. Under $\alpha=\frac{3}{2}$, the current best approximation ratio of $k$-CVRP is approximately $\frac{5}{2}-\frac{1.5}{k}-\max\{\Omega(\frac{1}{k^2}),\frac{1.005}{3000}\}$~\cite{BompadreDO06,blauth2022improving}.
By generalizing the concept of home-edges, we improve two previously-used lower bounds (previously used in Section~\ref{sec-split-initial}). 
Based on these bounds and a detailed analysis of ITP, we derive an approximation ratio of $\frac{5}{2}-\Theta(\frac{1}{\sqrt{k}})$, which improves the previous results for any $k\leq 1.7\times 10^7$ (Section~\ref{sec-split-final}). Our precise approximation ratio is presented in Theorem~\ref{res-split}, and some numerical values for specific $k$ are shown in Table~\ref{res-1} (the lines for Section~\ref{sec-split-final}).
\end{enumerate}

\begin{table}[ht]
\small
\centering
\resizebox{0.68\textheight}{!}{
\begin{tabular}{ccccccccc}
\hline
    $k$ & 3 & 4 & 5 & 6 & 7 & 8 & 9 &10\\
\hline
    Previous & {1.792} & {1.750} & {2.188} & {2.242} & {2.280} & {2.308} & {2.330} &{2.348}\\
    Results&~\cite{gupta2023local} &~\cite{4cvrp} &~\cite{BompadreDO06} &~\cite{BompadreDO06} &~\cite{BompadreDO06} &~\cite{BompadreDO06} &~\cite{BompadreDO06}&~\cite{BompadreDO06}\\
\hline
    Our Results & \multirow{2}*{$\boldsymbol{1.500}$} & \multirow{2}*{$\boldsymbol{1.500}$} & \multirow{2}*{-} & \multirow{2}*{-} & \multirow{2}*{-} & \multirow{2}*{-} & \multirow{2}*{-}&\multirow{2}*{-}\\
    in Section~\ref{SEC5}& \\
\hline
    Our Results  & \multirow{2}*{$1.556$} & \multirow{2}*{$1.709$} & \multirow{2}*{$\boldsymbol{1.800}$} & \multirow{2}*{$1.917$} & \multirow{2}*{$2.000$} & \multirow{2}*{$2.063$} & \multirow{2}*{$2.112$}& \multirow{2}*{$2.150$} \\
    in Section~\ref{sec-split-initial}  &\\
\hline
    Our Results & \multirow{2}*{$1.667$} & \multirow{2}*{$1.750$} & \multirow{2}*{$1.800$} & \multirow{2}*{$\boldsymbol{1.875}$} & \multirow{2}*{$\boldsymbol{1.929}$} & \multirow{2}*{$\boldsymbol{1.969}$} & \multirow{2}*{$\boldsymbol{2.000}$}& \multirow{2}*{$\boldsymbol{2.025}$} \\
     in Section~\ref{sec-split-final}&\\
\hline
\end{tabular}
}

\medskip

\resizebox{0.68\textheight}{!}{
\begin{tabular}{cccccccc}
\hline
    $k$ & 29 & $\dots$ & 5833 & 5834 & $\dots$ & $1.7\times 10^7$ & $1.8\times 10^7$\\
\hline
    Previous & {$2.44795$} & {$\dots$} & {$2.49941$} & {$2.49941$} & {$\dots$} & {$2.49967$} & {$\boldsymbol{2.49967}$}\\
    Results&~\cite{blauth2022improving}&$\dots$&\cite{blauth2022improving}&\cite{blauth2022improving}&$\dots$&\cite{blauth2022improving}&\cite{blauth2022improving}\\
\hline
    Our Results & \multirow{2}*{$2.37932$} & \multirow{2}*{$\dots$} & \multirow{2}*{$2.49940$} & \multirow{2}*{$2.49941$} & \multirow{2}*{$\dots$} & \multirow{2}*{$2.50000$} & \multirow{2}*{$2.50000$}\\
    in Section~\ref{sec-split-initial}&\\
\hline
    Our Results & \multirow{2}*{$\boldsymbol{2.22414}$} & \multirow{2}*{$\boldsymbol{\dots}$} & \multirow{2}*{$\boldsymbol{2.48140}$} & \multirow{2}*{$\boldsymbol{2.48141}$} & \multirow{2}*{$\boldsymbol{\dots}$} & \multirow{2}*{$\boldsymbol{2.49966}$} & \multirow{2}*{$2.49967$}\\
    in Section~\ref{sec-split-final}&\\
\hline
\end{tabular}
}
\caption{
Splittable $k$-CVRP and unit-demand $k$-CVRP: previous and our approximation ratios for different values of $k$ under $\alpha=\frac{3}{2}$, where the best results are marked in bold.
}
\label{res-1}
\end{table}

For unsplittable $k$-CVRP, we obtain similar results. However, we may use more techniques.

\begin{enumerate}
\item[1.] We first propose a refined UITP, improving its approximation ratio for fixed $k$. 
Building on this, we extend UITP to EX-UITP (Section~\ref{Sec-REUITP}), similar to the EX-ITP approach.
However, due to the unsplittable constraint, EX-UITP requires that each customer's demand cannot be too large. For cycles involving large-demand customers, we assign a single tour to each, then apply shortcutting to create a simplified cycle before applying EX-UITP. 
While shortcutting can significantly alter a cycle's structure, we carefully analyze these cycles' local properties and address them separately (Section~\ref{Sec-Unsplit-Property}).  
Using this approach, along with ideas from algorithms for splittable 3-CVRP and 4-CVRP, we improve the approximation ratios for unsplittable 3-CVRP to 1.500 and 4-CVRP to 1.750 (Section~\ref{Sec-Unsplit-34}).

\item[2.] To achieve improvements for larger $k$, we combine the refined UITP with the LP-based technique from~\cite{uncvrp} to obtain the LP-UITP algorithm.
By making a trade-off between two Hamiltonian cycles using LP-UITP\footnote{For the unsplittable case, we focus on the trade-off between two algorithms instead of three, as the third algorithm based on cycle cover is not applicable here.}, we obtain an approximation ratio of $\alpha+1+\ln2-\frac{2\alpha}{k}-\Theta(\frac{1}{k})$, improving the previous approximation ratio of $\alpha+1+\ln2-\frac{2\alpha}{k}-\varepsilon'$~\cite{uncvrp} for small $k$  (Section~\ref{Sec-Split-Initial}).
Under $\alpha=\frac{3}{2}$, the precise approximation ratio is presented in Theorem~\ref{res-intitial-unsplit}, and some numerical values for specific $k$ are shown in Table~\ref{res-2} (the lines for Section~\ref{Sec-Split-Initial}).
Note that $\alpha$ is smaller, we obtain better results.

\item[3.] Based on the refined two lower bounds (provided in Section~\ref{sec-split-final}) and using a deep analysis on LP-UITP, we achieve an approximation ratio of $\frac{5}{2}+\ln2-\Theta(\frac{1}{\sqrt{k}})$, improving the previous approximation ratio of $\frac{5}{2}+\ln2+\ln(1-\frac{1.005}{3000})-\frac{3}{k}$ in~\cite{uncvrp} for any $k\leq 1.7\times10^7$ (Section~\ref{Sec-Unsplit-Final}). Our precise approximation ratio is presented in Theorem~\ref{res-unsplit}, and some numerical values for specific $k$ are shown in Table~\ref{res-2} (the lines for Section~\ref{Sec-Unsplit-Final}).
Additionally, for unsplittable 5-CVRP, we further refine the analysis to improve the approximation ratio to 2.157 (Section~\ref{Sec-Unsplit-5}).
\end{enumerate}

\begin{table}[ht]
\small
\centering

\resizebox{0.68\textheight}{!}{
\begin{tabular}{ccccccccc}
\hline
    $k$ & 3 & 4 & 5 & 6 & 7 & 8 & 9&10 \\
\hline
    Previous & {1.792} & {2.051} & {2.249} & {2.416} & {2.558} & {2.684} & {2.795}& {2.893} \\
    Results&~\cite{gupta2023local}&\cite{gupta2023local}&\cite{gupta2023local}&\cite{gupta2023local}&\cite{gupta2023local}&\cite{gupta2023local}&\cite{gupta2023local}&\cite{uncvrp}\\
\hline
    Our Results in & \multirow{2}*{$\boldsymbol{1.500}$} & \multirow{2}*{$\boldsymbol{1.750}$} & \multirow{2}*{$\boldsymbol{2.157}$} & \multirow{2}*{-} & \multirow{2}*{-} & \multirow{2}*{-} & \multirow{2}*{-}& \multirow{2}*{-}\\
     Sections~\ref{Sec-Unsplit-34} and~\ref{Sec-Unsplit-5}&\\
\hline
    Our Results & \multirow{2}*{$1.906$} & \multirow{2}*{$1.955$} & \multirow{2}*{$2.178$} & \multirow{2}*{$\boldsymbol{2.163}$} & \multirow{2}*{$2.351$} & \multirow{2}*{$2.383$} & \multirow{2}*{$2.537$}& \multirow{2}*{$2.538$} \\
    in Section~\ref{Sec-Split-Initial} &\\
\hline
    Our Results & \multirow{2}*{$1.906$} & \multirow{2}*{$1.955$} & \multirow{2}*{$2.178$} & \multirow{2}*{$2.163$} & \multirow{2}*{$\boldsymbol{2.343}$} & \multirow{2}*{$\boldsymbol{2.337}$} & \multirow{2}*{$\boldsymbol{2.471}$}& \multirow{2}*{$\boldsymbol{2.448}$} \\
    in Section~\ref{Sec-Unsplit-Final}&\\
\hline
\end{tabular}
}

\medskip

\resizebox{0.68\textheight}{!}{
\begin{tabular}{cccccccc}
\hline
    $k$ & $11720$ & $11721$ & $11722$ & $\dots$ & $\dots$ & $1.7\times 10^7$ & $1.8\times 10^7$ \\
\hline
    Previous & {$3.19256$} & {$3.19269$} & {$3.19256$} & {$\dots$}  & {$\dots$} & {$3.19282$} & {$\boldsymbol{3.19282}$} \\
    Results&~\cite{uncvrp}&\cite{uncvrp}&\cite{uncvrp}&$\dots$&$\dots$&\cite{uncvrp}&\cite{uncvrp}\\
\hline
    Our Results & \multirow{2}*{$3.19255$} & \multirow{2}*{$3.19264$} & \multirow{2}*{$3.19256$} & \multirow{2}*{$\dots$} & \multirow{2}*{$\dots$} & \multirow{2}*{$3.19315$} & \multirow{2}*{$3.19315$}\\
    in Section~\ref{Sec-Split-Initial}&\\
\hline
    Our Results & \multirow{2}*{$\boldsymbol{3.17973}$} & \multirow{2}*{$\boldsymbol{3.17981}$} & \multirow{2}*{$\boldsymbol{3.17973}$} & \multirow{2}*{$\boldsymbol{\dots}$} & \multirow{2}*{$\boldsymbol{\dots}$} & \multirow{2}*{$\boldsymbol{3.19281}$} & \multirow{2}*{$3.19282$}\\
    in Section~\ref{Sec-Unsplit-Final}&\\
\hline
\end{tabular}
}
\caption{
Unsplittable $k$-CVRP: previous and our approximation ratios for different values of $k$ under $\alpha=\frac{3}{2}$, where the best results are marked in bold.
}
\label{res-2}
\end{table}

In Tables~\ref{res-1} and \ref{res-2}, the previous results are $\alpha+1-\frac{\alpha}{k}-\max\{\Omega(\frac{1}{k^3}),\varepsilon\}$
and $\alpha+1+\ln2-\frac{2\alpha}{k}-\varepsilon'$ for splittable and unsplittable $k$-CVRP, respectively, as well as $\min\{H_k-\frac{1}{8k},H_k-\sum_{i=1}^{k}\frac{\log i}{8ki}\}$ for both versions. The constant behind $\Omega$ is available in~\cite{BompadreDO06}.\footnote{See the calculation in Appendix~\ref{bomoadre}.}
The values $\varepsilon$ and $\varepsilon'$ depending on $\alpha$ and $k$ can be calculated using the results in~\cite{blauth2022improving,uncvrp}.
We adopt $\alpha=\frac{3}{2}$. For this case, the values $\varepsilon$ and $\varepsilon'$ are smaller than $\frac{1.005}{3000}$ and $-\ln(1-\frac{1.005}{3000})$, respectively.
Hence, the approximation ratio of splittable $k$-CVRP and unit-demand $k$-CVRP in~\cite{blauth2022improving} is at least $\frac{5}{2}-\frac{1.005}{3000}-\frac{1.5}{k}$
for $k\geq3$.\footnote
{
See the calculation in Appendix~\ref{best-splittable}.
Note that Blauth \emph{et al.}~\cite{blauth2022improving} proved an approximation ratio of $\frac{5}{2}-\frac{1}{3000}$, which works for any $k$. However, using their results, one may achieve slightly better approximation ratios for small $k$. Since we consider $k$ as a fixed integer, we also calculate the best approximation ratios that can be obtained by their results.
}
The approximation ratio of unsplittable $k$-CVRP in~\cite{uncvrp} is at least $\frac{5}{2}+\ln2+\ln(1-\frac{1.005}{3000})-\frac{3}{k}$
for even $k\geq3$.\footnote{See the calculation in Appendix~\ref{best-unsplittable}.
}
Note that for unsplittable $k$-CVRP with odd $k\geq 3$, we need to double the capacity and the demand, resulting in a slightly worse approximation ratio of at least $\frac{5}{2}+\ln2+\ln(1-\frac{1.005}{3000})-\frac{1.5}{k}$. 
The code used to calculate the approximation ratios in Tables~\ref{res-1} and \ref{res-2} is available at \url{https://github.com/JingyangZhao/CVRP}.

\section{Definitions, Assumptions, and Notations}\label{sec_pre}
In a graph, a \emph{walk} is a succession of edges, where an edge may appear more than once. We use a sequence of vertices to denote a walk:
$(v_1,v_2,v_3,\dots, v_l)$ represents a walk with the edges $(v_1,v_2)$, $(v_2,v_3)$, and so on.
A \emph{path} in a graph is a walk such that no vertex appears twice in the sequence, and a \emph{cycle} is a walk such that only the first and the last vertices are the same.
A cycle containing $l$ edges is called an \emph{$l$-cycle} and the \emph{length} of it is $l$.
Two subgraphs (or two sets of edges) are \emph{vertex-disjoint} if they do not have a common vertex.
Given an edge-weighted graph, where the number $n$ of vertices is a multiple of $k$, a \emph{minimum weight $k$-cycle cover} is a set of exactly $\frac{n}{k}$ vertex-disjoint $k$-cycles with the minimum total weight of edges in the $k$-cycles in the set.
A \emph{minimum weight mod-$k$-cycle cover} (resp., \emph{minimum weight mod-$k$-tree cover}) is a set of vertex-disjoint cycles (resp., trees) such that the length of each cycle (resp., the number of vertices on each tree) is divisible by $k$, each vertex of the graph appears in exactly one cycle, and the total weight of edges in the cycles in the set is minimized. 
A \emph{minimum weight cycle cover} is a set of vertex-disjoint cycles such that the length of each cycle is at least three, each vertex of the graph appears in exactly one cycle, and the total weight of edges in the cycles in the set is minimized.

\subsection{Problem Definitions}
We use $G=(V\cup\{v_0\}, E)$ to denote a complete graph, where the vertex $v_0$ represents the depot and vertices in $V$ represent customers. There is a non-negative weight function $w: E\to \mathbb{R}_{\geq0}$ on the edges in $E$, which denotes the distance between two endpoints of the edge.
The weight function $w$ is a metric function, i.e., it satisfies the symmetric and triangle inequality properties.
For any weight function $w: X\to \mathbb{R}_{\geq0}$, we will extend it to subsets of $X$ by defining $w(Y) = \sum_{x\in Y} w(x)$ for any $Y\subseteq X$.
A feasible solution of CVRP is also called an \emph{itinerary}, which is a walk starting and ending at vertex $v_0$.
It can be partitioned into several minimal itineraries containing $v_0$, each of which is called a \emph{tour}.

The Capacitated Vehicle Routing Problem (CVRP) can be described as follows.

\begin{definition}
An instance $(G=(V\cup \{v_0\},E),w,d,k)$ of CVRP consists of:
\begin{itemize}
\item a complete graph $G$, where $V=\{v_1,\dots,v_n\}$ represents the $n$ customers and $v_0$ represents the depot;
\item a metric weight function on edges $w$: $(V\cup\{v_0\})\times(V\cup\{v_0\})\rightarrow\mathbb{R}_{\geq 0}$, which represents the distances;
\item the demand of each customer $d=(d_1, \dots, d_n)$, where $d_i\in\mathbb{Z}_{\geq 1}$ is the demand required by customer $v_i\in V$;
\item the capacity $k\in\mathbb{Z}_{\geq 1}$ of the vehicle that initially stays at the depot $v_0$.
\end{itemize}
A feasible solution is an itinerary  such that
\begin{itemize}
\item each tour delivers at most $k$ of the demand to customers on the tour;
\item the union of tours meets the demand of every customer.
\end{itemize}
The goal is to find such an itinerary $I$, minimizing the total distances of the succession of edges in the walk, i.e., $w(I)\coloneqq \sum_{e\in I}w(e)$.
\end{definition}

According to the property of the demand, we can define three different versions of the problem. If each customer's demand should be delivered in one tour, we call it \emph{unsplittable CVRP}. If the demand of a customer can be split into several tours, we call it \emph{splittable CVRP}. If each customer's demand is 1, we call it \emph{unit-demand CVRP}.

In this paper, we mainly consider that the vehicle capacity $k$ is a fixed integer with $k\geq3$.

\subsection{Some Assumptions}
In our problems, we will also make some assumptions which can be guaranteed by some simple observations or polynomial-time reductions.

\begin{assumption}
In an optimal itinerary, each tour is a cycle.
\end{assumption}
If there exists a tour that is not a cycle, we can obtain a cycle by shortcutting, and the weight is non-increasing by the triangle inequality.

\begin{assumption}\label{ass1-add1}
If the vehicle travels to location $A$ and merely passes by location $B$ without making any delivery there, we do not consider it to have visited $B$. Therefore, in each tour, the vehicle visits only those customers to whom it delivers.
\end{assumption}

A tour is called \emph{trivial} if it only visits one customer and \emph{non-trivial} otherwise.
We show some properties of optimal itineraries to splittable CVRP.

\begin{lemma}[\cite{dror1990split}]~\label{notwo}
For splittable CVRP, there exists an optimal itinerary where no two tours visit two common customers.
\end{lemma}
To our knowledge, this property was first used in~\cite{dror1990split}. We explore some properties.

\begin{lemma}~\label{non-trivial-tour}
For splittable CVRP, there exists an optimal itinerary that contains at most $n-1$ non-trivial tours.
\end{lemma}
\begin{proof}
Given an optimal itinerary $I$, we construct an auxiliary graph $G_{I}$ on $V$, where two vertices in $V$ are adjacent in $G_{I}$ if and only if there is at least one tour in $I$ that visits both of them. By Lemma~\ref{notwo}, we assume that no two tours in $I$ visit two common customers. We further assume that $I$ is an itinerary to ensure that $G_{I}$ have the minimum number of edges.
Next, we assume to the contrary that there are more than $n-1$ non-trivial tours in $I$ and show a contradiction.

For two vertices visited in the same tour $C_i$, there is an edge between them and we color the edge with $i$.
By the assumption that no two tours in $I$ visit two common customers, we know that each edge is colored with exactly one color and the set of edges with the same color form a clique.
Since $I$ has more than $n-1$ non-trivial tours, each non-trivial tour will create a clique of size $\geq 2$ in $G_{I}$ with the same color for the edges in the clique, we know that the graph $G_I$ contains at least one cycle, denoted by $Cycle=(v_1,v_2,\dots, v_l,v_1)$, such that each edge has a different color.

Based on the cycle $\text{Cycle} = (v_1, v_2, \dots, v_l, v_1)$, we assume that tour $C_i$ visits both $v_i$ and $v_{i+1}$ for each $i \in \{1, \dots, l\}$, and that it delivers $x_{i,i}$ and $x_{i,i+1}$ to customers $v_i$ and $v_{i+1}$, respectively, where we define $v_{l+1} = v_1$ and $x_{l,l+1} = x_{l,1}$. We assume w.l.o.g.\ that $x_{1,1} = \min_i \{x_{i,i}\}$.
We construct another itinerary $I'$ by modifying the delivery amounts as follows: for each $i \in \{1, \dots, l\}$, tour $C_i$ delivers $x_{i,i} - x_{1,1}$ to customer $v_i$, and $x_{i,i+1} + x_{1,1}$ to customer $v_{i+1}$.
In $I'$, each customer's demand remains satisfied, and the total delivery on each tour is unchanged. Hence, $I'$ is still an optimal solution.
However, in $I'$, tour $C_1$ no longer needs to visit $v_1$, implying that the edge between $v_1$ and $v_2$ in $G_{I'}$ is removed. This contradicts the assumption that $G_I$ has the minimum number of edges.
\end{proof}

In an itinerary, if the vehicle always delivers an integer amount of demand to each customer in each tour, then we say that the itinerary satisfies the \emph{integer property}.
It is not hard to prove that there is an optimal itinerary satisfying the integer property by using the exchanging argument similar to that in the proof in Lemma~\ref{non-trivial-tour}. So, we also make the following assumption.
\begin{assumption}\label{ass4}
Each tour delivers an integer amount of demand to every customer it visits.
\end{assumption}

\begin{lemma}~\label{lessnk}
For splittable $k$-CVRP, if the demand $d_i$ of a customer $v_i$ is at least $(n-1)(k-1)+1$, then there is an optimal itinerary that contains a trivial tour visiting $v_i$.
\end{lemma}
\begin{proof}
By Lemma~\ref{non-trivial-tour}, we consider an optimal itinerary in which there are at most $n - 1$ non-trivial tours. Note that each such tour can deliver at most $k - 1$ demand to customer $v_i$.
Thus, if the demand $d_i > (n - 1)(k - 1)$, it follows that at least one trivial tour must visit $v_i$.
\end{proof}

By Lemma~\ref{lessnk}, we may iteratively assign trivial tours to serve $v_i$ until its residual demand is at most $(n - 1)(k - 1)$. Therefore, we may assume the following.

\begin{assumption}\label{ass2}
For splittable $k$-CVRP, the demand $d_i$ of each customer $v_i\in V$ is at most $(n-1)(k-1)$.
\end{assumption}

\begin{lemma}\label{reduction}
If $k=n^{O(1)}$, there is a polynomial-time reduction from splittable $k$-CVRP to unit-demand $k$-CVRP.
\end{lemma}
\begin{proof}
For a customer with $d_i$ of the demand in splittable $k$-CVRP, we consider it as $d_i$ unit-demand customers. Thus, we obtain an instance of unit-demand $k$-CVRP.
By Assumption~\ref{ass2}, we can assume that each customer's demand is at most $(n-1)(k-1)$. Hence, the reduction uses polynomial time when $k=n^{O(1)}$.
\end{proof}

We note that this property was also observed in~\cite{JayaprakashS22}. In this paper, we will always consider $k$ as a fixed integer. Thus, splittable $k$-CVRP can be reduced to unit-demand $k$-CVRP in polynomial time.
We will design approximation algorithms for unit-demand $k$-CVRP, which also work for splittable $k$-CVRP.

\begin{assumption}\label{ass3}
For unsplittable $k$-CVRP, the demand $d_i$ of each customer $v_i\in V$ is less than $k$.
\end{assumption}

In an optimal itinerary, if the vehicle always delivers $k$ of the demand in each tour, then we say that the itinerary satisfies the \emph{saturated property}.
\begin{assumption}\label{ass1}
For $k$-CVRP, there exists an optimal saturated itinerary where every tour delivers exactly $k$ of the demand.
\end{assumption}

Assumption~\ref{ass1} can be ensured by adding some dummy customers at the depot.

First, it is easy to see that after adding some unit-demand customers with at the position of $v_0$, the new instance is equivalent to the old one. Moreover, if there exists an optimal saturated itinerary, then after adding $n'$ with $n'\bmod k=0$ unit-demand customers at the position of $v_0$, there still exists an optimal saturated itinerary in the new instance since we can use $\frac{n'}{k}$ tours with zero weight to satisfy the new $n'$ unit-demand customers and each such tour delivers $k$ of the demand.

Consider unit-demand CVRP.
Suppose there exists an optimal itinerary where there are $m_i$ tours delivering $i$ of the demand for each $i\in\{1,2,\dots,k\}$, i.e., $n=\sum_{i=1}^{k}im_i$. 
If we add $\sum_{i=1}^{k}(k-i)m_i$ unit-demand customers at the position of $v_0$, it is easy to see that there exists an optimal saturated itinerary in the new instance.
Moreover, we have $\sum_{i=1}^{k}(k-i)m_i+n=\sum_{i=1}^{k}m_ik$ which is divisible by $k$.
But, we do not know the precise value of $\sum_{i=1}^{k}(k-i)m_i$.\footnote{An alternative proof, noted by an anonymous reviewer, could be that we simply guess the value of $\sum_{i=1}^{k}(k-i)m_i$ since it is polynomially bounded.}
We may further add $n'$ unit-demand customers at the position of $v_0$ such that the total number of added customers is $\sum_{i=1}^{k}(k-i)m_i+n'=k^2n+k-(n\bmod k)$, which is a known number. 
It is easy to see that $n'\bmod k=0$.
Since $m_i\leq\ceil{\frac{n}{i}}\leq n$, we know that $\sum_{i=1}^{k}(k-i)m_i\leq\sum_{i=1}^{k-1}kn\leq k^2n$ and then $n'\geq k-(n\bmod k)\geq 0$.
Therefore, if we add $k^2n+k-(n\bmod k)$ unit-demand customers in total at the position of $v_0$, there will exist an optimal saturated itinerary in the new instance.

For unsplittable CVRP, if we add $\sum_{i=1}^{n}k^2d_i+k-((\sum_{i=1}^{n}d_i)\bmod k)$ unit-demand customers at the position of $v_0$, by a similar argument, there will exist an optimal saturated itinerary in the new instance.

Assumption~\ref{ass1} guarantees the existence of an optimal solution with a good structure. We need this assumption to simplify some arguments and make some presentations neat. For unit-demand $k$-CVRP, such an itinerary consists of a set of $(k+1)$-cycles intersecting only at the depot. 

\subsection{Some Important Notations}
The following notations are illustrated with the unit-demand case. Most of them will be used to establish some lower bounds for our problems.
\begin{itemize}
\item $I^*$: an optimal solution to our problem;
\item $\Delta$: the sum of the weights of the edges from the depot $v_0$ to customer, i.e., $\sum_{v_i\in V}w(v_0, v_i)$;
\item $H^*$: a minimum weight Hamiltonian cycle on $V\cup\{v_0\}$;
\item $H_{CS}$: the Hamiltonian cycle on $V\cup\{v_0\}$ obtained by the Christofides-Serdyukov algorithm~\cite{christofides1976worst,serdyukov1978some};
\item $\M^*$: a minimum weight perfect matching in $G[V]$;
\item $\MST$: the total weight of the edges in a minimum weight spanning tree in $G$;
\item $\mathcal{C}^{*}$: a minimum weight cycle cover in $G[V]$;
\item $\mathcal{C}_{k}^{*}$: a minimum weight $k$-cycle cover in $G[V]$;
\item $\mathcal{C}_{\bmod k}^{*}$: a minimum weight mod-$k$-cycle cover in $G[V]$.
\end{itemize}

For an instance $G=(V\cup\{v_0\},E)$ of splittable or unsplittable CVRP, we construct a corresponding unit-demand instance $G'=(V'\cup\{v_0\},E')$ by replacing each customer with demand $d_i$ with $d_i$ unit-demand customers. 
By Assumption~\ref{ass2}, this reduction can be done in polynomial time when $k=O(1)$. It is easy to observe that the optimal value in $G'$ is at most the optimal value in $G$. 
Therefore, for an instance $G$ of splittable or unsplittable CVRP, the lower bound computed in the unit-demand instance $G'$ also constitutes a valid lower bound for $G$.

Although the above notations are defined for the unit-demand case, they can also be applied to the splittable and unsplittable cases.
For an instance $G=(V\cup\{v_0\},E)$ of splittable or unsplittable CVRP, the above notations, including $\Delta$, $H^{*}$, $H_{CS}$, $\M^*$, $\MST$, $\C^{*}$, $\C_{k}^{*}$, and $\C_{\bmod k}^{*}$, are defined based on the instance $G'$. 

Thus, in the splittable and unsplittable cases, \(\Delta\) becomes the sum of each customer's demand multiplied by the weight of the edge from \(v_0\) to the customer, i.e., \(\sum_{v_i \in V} d_i w(v_0, v_i)\), \(\mathcal{C}^*_k\) represents a minimum weight \(k\)-cycle cover in \(G'[V']\), and so on. Note that, for \(H^*\), \(H_{CS}\), and \(\MST\), there is no difference between \(G\) and \(G'\), so we do not distinguish between them.

We also mention the following.
A minimum weight perfect matching in a complete graph can be found in $O(n^3)$ time using classical algorithms~\cite{gabow1974implementation,lawler1976combinatorial}, although faster and simpler algorithms exist for other graphs; see, for example, Schrijver’s book~\cite{schrijver2003combinatorial}.
A minimum weight spanning tree in a complete graph can be computed in $O(n^2)$ time using Prim’s algorithm~\cite{prim1957shortest}, or in optimal time using the algorithm of Pettie and Ramachandran~\cite{pettie2002optimal}.
It is NP-hard to compute a minimum weight $k$-cycle cover for any $k \geq 3$~\cite{KirkpatrickH78}.
There exists an $O(n^2 \log n)$-time 2-approximation algorithm for the minimum weight mod-$k$-cycle cover problem based on the primal-dual method~\cite{GoemansW95}.
Additionally, a minimum weight cycle cover can be computed via a reduction to the minimum weight perfect matching problem, for which more efficient algorithms are also known~\cite{hartvigsen1984extensions,schrijver2003combinatorial}.
These results will be used in our algorithms.

For an optimal itinerary $I^*$, the edges incident to $v_0$ in $I^{*}$ are called \emph{home-edges}, and the set of home-edges of $I^*$ is denoted by $h(I^*)$.
We define $\ra$ as the proportion of the weights of all home-edges in $I^{*}$, i.e., 
$\ra=\frac{w(h(I^*))}{w(I^*)}$. 
Note that we assume $w(I^*)\neq 0$ to exclude the trivial case.

\section{Lower Bounds on Unit-Demand CVRP}\label{SEC3}
In this section, we study lower bounds related to $\Delta$, $H^{*}$, $H_{CS}$, $\M^*$, $\MST$, $\C_{k}^{*}$, $\C_{\bmod k}^{*}$, and $\C^{*}$.
As previously mentioned, the optimal value for $G'$ is at most the optimal value for $G$ for an instance $G$ of splittable or unsplittable CVRP. 
Therefore, the lower bounds we establish for the unit-demand case in this section will hold for all three versions of CVRP.
We now assume that the problem under consideration is unit-demand CVRP and proceed to prove some lower bounds.

We will use the concept of $\ra$ to derive some refined lower bounds, which were not considered in previous work. The first lower bound, related to the minimum weight Hamiltonian cycle $H^{*}$ on $V\cup\{v_0\}$, has been used in most previous papers.

\begin{lemma}[\cite{HaimovichK85,altinkemer1987heuristics}]\label{lb-tsp}
$w(I^{*}) \geq w(H^{*})$.
\end{lemma}

\begin{lemma}\label{lb-delta}
$\frac{k}{2}w(I^*)\geq(\ra+\frac{k-2}{2}) w(I^{*})\geq\Delta$.
\end{lemma}
\begin{proof}
Since $0\leq \ra\leq 1$, we have
$\frac{k}{2}= (1+\frac{k-2}{2})\geq (\ra+\frac{k-2}{2})$, which proves the first inequality.

Now, we show the second inequality.
By Assumption~\ref{ass1}, $I^*$ consists of a set of $(k+1)$-cycles. We consider an arbitrary $(k+1)$-cycle $C=(v_0,v_1,\dots, v_k,v_0)$ in $I^*$.
Since $k\geq3$, the triangle inequality implies that $w(C)\geq 2 w(v_0,v_{i})$ for each $i\in \{2,3,\dots, k-1\}$.
Thus, we have
$$
\sum_{i=1}^kw(v_0,v_i)=w(h(C))+\sum_{i=2}^{k-1}w(v_0,v_i)\leq w(h(C))+\frac{k-2}{2}w(C).
$$

Summing the above inequality over all cycles in $I^*$, we obtain
\begin{align*}
\Delta=\sum_{v_i\in V} w(v_0,v_i)&\leq\sum_{C\in I^*} \lrA{w(h(C))+ \frac{k-2}{2}w(C)}\\
&=\ra w(I^{*}) + \frac{k-2}{2}w(I^{*})=\lrA{\ra+\frac{k-2}{2}}w(I^*),
\end{align*}
as desired.
\end{proof}

\begin{lemma}\label{lb-tree}
$(1-\frac{1}{2}\ra)w(I^{*})\geq \mbox{MST}$.
\end{lemma}
\begin{proof}
For each tour in $I^*$, there are exactly two home-edges. We can obtain a spanning tree of the graph from $I^*$ by deleting the longer home-edge in each tour. Since $w(h(I^*))=\ra w(I^*)$ by definition, the weight of this spanning tree is at least $w(I^*)-\frac{1}{2}\ra w(I^*)=(1-\frac{1}{2}\ra)w(I^{*})$, which is at least the weight of the minimum weight spanning tree.
\end{proof}

Recall that $H_{CS}$ is obtained by the Christofides-Serdyukov algorithm.
\begin{lemma}[\cite{christofides1976worst,serdyukov1978some}]\label{chris}
$\mbox{MST}+\frac{1}{2}w(H^*) \geq w(H_{CS})$.
\end{lemma}

By Lemmas~\ref{lb-tsp}, \ref{lb-tree}, and \ref{chris}, we have the following result.
\begin{lemma}\label{lb-chris}
$\frac{3-\ra}{2}w(I^*)\geq w(H_{CS})$.
\end{lemma}

\begin{lemma}\label{lb-cyclepacking}\label{lb-constrained-cyclepacking}
$\min\{2(1-\ra),1\} w(I^{*})\geq w(\mathcal{C}_{k}^{*})\geq w(\mathcal{C}_{\bmod k}^{*})\geq w(\mathcal{C}^{*})$.
\end{lemma}
\begin{proof}
We first show that there exists a $k$-cycle cover in $G[V]$ whose weight is at most $\min\{2(1-\ra),1\}w(I^{*})$.

By Assumption~\ref{ass1}, $I^*$ consists of a set of $(k+1)$-cycles. 
Consider an arbitrary $(k+1)$-cycle $C=(v_0,v_1,\dots, v_{k},v_0)$ in $I^*$. Let $C'=(v_1,\dots, v_{k},v_1)$ be the $k$-cycle obtained by shortcutting $v_0$ from $C$.
By the triangle inequality, we have $w(v_1,v_k)\leq \min\{w(v_0,v_1)+w(v_0,v_k),w(C)-w(v_0,v_1)-w(v_0,v_k)\}$. Thus, we have $w(C')=w(C)-w(v_0,v_1)-w(v_0,v_k)+w(v_1,v_k)\leq \min\{w(C),2w(C)-2w(v_0,v_1)-2w(v_0,v_k)\}$.
Note that the edges $(v_0,v_1)$ and $(v_0,v_k)$ are home-edges of the cycle $C$.
By summing the above inequality over all cycles in $I^*$, we obtain the desired result.

Since $\C_{k}^*$ is the minimum weight $k$-cycle cover in $G[V]$, we have
\[
w(\C_{k}^*)\leq\min\{2(1-\ra),1\} w(I^{*}).
\]

Since any $k$-cycle cover is a mod-$k$-cycle cover, and any mod-$k$-cycle cover is a cycle cover, we have the relationship $w(\C^{*})\leq w(\C_{\bmod k}^{*})\leq w(\C_{k}^{*})$.
\end{proof}

Since $H_{CS}$ is a $\frac{3}{2}$-approximate Hamiltonian cycle~\cite{christofides1976worst,serdyukov1978some}, some papers~\cite{altinkemer1987heuristics,HaimovichK85} used the inequalities $\Delta\leq \frac{k}{2}w(I^*)$ and $w(H_{CS})\leq \frac{3}{2}w(I^*)$ in the analysis of ITP and UITP. However, if the upper bound of $\Delta$ is tight, i.e., $\Delta=\frac{k}{2}w(I^*)$, we have $\ra=1$ by Lemma~\ref{lb-delta}. In this case, by Lemma~\ref{lb-chris}, we also have $w(H_{CS})\leq w(I^*)$. Therefore, our new lower bounds show that these two bounds cannot be tight simultaneously, indicating the potential for better approximation ratios. Besides, if $\ra=1$, by Lemma~\ref{lb-constrained-cyclepacking}, we have $w(\C^*_{k})=w(\mathcal{C}_{\bmod k}^{*})=w(\C^*)=0$. So, any $O(1)$-approximate cycle covers may outperform the optimal Hamiltonian cycle. 

Before showing how these insights can lead to improved approximation ratios, we first review the ITP algorithms~\cite{altinkemer1987heuristics,HaimovichK85} that operate with a Hamiltonian cycle, and then we will propose our EX-ITP algorithm, which works for any cycle cover. 

\section{The ITP Algorithms}
ITP (Iterated Tour Partitioning) is a frequently used technique for (unit-demand) CVRP. The main idea of ITP is to construct feasible solutions for CVRP based on given Hamiltonian cycles: first split the Hamiltonian cycle into several connected pieces of length at most $k$ and then construct a tour for each piece. The algorithm will consider several different ways to split the Hamiltonian cycle and choose the best one.

According to the used Hamiltonian cycle containing the depot $v_0$ or not, there are two versions of ITP.
We briefly introduce them below and then introduce an extension. Now, we consider unit-demand $k$-CVRP.

\subsection{The AG-ITP Algorithm}
We first review the ITP algorithm introduced by Altinkemer and Gavish~\cite{altinkemer1987heuristics}, where the Hamiltonian cycle needs to go through the depot $v_0$.
Assume that there is a Hamiltonian cycle $H=(v_0,v_1,\dots, v_n,v_0)$ on $V\cup\{v_0\}$ as part of the input. AG-ITP will select the best solution from $k$ possible solutions.
For each $1\leq i\leq k$, the $i$-th solution consists of the following $(\ceil{\frac{n-i}{k}}+1)$ tours: $(v_0,v_1,\dots, v_i, v_0)$, $(v_0,v_{i+1},\dots, v_{i+k},v_0)$, $(v_0,v_{i+k+1},\dots, v_{i+2k},v_0)$, $\dots$, $(v_0,v_{i+(\ceil{\frac{n-i}{k}}-1)k+1},\dots, v_n,v_0)$. Except for the first and last tours, each tour consists of $k$ customers. Note that AG-ITP can be carried out in $O(nk)$ time.

\begin{lemma}[\cite{altinkemer1987heuristics}]~\label{AGITP}
Given a Hamiltonian cycle $H$ on $V\cup\{v_0\}$ as part of the input, for unit-demand $k$-CVRP with any $k\geq 3$, AG-ITP in $O(nk)$ time outputs a solution with weight at most $\frac{2}{k}\Delta+(1-\frac{1}{k})w(H)$.
\end{lemma}

Using an $\alpha$-approximate Hamiltonian cycle $H$ on $V\cup\{v_0\}$, by Lemmas~\ref{lb-tsp}, \ref{lb-delta} and \ref{AGITP}, AG-ITP computes a solution with weight at most
$
\frac{2}{k}\Delta+(1-\frac{1}{k})w(H)\leq w(I^*)+(1-\frac{1}{k})\alpha w(I^*)=(\alpha+1-\frac{\alpha}{k})w(I^*).
$
Therefore, AG-ITP achieves an approximation ratio of $\alpha+1-\frac{\alpha}{k}$.

\subsection{The HR-ITP Algorithm}
Now, we review the ITP algorithm introduced by Haimovich and Rinnooy Kan~\cite{HaimovichK85}, where the used Hamiltonian cycle $H$ does not go through the depot $v_0$. Assume that $H=(v_1,v_2,\dots, v_n,v_1)$. HR-ITP will select the best solution from $n$ possible solutions.
For each $1\leq i\leq n$, the $i$-th solution consists of the following $\ceil{\frac{n}{k}}$ tours: $(v_0,v_i,\dots, v_{i+k-1}, v_0)$, $(v_0,v_{i+k},\dots, v_{i+2k-1},v_0)$, $(v_0,v_{i+2k},\dots, v_{i+3k-1},v_0)$, $\dots$, $(v_0,v_{i+(\ceil{\frac{n}{k}}-1)k},\dots, v_{i+n-1},v_0)$ where we let $v_{n+i'}=v_{i'}$ for each $1\leq i'\leq n$. Except possibly for the last tour, each tour contains exactly $k$ customers.
The running time is $O(n^2)$.

\begin{lemma}[\cite{HaimovichK85}]~\label{HRITP}
Given a Hamiltonian cycle $H$ on $V$ as part of the input, for unit-demand $k$-CVRP with any $k\geq 3$, HR-ITP in $O(n^2)$ time outputs a solution with weight at most $\frac{2\ceil{n/k}}{n}\Delta+(1-\frac{\ceil{n/k}}{n})w(H)$.
\end{lemma}

When $n$ is divisible by $k$, HR-ITP also computes a solution with weight at most $\frac{2}{k}\Delta+(1-\frac{1}{k})w(H)$. Otherwise, the weight may exceed $\frac{2}{k}\Delta+(1-\frac{1}{k})w(H)$.

\subsection{An Extension of ITP}
The above two ITP algorithms are based on given Hamiltonian cycles. In fact, the requirement of Hamiltonian cycles is not necessary. We can replace the Hamiltonian cycle in the algorithms with a cycle cover to construct feasible solutions in a similar way.

Given a cycle cover $\C$ in the graph $G[V]$ or $G[V\cup\{v_0\}]$ (either containing the depot $v_0$ or not), for each cycle $C\in\C$, we call HR-ITP on it if $C$ does not contain the depot $v_0$, and call AG-ITP on it if $C$ contains the depot $v_0$. By putting all together, we can obtain a feasible solution. The quality of the solution is related to the cycle cover. We refer to this algorithm as the EX-ITP algorithm. Although EX-ITP is still simple, it will play an important role in our algorithms.

\begin{lemma}~\label{EXITP}
Given a cycle cover $\C$ in the graph $G[V]$ or $G[V\cup\{v_0\}]$ as part of the input, for unit-demand $k$-CVRP with any $k\geq 3$, EX-ITP in $O(n^2)$ time outputs a feasible solution with weight at most $2g\Delta+(1-g)w(\C)$, where $g=\max_{C\in\C}\frac{\ceil{\size{C}/k}}{\size{C}}$.
\end{lemma}
\begin{proof}
Define $\Delta_C=\sum_{v_i\in C}w(v_0,v_i)$. Since $\mathcal{C}$ is a cycle cover, it contains at most one cycle that includes the depot. 
For the possible cycle $C\in\C$ with $v_0\in C$, AG-ITP can compute an itinerary on $C$ with weight at most $\frac{2}{k}\Delta_C+(1-\frac{1}{k})w(C)$.
For each cycle $C\in\C$ with $v_0\notin C$, HR-ITP can compute an itinerary on $C$ with weight at most $\frac{2\ceil{\size{C}/k}}{\size{C}}\Delta_C+(1-\frac{\ceil{\size{C}/k}}{\size{C}})w(C)$.

By the triangle inequality, we have $w(C)\leq 2\Delta_C$. 
Then, for any $0\leq x\leq y$, we have $2x\Delta_C+(1-x)w(C)\leq 2y\Delta_C+(1-y)w(C)$. 
Since $\frac{1}{k}\leq \frac{\ceil{\size{C}/k}}{\size{C}}\leq g$ by the definition of $g$, we have 
\begin{align*}
\frac{2}{k}\Delta_C+\lrA{1-\frac{1}{k}}w(C)&\leq \frac{2\ceil{\size{C}/k}}{\size{C}}\Delta_C+\lrA{1-\frac{\ceil{\size{C}/k}}{\size{C}}}w(C)\\
&\leq 2g\Delta_C+(1-g)w(C).
\end{align*}
Therefore, the itinerary on $C\in \C$ has a weight of at most $2g\Delta_C+(1-g)w(C)$.

Since $\sum_{C\in\C}\Delta_C=\Delta$, EX-ITP can output a solution with weight at most 
\[
\sum_{C\in\C}(2g\Delta_C+(1-g)w(C))=2g\Delta+(1-g)w(\C),
\]
as desired.
\end{proof}

Intuitively, the value $\frac{1}{g}$ in Lemma~\ref{EXITP} represents the maximum, over all cycles (excluding the depot), of the average number of customers per tour used to serve that cycle. Thus, a larger value of $g$ indicates that each tour serves fewer customers on average. 

In our algorithm, we frequently use a special case that the input cycle cover is a mod-$k$-cycle cover $\C_{\bmod k}$ in $G[V]$. We have a good approximation ratio for this special case and a good algorithm to find a mod-$k$-cycle cover.
By Lemma~\ref{EXITP} and the previous analysis, we have the following reuslt.

\begin{corollary}\label{ITPpac}
Given a mod-$k$-cycle cover $\C_{\bmod k}$ in $G[V]$ as part of the input, for unit-demand $k$-CVRP with any $k\geq 3$, EX-ITP in $O(n^2)$ time outputs a feasible solution with weight at most $\frac{2}{k}\Delta+(1-\frac{1}{k})w(\C_{\bmod k})$.
\end{corollary}

Moreover, if the cycle cover $\C$ used in Lemma~\ref{EXITP} satisfies $w(\C)\leq \min\{2(1-\ra),1\} w(I^{*})$, as in Lemma~\ref{lb-cyclepacking}, and if $g\leq \frac{1}{2}$, we have the following corollary.

\begin{corollary}\label{EXITPp}
Given a cycle cover $\mathcal{C}$ in $G[V]$ or $G[V\cup\{v_0\}]$ as part of the input, where $w(\C)\leq \min\{2(1-\ra),1\} w(I^{*})$ and $g\leq \frac{1}{2}$,  for unit-demand $k$-CVRP with any $k\geq 3$, EX-ITP achieves an approximation ratio of $g(k-2)+1$. Moreover, in the worst case, we have $\chi=\frac{1}{2}$. (If $g=\frac{1}{2}$, the approximation ratio remains the same for any $\frac{1}{2}\leq \chi\leq 1$.)
\end{corollary}
\begin{proof}
By Lemma~\ref{EXITP}, EX-ITP computes a solution with weight at most 
\begin{align*}
&2g\Delta+(1-g)w(\C)\\
&\leq 2g\lrA{\chi+\frac{k-2}{2}}\cdot\OPT+(1-g)\min\{2(1-\ra),1\}\cdot \OPT\\
&\leq 2g\lrA{\chi+\frac{k-2}{2}}\cdot\OPT+(1-2g)\cdot \OPT+g\cdot2(1-\ra)\cdot\OPT\\
&=2g\cdot\frac{k-2}{2}\cdot\OPT+\OPT=(g(k-2)+1)\cdot\OPT,
\end{align*}
where the first inequality follows from $w(\C)\leq \min\{2(1-\ra),1\} w(I^{*})$ by definition and Lemma~\ref{lb-delta}, and the second from $g\leq\frac{1}{2}$. It can be verified that we have $\chi=\frac{1}{2}$ in the worst case. Moreover, if $g=\frac{1}{2}$, the approximation ratio remains the same for any $\frac{1}{2}\leq \chi\leq 1$.
\end{proof}

\section{Applications of the EX-ITP Algorithm}\label{SEC5}
In this section, we demonstrate that EX-ITP can be used to design improved approximation algorithms for splittable $k$-CVRP and unit-demand $k$-CVRP. We show as examples, the approximation ratio can be significantly improved from $1.792$~\cite{gupta2023local} to $\frac{3}{2}=1.500$ for 3-CVRP, and from $1.750$~\cite{4cvrp} to $\frac{5}{3}<1.667$ for 4-CVRP. At last, we show that the approximation ratio of 4-CVRP can be further improved to $\frac{3}{2}=1.500$ based on a unique structural property.

\subsection{Splittable 3-CVRP and Unit-demand 3-CVRP}\label{split-3-cvrp}
There are little improvements for 3-CVRP in the last 20 years. Only recently an improved result for the general weighted $k$-set cover problem~\cite{gupta2023local} leads to the current best result for 3-CVRP. We show that we can easily improve the approximation ratio by using EX-ITP. Our algorithm computes a minimum weight cycle cover $\mathcal{C}^*$ in $G[V]$, and then calls EX-ITP on $\mathcal{C}^*$.

\begin{theorem}\label{2_second_approach}
For splittable 3-CVRP and unit-demand 3-CVRP, there is a $\frac{3}{2}$-approximation algorithm.
\end{theorem}
\begin{proof}
Our algorithm first computes a minimum weight cycle cover $\C^*$ in $G[V]$ in polynomial time, and then calls EX-ITP on $\C^*$.

Next, we analyze the quality of the solution. Since $\size{C}\geq3$ for each cycle $C\in\C^*$, we have
\[
g=\max_{C\in\C^*}\frac{\ceil{\size{C}/3}}{\size{C}}\leq\max_{\size{C}\geq3}\frac{\ceil{\size{C}/3}}{\size{C}}=\frac{\ceil{4/3}}{4}=\frac{1}{2}.
\]

Note that $w(\C^*)\leq \min\{2(1-\ra),1\} w(I^{*})$ by Lemma~\ref{lb-cyclepacking}. Then, by Corollary~\ref{EXITPp}, EX-ITP achieves an approximation ratio of at most $g(k-2)+1\leq\frac{3}{2}$, where we have $\ra\geq\frac{1}{2}$ in the worst case.
\end{proof}

\subsection{Splittable 4-CVRP and Unit-demand 4-CVRP}\label{split-4-cvrp}
The idea of the algorithm is similar.
However, we will first construct a good mod-2-cycle cover $\C_{\bmod2}$ on $V$, instead of using a minimum weight cycle cover on $V$, and then call EX-ITP.

We compute the mod-2-cycle cover $\C_{\bmod2}$ in this way: first find a minimum weight perfect matching $\M^*$ in the graph $G[V]$, then find a minimum weight perfect matching $\M^{**}$ in the graph $G[V]\setminus \M^*$, and then let
$\C_{\bmod2}=\M^*\cup\M^{**}$.
It is easy to see that $\C_{\bmod2}$ is a
mod-2-cycle cover without 2-cycles since $\M^*$
and $\M^{**}$ are two edge-disjoint perfect matchings of the graph. Recall that we use $\C_{4}^{*}$ to denote a minimum weight 4-cycle cover in $G[V]$. We have the following lemma.

\begin{lemma}\label{lb-mod2}
$w(\C_{\bmod2})\leq w(\C_{4}^{*})$.
\end{lemma}
\begin{proof}
We first prove the following property.

\begin{claim}
Given a minimum weight 4-cycle cover $\C^*_4$ and a minimum weight perfect matching $\M^*$, there is a way to color edges in $\C^*_4$ red and blue such that
\begin{enumerate}
\item[(1)]  the blues (resp., red) edges form a perfect matching $\M_b$ (resp., $\M_r$);
\item[(2)] $\C^*_4=\M_b\cup\M_r$;
\item[(3)] $\M_b\cup\M^*$ is a mod-2-cycle cover without 2-cycles.
\end{enumerate}
\end{claim}
\begin{proof}[Claim Proof]
For each 4-cycle $C=(v_1,v_2,v_3,v_4,v_1)$ in $\C^*_4$, we color the four edges in it by considering the following three cases.

\textbf{Case~1: $\size{C\cap\M^*}=0$.} We color the edges $(v_1,v_2)$ and $(v_3,v_4)$ blue, and the edges $(v_1,v_4)$ and $(v_2,v_3)$ red. Now, the blue edges (resp., red edges) are vertex-disjoint. Alternatively, we could color the edges $(v_1,v_2)$ and $(v_3,v_4)$ red, and the edges $(v_1,v_4)$ and $(v_2,v_3)$ blue.

\textbf{Case~2: $\size{C\cap\M^*}=1$.} Assume w.l.o.g.\ that $C\cap\M^*=\{(v_1,v_2)\}$. Then, we color the edges $(v_1,v_2)$ and $(v_3,v_4)$ red, and the edges $(v_1,v_4)$ and $(v_2,v_3)$ blue.

\textbf{Case~3: $\size{C\cap\M^*}=2$.} Assume w.l.o.g.\ that $C\cap\M^*=\{(v_1,v_2),(v_3,v_4)\}$. Then, we color the edges $(v_1,v_2)$ and $(v_3,v_4)$ red, and the edges $(v_1,v_4)$ and $(v_2,v_3)$ blue.

An illustration of the above three cases can be found in Figure~\ref{example0}.

\begin{figure}[!t]
    \centering
    \begin{subfigure}[t]{0.66\textwidth}  
        \centering
        \includegraphics[scale=0.67]{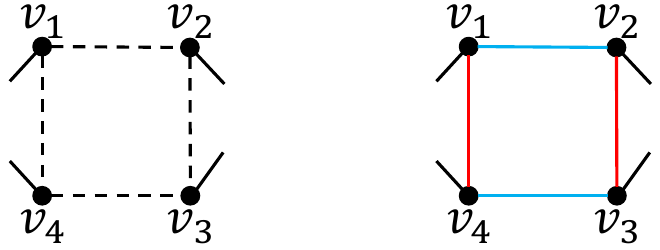}
        \caption{The first case: ${C\cap\M^*}=\emptyset$.}
    \end{subfigure}
    
    \vspace{0.2cm}  
    
    \begin{subfigure}[t]{0.66\textwidth}  
        \centering
        \includegraphics[scale=0.67]{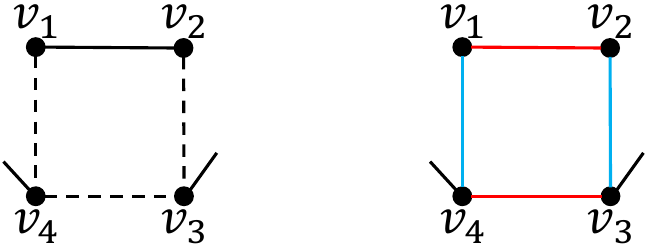}
        \caption{The second case: ${C\cap\M^*}=\{(v_1,v_2)\}$.}
    \end{subfigure}

    \vspace{0.2cm}  
    
    \begin{subfigure}[t]{0.66\textwidth}  
        \centering
        \includegraphics[scale=0.67]{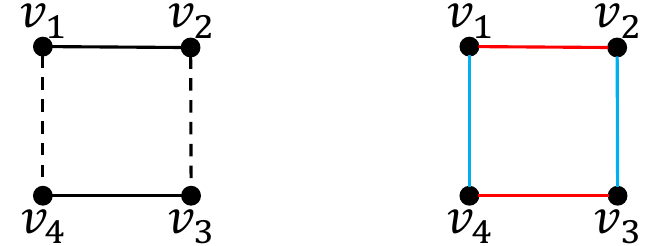}
        \caption{The third case: ${C\cap\M^*}=\{(v_1,v_2),(v_3,v_4)\}$.}
    \end{subfigure}
    \caption{An illustration of the three cases, where the edges in $\M^*$ are represented by black edges before coloring.}
    \label{example0}
\end{figure}

It is easy to see that the set of blues edges $\M_b$ and the set of red edges $\M_r$ are two perfect matchings in $G[V]$, and $\M_b\cup\M_r=\C^*_4$. Moreover, we know that $\M_b\cap\M^*=\emptyset$, and hence $\M_b\cup\M^*$ is a mod-2-cycle cover without 2-cycles. Thus, the claim holds.
\end{proof}

By the claim, we obtain
\[
w(\M_b\cup\M^*)\leq w(\M_b\cup\M_r)=w(\C_{4}^{*}).
\]

Since the mod-2-cycle cover $\C_{\bmod2}=\M^*\cup\M^{**}$ is a minimum weight mod-2-cycle cover containing the perfect matching $\M^*$, we have
\[
w(\C_{\bmod2})=w(\M^*\cup\M^{**})\leq w(\M_b\cup\M^*).
\]
Thus, we have $w(\C_{\bmod2})\leq w(\C_{4}^{*})$, and then the lemma holds.
\end{proof}

\begin{theorem}
For splittable 4-CVRP and unit-demand 4-CVRP, there is a $\frac{5}{3}$-approximation algorithm.
\end{theorem}
\begin{proof}
We call EX-ITP on the mod-2-cycle cover $\C_{\bmod2}$ in Lemma~\ref{lb-mod2}. Since $\size{C}$ is even and $\size{C}\geq4$ for each cycle $C\in\C_{\bmod2}$, we have
\[
g=\max_{C\in\C_{\bmod2}}\frac{\ceil{\size{C}/4}}{\size{C}}\leq\max_{\size{C}\in\{4,6,\dots\}}\frac{\ceil{\size{C}/4}}{\size{C}}=\frac{\ceil{6/4}}{6}=\frac{1}{3}.
\]

Note that $w(\C_{\bmod2})\leq w(\C^*_4)\leq \min\{2(1-\ra),1\} w(I^{*})$ by Lemmas~\ref{lb-mod2} and \ref{lb-cyclepacking}. Then, by Corollary~\ref{EXITPp}, EX-ITP achieves an approximation ratio of at most $g(k-2)+1\leq\frac{5}{3}$, where we have $\ra=\frac{1}{2}$ in the worst case.
\end{proof}

For even $k>4$, we have not been able to find similar good cycle covers based on a minimum weight perfect matching. 

\subsection{A Further Improvement on Splittable 4-CVRP and Unit-demand 4-CVRP}
Note that we can obtain a solution with weight at most $\frac{1}{2}\Delta+\frac{3}{4}w(\C_{4}^{*})$ for unit-demand 4-CVRP by calling EX-ITP on a minimum weight 4-cycle cover $\C^*_4$. However, it is NP-hard to compute $\C^*_4$ as mentioned before. Surprisingly, in this subsection, we show how to obtain a solution with the same upper bound in polynomial time without using $\C^*_4$, and hence obtain an improved $\frac{3}{2}$-approximation algorithm, which even matches the approximation ratio for unit-demand 3-CVRP.

The claim in Lemma~\ref{lb-mod2} can be strengthened based on a similar odd cycle elimination technique in \cite{DBLP:journals/ipl/HassinR97}.

\begin{lemma}\label{lb-mod4}
Given a minimum weight 4-cycle cover $\C^*_4$ and a minimum weight perfect matching $\M^*$, there is a way to color edges in $\C^*_4$ red and blue such that
\begin{enumerate}
\item[(1)]  the blues (resp., red) edges form a perfect matching $\M_b$ (resp., $\M_r$);
\item[(2)] $\C^*_4=\M_b\cup\M_r$;
\item[(3)] $\M_b\cup\M^*$ is a mod-4-cycle cover.
\end{enumerate}
\end{lemma}
\begin{proof}
Previously, we have shown that (1) and (2) holds, and $\M_b\cup\M^*$ is a mod-2-cycle cover without 2-cycles.

Then, we modify the cycle cover $\M_b\cup\M^*$ so that (1), (2), and (3) hold. By the previous coloring, we know that each cycle of $\C^*_4$ contains two blue edges and two red edges. Two blue/red edges are called \emph{matched} if they fall on the same cycle of $\C^*_4$. 

Consider a cycle $C_i\in\M_b\cup\M^*$ with a minimum length not divisible by 4. Since $\M_b\cup\M^*$ is a cycle cover without 2-cycles and the length of every cycle is divisible by 2, the number of blue edges on $C_i$ is odd. Hence, there must be a blue edge $e_1$ such that its matched blue edge $e_3$ falls on a different cycle $C_j\in\M_b\cup\M^*$. Moreover, there are two red edges $e_2$ and $e_3$ sharing common vertices with them. See Figure~\ref{fig00} for an illustration.

\begin{figure}[ht]
\centering
\begin{tikzpicture}
\filldraw [black]
(-1/2,1/2) circle [radius=2pt]
(-1/2-0.866,1/2+0.5) circle [radius=2pt]
(-1/2,-1/2) circle [radius=2pt]
(-1/2-0.866,-1/2-0.5) circle [radius=2pt]

(1/2,1/2) circle [radius=2pt]
(1/2+0.866,1/2+0.5) circle [radius=2pt]
(1/2,-1/2) circle [radius=2pt]
(1/2+0.866,-1/2-0.5) circle [radius=2pt];

\node (left) at (-0.5-0.25,0) {\small $e_1$};
\node (right) at (0.5+0.25,0) {\small $e_3$};

\node at (-2,0) {\small $C_i$};
\node at (2,0) {\small $C_j$};

\node (up) at (0,0.25+0.5) {\small $e_2$};
\node (down) at (0,-0.25-0.5) {\small $e_4$};

\draw[very thick,blue] (-1/2,1/2) to (-1/2,-1/2);
\draw[very thick] (-1/2,1/2) to (-1/2-0.866,1/2+0.5);
\draw[very thick] (-1/2,-1/2) to (-1/2-0.866,-1/2-0.5);
\draw[very thick,red] (-1/2,1/2) to (1/2,1/2);
\draw[very thick,red] (-1/2,-1/2) to (1/2,-1/2);
\draw[very thick,blue] (1/2,1/2) to (1/2,-1/2);
\draw[very thick] (1/2,1/2) to (1/2+0.866,1/2+0.5);
\draw[very thick] (1/2,-1/2) to (1/2+0.866,-1/2-0.5);

\draw[very thick,dotted] (-1/2-0.866,1/2+0.5) ..controls (-3,1) and (-3,-1).. (-1/2-0.866,-1/2-0.5);
\draw[very thick,dotted] (1/2+0.866,1/2+0.5) ..controls (3,1) and (3,-1).. (1/2+0.866,-1/2-0.5);
\end{tikzpicture}
\caption{An illustration of the cycles $C_i,C_j\in\M_b\cup\M^*$, the blue edges $\{e_1,e_3\}$, and the red edges $\{e_2,e_4\}$, where $e_1\in C_i$, $e_3\in C_j$, and each blue edge is adjacent to exactly two (black) edges in $\M^*$.}
\label{fig00}
\end{figure}

We can modify the colors of $\{e_1,e_2,e_3,e_4\}$ such that the new color of $\{e_1,e_3\}$ is red, and the new color of $\{e_2,e_4\}$ is blue. We can see that (1) and (2) still holds. Moreover, by repeating this, we obtain a cycle cover $\M_b\cup\M^*$ such that the length of every cycle is divisible by 4. So, (3) also holds.
\end{proof}

Consider a multi-graph $G'=G[V]/\M^*$, i.e., $G'$ is obtained by contracting edges of $\M^*$ in $G[V]$. There are $\frac{n}{2}$ super-vertices in $G'$ and each of them corresponds to an edge of $\M^*$. For any two super-vertices $e_i,e_j\in\M^*$, there are four parallel edges between them. Assume that $e_i=(u,u')$ and $e_j=(v,v')$. The augmented weights $\widetilde{w}$ on the four edges are 
\[
w(v_0,u')+w(e_i)+w(u,v)+w(e_j)+w(v',v_0), w(v_0,u')+w(e_i)+w(u,v')+w(e_j)+w(v,v_0),
\]
\[
w(v_0,u)+w(e_i)+w(u',v)+w(e_j)+w(v',v_0), w(v_0,u)+w(e_i)+w(u',v')+w(e_j)+w(v,v_0),
\]
which measure the weights of THE tours: $(v_0,u',u,v,v',v_0)$, $(v_0,u',u,v',v,v_0)$, $(v_0,u,u',v,v',v_0)$, $(v_0,u,u',v',v,v_0)$, respectively.
Then, a perfect matching $\M$ in $G'$ corresponds to a solution of unit-demand 4-CVRP with an augmented weight of $\widetilde{w}(\M)$. The minimum augmented weight perfect matching in $G'$, denoted by $\M^{**}$, has the following property.

\begin{lemma}\label{4cvrp+}
$\widetilde{w}(\M^{**})\leq \frac{1}{2}\Delta+\frac{3}{4}w(\C^*_4)$.
\end{lemma}
\begin{proof}
By Lemma~\ref{lb-mod4}, $\M^*\cup\M_b$ is a mod-4-cycle cover. So, the number of blue edges on each cycle $C\in\M^*\cup\M_b$ is even. All blue edges in $\M_b$ can be decomposed into two perfect matchings $\M_1$ and $\M_2$ in $G[V]/\M^*$. Note that 
$\widetilde{w}(\M_1)+\widetilde{w}(\M_2)= \Delta+2w(\M^*)+w(\M_b)$. We have $\widetilde{w}(\M^{**})\leq \frac{1}{2}(\widetilde{w}(\M_1)+\widetilde{w}(\M_2))= \frac{1}{2}\Delta+w(\M^*)+\frac{1}{2}w(\M_b)$. Recall that $\C^*_4=\M_b\cup\M_r$ and $w(\M^*)\leq \min\{w(\M_b),w(\M_r)\}\leq\frac{1}{2}(w(\M_b)+w(\M_r))=\frac{1}{2}w(\C^*_4)$. We have $w(\M^*)+\frac{1}{2}w(\M_b)\leq \frac{1}{2}w(\M^*)+\frac{1}{2}w(\M_r)+\frac{1}{2}w(\M_b)\leq\frac{3}{4}w(\C^*_4)$.
\end{proof}

Since the minimum augmented weight perfect matching $\M^{**}$ in $G'$ can be found by the minimum weight perfect matching algorithm in $O(n^3)$ time~\cite{gabow1974implementation,lawler1976combinatorial,schrijver2003combinatorial},  we can find a solution of unit-demand 4-CVRP with weight $\widetilde{w}(\M^{**})$ in polynomial time. By Lemma~\ref{4cvrp+}, we have the following lemma.

\begin{lemma}\label{4cvrp++}
For unit-demand 4-CVRP, there is a polynomial-time algorithm that computes a solution with weight at most $\frac{1}{2}\Delta+\frac{3}{4}w(\C_{4}^{*})$.
\end{lemma}

By Lemmas~\ref{lb-delta} and~\ref{lb-cyclepacking}, we have
\[
\frac{1}{2}\Delta+\frac{3}{4}w(\C_{4}^{*})\leq\frac{1}{2}(\ra+1)w(I^*)+\frac{3}{4}\min\{(2-2\ra),1\}w(I^*)\leq\frac{3}{2}w(I^*),
\]
where we have $\ra=\frac{1}{2}$ in the worst case. 
Thus, we have the following result.
\begin{theorem}
For splittable 4-CVRP and unit-demand 4-CVRP, there is a $\frac{3}{2}$-approximation algorithm.
\end{theorem}

We remark that the idea of the previous $\frac{5}{3}$-approximation algorithm (based on a good mod-2-cycle cover) for unit-demand 4-CVRP can also be applied to unsplittable 4-CVRP (as we will show in Section~\ref{Sec-Unsplit-34}). However, the $\frac{3}{2}$-approximation algorithm (based on a good 4-path cover instead) in this subsection may not be easy to extend because by shortcutting some vertices on a cycle cover we may still obtain a cycle cover in the remaining graph but the structure may change greatly for a path cover.

\section{An Improvement for Splittable $k$-CVRP}\label{sec-split-initial}
In this section, we show how to use ITP and EX-ITP with the new lower bounds in Section~\ref{SEC3} to improve the approximation ratio of splittable $k$-CVRP with larger $k$.
Given an $\alpha$-approximation algorithm for metric TSP, where $1\leq\alpha\leq\frac{3}{2}$, we obtain an approximation ratio of $\alpha+1-\frac{\alpha}{k}-\Theta(\frac{1}{k})$ for splittable $k$-CVRP, which improves the previous approximation ratio of $\alpha+1-\frac{\alpha}{k}-\max\{\Omega(\frac{1}{k^3}),\varepsilon\}$~\cite{BompadreDO06,blauth2022improving} for some small $k$.

To obtain the claimed approximation ratio, we need to make a trade-off among three algorithms. 
Recall that if the upper bound of $\Delta$ in Lemma~\ref{lb-delta} is tight, we have $\ra=1$, and then we have $w(H_{CS})\leq w(I^*)$ and $w(\C_{\bmod k}^{*})=0$ by Lemmas~\ref{lb-chris} and \ref{lb-cyclepacking}. Thus, apart from using an $\alpha$-approximate Hamiltonian cycle $H$, we may also consider using the Hamiltonian cycle $H_{CS}$~\cite{christofides1976worst,serdyukov1978some} and a modification of a 2-approximate mod-$k$-cycle cover~\cite{GoemansW95}. We will apply the mod-$k$-cycle cover to call EX-ITP and the two Hamiltonian cycles to call AG-ITP, and then make a trade-off among them. 

Let $\rho_1(\ra)$, $\rho_2(\ra)$, and $\rho_3(\ra)$ denote the approximation ratio of these three algorithms, respectively. Note that they are functions on the parameter $\ra$. Then, our final approximation ratio is $\max_{0\leq \ra\leq 1}\min\lrc{\rho_1(\ra),\rho_2(\ra),\rho_3(\ra)}$.

Our first algorithm is to call EX-ITP on a mod-$k$-cycle cover $\C_{\bmod k}$ on $V$. Note that the 2-approximate mod-$k$-cycle cover in~\cite{GoemansW95} is obtained by doubling a mod-$k$-tree cover and then shortcutting. We may find a minimum weight perfect matching on the odd-degree vertices in the mod-$k$-tree cover and then shortcutting to obtain mod-$k$-cycle cover like the Christofides-Serdyukov algorithm for metric TSP~\cite{christofides1976worst,serdyukov1978some}. We use three steps to compute $\C_{\bmod k}$.

\medskip
\noindent\textbf{Step~1.} Compute a mod-$k$-tree cover $\T_{\bmod k}$ in $G[V]$ using the primal-dual algorithm in~\cite{GoemansW95}.

\noindent\textbf{Step~2.} Find a minimum weight perfect matching $\M$ on the odd-degree vertices of $\T_{\bmod k}$.

\noindent\textbf{Step~3.} Construct a mod-$k$-cycle cover $\C_{\bmod k}$ in $G[V]$ from $\T_{\bmod k}\cup\M$ by shortcutting.
\medskip

The second algorithm is to call AG-ITP on an $\alpha$-approximate Hamiltonian cycle on $V\cup\{v_0\}$. The third algorithm is to call AG-ITP on the Hamiltonian cycle $H_{CS}$ on $V\cup\{v_0\}$.

We first analyze some properties of the first algorithm. The mod-$k$-tree cover computed in the first step
has the following property.

\begin{lemma}[\cite{GoemansW95}]\label{mod-tree}
There is a polynomial-time algorithm that computes a mod-$k$-tree cover $\T_{\bmod k}$ in $G[V]$ such that $w(\T_{\bmod k})\leq w(\C_{\bmod k}^{*})$.
\end{lemma}

In Steps 2 and 3, we obtain a mod-$k$-cycle cover
based on the mod-$k$-tree cover $\T_{\bmod k}$ and a
minimum weight perfect matching $\M$ on the odd-degree vertices of $\T_{\bmod k}$ using the Christifides-Serdykov method for metric TSP. Although it may not guarantee a better approximation ratio for the minimum weight mod-$k$-cycle cover problem, it is useful for $k$-CVRP.

\begin{lemma}\label{mod2}
There is a polynomial-time algorithm that computes a mod-$k$-cycle cover $\C_{\bmod k}$ in $G[V]$ such that $w(\C_{\bmod k})\leq w(\C_{\bmod k}^{*})+\frac{1}{2}w(H^*)$.
\end{lemma}
\begin{proof}
After finding the mod-$k$-tree cover $\T_{\bmod k}$ in $G[V]$, we go to find a minimum weight perfect matching $\M$ on the odd-degree vertices of $\T_{\bmod k}$. Note that $w(\M)\leq\frac{1}{2}w(H^*)$ by the triangle inequality.
Then, we obtain a set of components, where in each component the number of vertices is divisible by $k$ and the degree of each vertex is even. By shortcutting each component, we obtain a mod-$k$-cycle cover $\C_{\bmod k}$ with $w(\C_{\bmod k})\leq w(\T_{\bmod k})+w(\M)\leq w(\C_{\bmod k}^{*})+\frac{1}{2}w(H^*)$.
\end{proof}

Now, we are ready to analyze the three algorithms.

\begin{theorem}\label{res-intitial-split}
Given an $\alpha$-approximate Hamiltonian cycle $H$ on $V\cup\{v_0\}$, there is an approximation algorithm for splittable $k$-CVRP and unit-demand $k$-CVRP such that
\begin{itemize}
\item If $1\leq\alpha\leq \frac{7}{6}$ and $k\geq 3$, the approximation ratio is $\alpha+1-\frac{\alpha}{k}-\frac{\alpha-0.5}{k}$;
\item If $\frac{7}{6}\leq\alpha\leq\frac{3}{2}$ and $3\leq k\leq 5$, the approximation ratio is $\frac{13k-11}{6k}$;
\item If $\frac{7}{6}\leq\alpha\leq\frac{3}{2}$ and $k\geq 6$, the approximation ratio is $\alpha+1-\frac{\alpha}{k}-\frac{4(\alpha-1)}{k}$.
\end{itemize}
\end{theorem}
\begin{proof}
Recall that we will make a trade-off between three algorithms. The first algorithm uses the mod-$k$-cycle cover $\C_{\bmod k}$ in $G[V]$ in Lemma~\ref{mod2}, the second uses an $\alpha$-approximate Hamiltonian cycle $H$ on $V\cup\{v_0\}$, and the third uses the Hamiltonian cycle $H_{CS}$ on $V\cup\{v_0\}$.

Using the mod-$k$-cycle cover $\C_{\bmod k}$, by Corollary~\ref{ITPpac} and Lemmas~\ref{lb-tsp}, \ref{lb-delta}, \ref{lb-cyclepacking}, and \ref{mod2}, EX-ITP can compute a solution with weight at most
\begin{align*}
\frac{2}{k}\Delta+\lrA{1-\frac{1}{k}}w(\C_{\bmod k})\leq\ &\frac{2}{k}\Delta+\frac{k-1}{k}\lrA{w(\C_{\bmod k}^{*})+\frac{1}{2}w(H^*)}\\
\leq\ &\frac{2}{k}\lrA{\ra+\frac{k-2}{2}}w(I^*)+\frac{k-1}{k}\lrA{2(1-\ra)+\frac{1}{2}}w(I^*)\\
=\ &\lrA{-\frac{2k-4}{k}\ra+\frac{7k-9}{2k}}w(I^*).
\end{align*}
The approximation ratio, denoted by $\rho_1(\ra)$, satisfies $\rho_1(\ra)=-\frac{2k-4}{k}\ra+\frac{7k-9}{2k}$.

Using the Hamiltonian cycle $H$, by Lemmas~\ref{lb-tsp}, \ref{lb-delta}, and \ref{AGITP}, AG-ITP can compute a solution with weight at most
\begin{align*}
\frac{2}{k}\Delta+\lrA{1-\frac{1}{k}}w(H)\leq\ &\frac{2}{k}\Delta+\frac{k-1}{k}\alpha w(H^*)\\
\leq\ &\frac{2}{k}\lrA{\ra+\frac{k-2}{2}}w(I^*)+\frac{k-1}{k}\alpha w(I^*)\\
=\ &\lrA{\frac{2}{k}\ra+\frac{(k-1)\alpha+k-2}{k}}w(I^*).
\end{align*}
The approximation ratio, denoted by $\rho_2(\ra)$, satisfies $\rho_2(\ra)=\frac{2}{k}\ra+\frac{(k-1)\alpha+k-2}{k}$.

Using the Hamiltonian cycle $H_{CS}$, by Lemmas~\ref{lb-delta}, \ref{lb-chris}, and \ref{AGITP}, AG-ITP can compute a solution with weight at most
\begin{align*}
\frac{2}{k}\Delta+\lrA{1-\frac{1}{k}} w(H_{CS})\leq\ &\frac{2}{k}\lrA{\ra+\frac{k-2}{2}}w(I^*)+\frac{k-1}{k}\cdot\frac{3-\ra}{2}w(I^*)\\
=\ &\lrA{-\frac{k-5}{2k}\ra+\frac{5k-7}{2k}}w(I^*).
\end{align*}
The approximation ratio, denoted by $\rho_3(\ra)$, satisfies $\rho_3(\ra)=-\frac{k-5}{2k}\ra+\frac{5k-7}{2k}$.

By making a trade-off, we can obtain an approximation ratio of
\[
\max_{0\leq \ra\leq 1}\min\lrc{\rho_1(\ra),\rho_2(\ra),\rho_3(\ra)}.
\]
The results can be briefly summarized as follows.

\textbf{Case~1: $1\leq\alpha\leq \frac{7}{6}$ and $k\geq 3$.} The best approximation ratio is
\[
\max_{0\leq \ra\leq1}\min\{\rho_1(\ra),\rho_2(\ra)\}=\alpha+1-\frac{\alpha}{k}-\frac{\alpha-0.5}{k}.
\]

\textbf{Case~2: $\frac{7}{6}\leq\alpha\leq\frac{3}{2}$ and $3\leq k\leq 5$.} The best approximation ratio is
\[
\max_{0\leq \ra\leq1}\min\{\rho_1(\ra),\rho_3(\ra)\}=\frac{13k-11}{6k}.
\]

\textbf{Case~3: $\frac{7}{6}\leq\alpha\leq\frac{3}{2}$ and $k\geq 6$.} The best approximation ratio is
\[
\max_{0\leq \ra\leq1}\min\{\rho_2(\ra),\rho_3(\ra)\}=\alpha+1-\frac{\alpha}{k}-\frac{4(\alpha-1)}{k}.
\]
\end{proof}

By Theorem~\ref{res-intitial-split}, we improve the approximation ratio $\alpha+1-\frac{\alpha}{k}$ by a term of at least $\frac{0.5}{k}$ (when $\alpha=1$, the improvement is $\frac{0.5}{k}$), which is larger than the improvement $\frac{1}{3k^3}$ in~\cite{BompadreDO06}. When $\alpha=\frac{3}{2}$, our approximation ratio is $\frac{5}{2}-\frac{3.5}{k}$, and the improvement is $\frac{2}{k}$, which is also larger than the improvement $\frac{1}{4k^2}$ in~\cite{BompadreDO06}.
For some detailed comparisons, readers can refer to Table~\ref{res-1}. According to our results, better approximation algorithms for the minimum weight mod-$k$-cycle cover problem may directly lead to further improvements for $k$-CVRP.

\section{A Further Improvement for Splittable $k$-CVRP}\label{sec-split-final}
The current approximation ratio for metric TSP is about $\frac{3}{2}$. Under $\alpha=\frac{3}{2}$, the approximation ratio of splittable $k$-CVRP and unit-demand $k$-CVRP is at most $\frac{5}{2}-\frac{3.5}{k}=\frac{5}{2}-\Theta(\frac{1}{k})$ in the previous section.
In this section, we further improve the approximation ratio to $\frac{5}{2}-\Theta(\frac{1}{\sqrt{k}})$ by carefully analyzing one of the three algorithms in the previous section.

Our motivation is that the approximation ratio $\frac{5}{2}-\frac{3.5}{k}$ in the previous section is obtained with $\ra=0$ in the worst case. That is, in the worst case, the total weight of all home-edges in the optimal solution $I^*$ is zero. Recall that each tour of $I^*$ contains two home-edges and the upper bound on $\MST$ in Lemma~\ref{lb-tree} is obtained by deleting the more weighted home-edge in each tour of $I^*$. Since $\ra=0$ in the worst case, we have $w(I^*)\geq \MST$ by Lemma~\ref{lb-tree}. So, by deleting a zero-weighted home-edge, we gain no revenues. We may obtain a refined upper bound on $\MST$ by deleting the highest weighted edge in each tour, which is clearly better than just deleting a zero-weighted home-edge. Based on this, we give two refined lower bounds based on $\MST$ as well as $\Delta$ to combine them and prove an approximation ratio of $\frac{5}{2}-\Theta(\frac{1}{\sqrt{k}})$.

Our algorithm is simply to call AG-ITP on the Hamiltonian cycle $H_{CS}$ on $V\cup\{v_0\}$. Using the Hamiltonian cycle $H_{CS}$, by Lemmas~\ref{lb-tsp}, \ref{chris} and \ref{AGITP}, AG-ITP can compute an itinerary $I$ with
\begin{equation}\label{Choris+ITP}
\begin{split}
w(I)&\leq \frac{2}{k}\Delta+\lrA{1-\frac{1}{k}}w(H_{CS})\\
&\leq \frac{2}{k}\Delta+\lrA{1-\frac{1}{k}}\mbox{MST}+\frac{1}{2}\lrA{1-\frac{1}{k}}w(I^{*}).
\end{split}
\end{equation}
Then, to prove that AG-ITP achieves an approximation ratio of $\frac{5}{2}-\Theta(\frac{1}{\sqrt{k}})$, we show that $\frac{2}{k}\Delta+(1-\frac{1}{k})\MST\leq (2-\Theta(\frac{1}{\sqrt{k}}))w(I^*)$.

Now, we analyze the structure in more detail. By Assumption~\ref{ass1}, $I^*$ consists of a set of $(k+1)$-cycles.
Consider an arbitrary $(k+1)$-cycle $C=(v_0,v_1,\dots, v_{k},v_0)$ in $I^*$.
First, we use $\Delta_C$ to denote the sum of the weights of the edges between vertices $v_{0}$ and $v_i$ for $1\leq i\leq k$, i.e., $\Delta_C=\sum_{i=1}^{k}w(v_0, v_i)$. Note that $\Delta=\sum_{C\in I^{*}}\Delta_C$.
Then, we use $T_C$ to denote the spanning tree in $G[C]$ obtained by deleting the highest weighted edge in $C$, i.e., $w(T_C)=w(C)-\max_{0\leq i\leq k}w(v_i,v_{(i+1)\bmod (k+1)})$.
After we delete the highest weighted edge for each cycle $C\in I^{*}$, the remaining graph forms a spanning tree in $G$. Thus, we have $\mbox{MST}\leq\sum_{C\in I^{*}}w(T_C)$. Moreover, since $w(I^*)=\sum_{C\in I^*}w(C)$, to prove $\frac{2}{k}\Delta+(1-\frac{1}{k})\MST\leq (2-\Theta(\frac{1}{\sqrt{k}}))w(I^*)$, it suffices to prove 
\[
\frac{2}{k}\Delta_C+\lrA{1-\frac{1}{k}}w(T_C)\leq \lrA{2-\Theta\lrA{\frac{1}{\sqrt{k}}}}w(C). 
\]

\begin{lemma}\label{core}
When $k\geq3$, for any cycle $C=(v_0,v_1,...,v_k,v_0)\in I^*$, we have
\[
w(T_C)\leq\lrA{1-\max_{1\leq i\leq m}\frac{1}{2}x_{i}}w(C)\quad\text{and}\quad\Delta_C\leq\lrA{\frac{k+2}{2}-\sum_{i=1}^{m}ix_{i}}w(C),
\]
where it holds that $0\leq x_i\leq 1$ for each $1\leq i\leq m$, $\sum_{i=1}^{m}x_i=1$, and $m=\Ceil{\frac{k+1}{2}}$.
Moreover, we have
\[
\frac{2}{k}\Delta_C+\lrA{1-\frac{1}{k}}w(T_C)\leq\lrA{2-\frac{2l^2+k-1}{2kl}}w(C),
\]
where $l=\ceil{\frac{\sqrt{2k-1}-1}{2}}$.
\end{lemma}
\begin{proof}
Our proof is based on generalizing the concept of home-edges. The details are put in Section~\ref{proof-1}.    
\end{proof}


Next, we are ready to analyze AG-ITP.

\begin{theorem}\label{res-split}
For splittable $k$-CVRP and unit-demand $k$-CVRP, AG-ITP admits an approximation ratio of $\frac{5}{2}-\frac{2l^2+k+l-1}{2kl}<\frac{5}{2}-\sqrt{\frac{2}{k}}$, where $l=\ceil{\frac{\sqrt{2k-1}-1}{2}}$.
\end{theorem}
\begin{proof}
Using the Hamiltonian cycle $H_{CS}$, by (\ref{Choris+ITP}), AG-ITP outputs a solution with weight at most $\frac{2}{k}\Delta+(1-\frac{1}{k})\mbox{MST}+\frac{1}{2}(1-\frac{1}{k})w(I^*)$.
Then, by Lemma~\ref{core}, we have
\begin{align*}
&\frac{2}{k}\Delta+ \lrA{1- \frac{1}{k}}\mbox{MST}+ \frac{1}{2} \lrA{1- \frac{1}{k}}w(I^*)\\
&\leq \sum_{C\in I^*} \lrA{ \frac{2}{k}\Delta_C+ \frac{k-1}{k}w(T_C)}+\frac{k-1}{2k}w(I^*)\\
&\leq \sum_{C\in I^*} \lrA{2- \frac{2l^2+k-1}{2kl}}w(C)+ \frac{k-1}{2k}w(I^*)\\
&= \lrA{2- \frac{2l^2+k-1}{2kl}}\sum_{C\in I^*}w(C)+ \frac{k-1}{2k}w(I^*)\\
&=  \lrA{2- \frac{2l^{2}+k-1}{2kl}+ \frac{k-1}{2k}}w(I^*)\\
&= \lrA{ \frac{5}{2}- \frac{2l^{2}+k+l-1}{2kl}}w(I^*).
\end{align*}
To show that $\frac{5}{2}-\frac{2l^2+k+l-1}{2kl}<\frac{5}{2}-\sqrt{\frac{2}{k}}$ holds for any $k\geq 3$, it suffices to prove
\[
\frac{2l^2+k+l-1}{2kl}>\sqrt{\frac{2}{k}}\quad\Longleftrightarrow\quad2l^2-(2\sqrt{2k}-1)l+k-1>0.
\]
The latter holds because the discriminant of the quadratic equation satisfies $(2\sqrt{2k}-1)^2-8(k-1)=9-4\sqrt{2k}<0$ for any $k\geq 3$.
\end{proof}

In the previous section, when $\alpha=\frac{3}{2}$ and $k>5$, the approximation ratio of our algorithm is given by $\frac{5}{2}-\frac{3.5}{k}$. 
It can be verified that the new result $\frac{5}{2}-\Theta(\frac{1}{\sqrt{k}})$ in Theorem~\ref{res-split} is strictly better for any $k>5$. It is also better than $\frac{5}{2}-\frac{1.005}{3000}-\frac{1.5}{k}$ for any $k\leq 1.7\times10^7$ (see Table~\ref{res-1}).
When $3\leq k\leq 5$, the approximation ratio equals $2-\frac{1}{k}$, which is tight for this algorithm, since it has been shown~\cite{tight} that the approximation ratio of AG-ITP is at least $2-\frac{1}{k}$ in general metrics, even using an optimal Hamiltonian cycle.

In this section, we mainly consider $\alpha=\frac{3}{2}$. 
For other values of $\alpha$, similar improved results, as in Theorem~\ref{res-intitial-split}, could potentially be obtained by applying the refined lower bounds based on $\MST$ and $\Delta$.
However, the analysis is sophisticated and achieving significant improvements for metric TSP appears challenging. 
We leave the exploration of this direction for future research.

\section{ITP for Unsplittable CVRP}\label{Sec-REUITP}
From this section, we consider unsplittable CVRP.
The ITP-based algorithms play an important role in solving splittable $k$-CVRP. A natural idea is to extend  ITP-based algorithms for unsplittable $k$-CVRP.
In order to compute a feasible solution without splitting the demand of  one customer into two tours, we need to modify the initial solution obtained by ITP.
One of the most famous algorithms for unsplittable CVRP is the algorithm (AG-UITP) introduced by Altinkemer and Gavish~\cite{altinkemer1987heuristics}.
The main idea of the algorithm is as follows.
Given a Hamiltonian cycle $H$ on $V\cup\{v_0\}$ as part of the input, the algorithm first uses AG-ITP to compute a solution for unit-demand $k/2$-CVRP with weight at most $\frac{4}{k}\Delta+(1-\frac{2}{k})w(H)$. Then, in the initial solution, if the demand of a customer is delivered by two tours, then set all the demand to one tour to obtain a feasible solution for unsplittable $k$-CVRP without increasing the weight of the solution. By Lemmas~\ref{lb-tsp} and \ref{lb-delta}, we have
\[
\frac{4}{k}\Delta+\lrA{1-\frac{2}{k}}\alpha w(H^*)\leq 2w(I^*)+\lrA{1-\frac{2}{k}}\alpha w(I^*)=\lrA{\alpha+2-\frac{2\alpha}{k}}w(I^*).
\]
Thus, the approximation ratio of AG-UITP is $\alpha+2-\frac{2\alpha}{k}$. Note that $k$ should be even  in AG-UITP. For odd $k$, there are two possible ways to solve it.
The first is to double the vehicle capacity and all customers' demand to obtain a new instance, and then call AG-UITP. For this case, we may obtain an approximation ratio of  $\alpha+2-\frac{\alpha}{k}$; the second is to solve unit-demand $\lfloor k/2\rfloor $-CVRP in the first step of AG-UITP. However, for the second method, we may obtain an even worse approximation ratio.


Next, we propose a refined AG-UITP algorithm for unsplittable CVRP. Our algorithm has two advances: it has an uniform form for both odd and even $k$, and the approximation ratio will be improved.

A customer is called \emph{big} (resp., \emph{small}) if its demand is greater than (resp., at most) $\floor{k/2}$. Let $V_{big}$ (resp., $V_{mall}$) denote the set of big (resp., small) customers. Define $\Delta_{big}=\sum_{v_i\in V_{big}}d_iw(v_0,v_i)$, and
$\Delta_{small}=\sum_{v_i\in V_{small}}d_iw(v_0,v_i)$.

The refined AG-UITP works as follows. First, we assign a single trivial tour for each big customer. Second, we compute a Hamiltonian cycle $H_{small}$ on $V_{small}\cup\{v_0\}$ by shortcutting all big customers from a Hamiltonian cycle $H$ on $V\cup\{v_0\}$. Third, we use $H_{small}$ to call AG-ITP for unit-demand $(\ceil{k/2}+1)$-CVRP. Last, if a customer's demand is served by two tours, then adjust all its demand to only one tour: if one customer $v_i$ is served by two consecutive tours $T_1=(v_0,v_{h},\dots,v_{i},v_0)$ and $T_2=(v_0,v_i,\dots, v_j,v_0)$ in the clockwise direction of the Hamiltonian cycle, we let $T_2$ deliver all the demand for $v_i$ without exceeding the capacity and then delete $v_i$ from $T_1$; if one customer $v_i$ is served by three tours, say $T_1=(v_0,\dots,v_i,v_0)$, $T_2=(v_0,v_i,v_0)$, and $T_3=(v_0,v_i,\dots,v_0)$, we take $T_2$ as a trivial tour delivering all the demand of $v_i$ and then remove $v_i$ from both $T_1$ and $T_2$.

\begin{lemma}\label{REUITP}
Given a Hamiltonian cycle $H$ on $V\cup\{v_0\}$ as part of the input, for unsplittable $k$-CVRP with any $k\geq 3$, the refined AG-UITP in $O(nk)$ time computes a feasible solution of weight at most 
$
\frac{2}{\floor{ k/2}+1}\Delta+\lra{1-\frac{1}{\floor{ k/2}+1}}w(H).
$
\end{lemma}
\begin{proof}
First, we analyze the quality of the solution before modifying the demand of customers. For each big customer, we assign a trivial tour. The total weight of these trivial tours is at most
\[
\sum_{v_i\in V_{big}}2w(v_0,v_i)\leq\frac{2}{\floor{ k/2}+1}\sum_{v_i\in V_{big}}d_iw(v_0,v_i)=\frac{2}{\floor{k/2}+1}\Delta_{big},
\]
where the inequality follows from the fact that each big customer's demand is at least $\floor{k/2}+1$.

For all small customers, we obtain a Hamiltonian cycle $H_{small}$ from $H$ by shortcutting all big customers, and then use AG-ITP with a capacity of $\ceil{k/2}+1$ to obtain an itinerary. By Lemma~\ref{AGITP}, its weight is at most
\[
\frac{2}{\ceil{k/2}+1}\Delta_{small}+\lrA{1-\frac{1}{\ceil{k/2}+1}}w(H_{small}).
\]

Note that $w(H_{small})\leq 2\Delta_{small}$ and $w(H_{small})\leq w(H)$ by the triangle inequality. Then, we have
\begin{align*}
&\frac{2}{\ceil{k/2}+1}\Delta_{small}+\lrA{1-\frac{1}{\ceil{k/2}+1}}w(H_{small})\\
&\leq \frac{2}{\floor{ k/2}+1}\Delta_{small}+\lrA{1-\frac{1}{\floor{ k/2}+1}}w(H_{small})\\
&\leq \frac{2}{\floor{ k/2}+1}\Delta_{small}+\lrA{1-\frac{1}{\floor{ k/2}+1}}w(H).
\end{align*}

Thus, the total weight is at most
$
\frac{2}{\floor{k/2}+1}(\Delta_{big}+\Delta_{small})+\lra{1-\frac{1}{\floor{k/2}+1}}w(H)= \frac{2}{\floor{ k/2}+1}\Delta+\lra{1-\frac{1}{\floor{ k/2}+1}}w(H).
$

Next, we analyze the modification in the last step of our algorithm. 

If a customer is served by three tours,
the modification clearly erases an infeasible case without increasing the weight. If a customer $v_i$ is split into two consecutive tours $T_1=(v_0,v_{h},\dots,v_{i},v_0)$ and $T_2=(v_0,v_i,\dots, v_j,v_0)$ in the clockwise direction of the Hamiltonian cycle, we modify them by letting $T_{1}^{*}=(v_0,v_{h},\dots,v_{i-1},v_0)$ and $T_{2}^{*}=(v_0,v_{i},\dots, v_j,v_0)$.

By the triangle inequality, it holds
$
w(T_{1}^{*})+w(T_{2}^{*})\leq w(T_1)+w(T_2).
$
So, the weight is non-increasing. To prove the feasibility, we only need to prove that the total demand in $T_{1}^{*}$ and $T_{2}^{*}$ is at most $k$, respectively.
For $T_{1}^{*}$, the total demand becomes less and hence is feasible. For $T_{2}^{*}$, the total demand is bigger than the total demand of $T_2$. Note that the tour $T_2$ has a total demand of at most $\ceil{k/2}+1$. We analyze the newly added demand. Since any small customer has a demand of at most $\floor{k/2}$, the newly added demand is at most $\floor{k/2}-1$. So, for $T_{2}^{*}$, the total demand is at most $(\ceil{k/2}+1)+(\floor{ k/2}-1)=k$. Thus, after modifying all two consecutive conflict tours, we obtain a feasible itinerary with a non-increasing weight.
\end{proof}

Using an $\alpha$-approximate Hamiltonian cycle on $V\cup\{v_0\}$, by Lemmas~\ref{lb-tsp} and \ref{REUITP}, the refined AG-UITP satisfies that
$
\frac{2}{\floor{ k/2}+1}\Delta+\lra{1-\frac{1}{\floor{ k/2}+1}}\alpha w(H^*)\leq \lra{\alpha+\frac{k}{\floor{ k/2}+1}-\frac{\alpha}{\floor{ k/2}+1}}w(I^*).
$
Hence, the approximation ratio is $\alpha+\frac{k}{\floor{ k/2}+1}-\frac{\alpha}{\floor{ k/2}+1}$ for any $k\geq 3$, which is even better than the result $\alpha+2-\frac{2\alpha}{k}-\Omega(\frac{1}{k^3})$ in~\cite{BompadreDO06}.

In the refined AG-UITP, we call AG-ITP with a capacity of $\ceil{k/2}+1$ on the Hamiltonian cycle of small customers and then we can modify them by adjusting the demand. This idea can be used on a cycle cover of small customers. We call EX-ITP with a capacity of $\ceil{k/2}+1$ on the cycle cover of small customers and then modify the conflicting tours in the same way. The algorithm is called the EX-UITP algorithm. By Lemma~\ref{EXITP}, we have the following result.
\begin{lemma}~\label{EXUITP}
Given a cycle cover $\C$ in $G[V]$ or $G[V\cup\{v_0\}]$ as part of the input, assuming all customers are small customers, for unsplittable $k$-CVRP with any $k\geq 3$, EX-UITP in $O(n^2)$ time computes a feasible solution with weight at most $2g\Delta+(1-g)w(\C)$, where $g=\max_{C\in\C}\frac{\ceil{\size{C}/(\ceil{k/2}+1)}}{\size{C}}$.
\end{lemma}
The lemma can be used to analyze the algorithms based on a cycle cover $\C$, where we may assign a trivial tour to each big customer and then call EX-UITP on the cycle cover $\C'$ of small customers obtained by shortcutting all big customers in $\C$.
However, we remark that if we simply assign a trivial tour to each big customer, the structure of the obtained cycle cover $\C'$ may become extremely bad, e.g., it may become a cycle that contains only a unit-demand customer and in this case $g=1$. Thus, for unsplittable $k$-CVRP, the algorithms based on a cycle cover are harder to design and analyze.

\section{Some Structural Properties}\label{Sec-Unsplit-Property}
EX-UITP provides the framework of our algorithm. In this section, we prove some properties that may be used for some local structures.

For unsplittable $k$-CVRP, given a cycle $C$, we let $\Delta_C=\sum_{v_i\in C}d_iw(v_0,v_i)$ and $\size{C}=\sum_{v_i\in C}d_i$.
Recall that the demand of a small (resp., big) customer is at most $\floor{k/2}$ (resp., at least $\floor{k/2}+1$). 
Next, we propose three lemmas for three special cases, where the results are better than the result in Lemma~\ref{EXUITP}.

\begin{lemma}\label{rule1}
For unsplittable $k$-CVRP, if there is a cycle $C$ with $1\leq\size{C}\leq k$, we can assign a single tour on $C$ with weight at most $2g\Delta_C+(1-g)w(C)$, where $g=\frac{1}{\size{C}}$.
\end{lemma}
\begin{proof}
We can obtain a tour by using EX-ITP with a capacity of $k$. By Lemma~\ref{EXITP}, the weight is at most $2g\Delta_C+(1-g)w(C)$, where $g=\frac{\ceil{\size{C}/k}}{\size{C}}=\frac{1}{\size{C}}$ since $\size{C}\leq k$. Moreover, this tour is feasible since $\size{C}\leq k$.
\end{proof}

Note that Lemma~\ref{EXUITP} does not allow big customers on $C$, whereas Lemma~\ref{rule1} does.

\begin{lemma}\label{rule2}
For unsplittable $k$-CVRP, if there is a cycle $C$ with $\size{C}>k\geq4$ and there are no big customers in $C$, we can assign $m=\ceil{\frac{\size{C}-1}{\ceil{k/2}+1}}$ tours on $C$ with weight at most $2g\Delta_C+(1-g)w(C)$, where $g=\frac{m}{\size{C}}$.
Moreover, the first tour has a capacity of $\ceil{k/2}+2$, each of the middle $m-2$ tours has a capacity of $\ceil{k/2}+1$, and the last tour has a capacity of $\size{C}-(\ceil{k/2}+2)-(m-2)(\ceil{k/2}+1)$.
\end{lemma}
\begin{proof}
Note that $m\geq2$ since $\size{C}>k\geq4$.
We will show that we can assign $m$ tours on $\size{C}$ with weight at most $2g\Delta_C+(1-g)w(C)$.

The first tour has a capacity of $\ceil{k/2}+2$, each of the middle $m-2$ tours has a capacity of $\ceil{k/2}+1$, and the last tour has a capacity of $\size{C}-(\ceil{k/2}+2)-(m-2)(\ceil{k/2}+1)$. Note that $0\leq\size{C}-(\ceil{k/2}+2)-(m-2)(\ceil{k/2}+1)\leq\ceil{k/2}+1$ since $\frac{\size{C}-1}{\ceil{k/2}+1}\leq m=\ceil{\frac{\size{C}-1}{\ceil{k/2}+1}}<\frac{\size{C}-1}{\ceil{k/2}+1}+1$. Hence, each tour has a non-negative capacity, and the total capacity of these $m$ tours is exactly $\size{C}$. Then, we can obtain $m$ such tours on $C$. Similar to HR-ITP, by considering $\size{C}$ solutions and selecting the best one, we can obtain a solution on $\size{C}$ with weight at most $\frac{2m}{\size{C}}\Delta_C+(1-\frac{m}{\size{C}})w(C)=2g\Delta_C+(1-g)w(C)$.

Then, we show these tours can be modified into a feasible solution with a non-increasing weight.

The idea is similar to the analysis in Lemma~\ref{REUITP}. Note that the tour has a capacity of at most $\ceil{k/2}+1$ in Lemma~\ref{REUITP}.
However, in the $m$ tours, the first tour, denoted by $T_1$, has a capacity of $\ceil{k/2}+2$, while the other tours, denoted by $T_2,\dots,T_m$ in the clockwise direction of the cycle $C$, have a capacity of at most $\ceil{k/2}+1$. The main difference is due to the tour $T_1$. We will use a similar modification in Lemma~\ref{REUITP} while making sure that the tour $T_1$ is feasible. 

Similar to AG-UITP, if a customer $v_1$ is served by three tours, say $T_1=(v_0,\dots,v_1,v_0)$, $T_2=(v_0,v_1,v_0)$, and $T_3=(v_0,v_1,\dots,v_0)$, we take $T_2$ as a trivial tour delivering all the demand of $v_1$, and then remove $v_1$ from both $T_1$ and $T_2$.
The modification is safe. Then, we assume that each customer is served by at most two tours.

Then, we consider the following three cases.

\textbf{Case~1: there are no tours having a conflict with $T_1$.} We can modify every two consecutive conflict tours $T_l=(v_0,v_h,\dots,v_i,v_0)$ and $T_{l+1}=(v_0,v_i,\dots,v_j,v_0)$ by letting $T^*_l=(v_0,v_h,\dots,v_{i-1},v_0)$ and $T^*_{l+1}=(v_0,v_{i},\dots,v_j,v_0)$. The modification is in the clockwise direction, which is the same as that in Lemma~\ref{REUITP}. In this case, the tour $T_1$ does not affect the modification. For other tours, by the proof of Lemma~\ref{REUITP}, the modification is feasible.

\textbf{Case~2: there exists one tour having a conflict with $T_1$.} In this case, the modification in the clockwise direction in Lemma~\ref{REUITP} may cause the tour $T_1$ infeasible. For example, we let $T_m=(v_0,\dots,v_i,v_0)$ and $T_1=(v_0,v_i,\dots,v_0)$, where the delivered demand of $v_i$ in $T_m$ and $T_1$ are $\floor{k/2}-1$ and 1, respectively. 
Then, using the modification in Lemma~\ref{REUITP}, we have $T^*_m=(v_0,\dots,v_{i-1},v_0)$ and $T^*_1=(v_0,v_{i},\dots,v_0)$. The original demand of $T_1$ is $\ceil{k/2}+2$ and the newly added demand is $\floor{k/2}-1$. The total demand is $k+1$ and hence infeasible. Actually, in this case, we can modify the tours $T_m$ and $T_1$ in the counterclockwise direction by letting $T^*_m=(v_0,\dots,v_{i},v_0)$ and $T^*_1=(v_0,v_{i+1},\dots,v_0)$. Next, we consider the following two cases.

\textbf{Case~2.1: there is only one tour having a conflict with $T_1$.} If the tour is $T_2$, we modify every two consecutive conflict tours $T_l=(v_0,v_h,\dots,v_i,v_0)$ and $T_{l+1}=(v_0,v_i,\dots,v_j,v_0)$ in the clockwise direction by letting $T^*_l=(v_0,v_h,\dots,v_{i-1},v_0)$ and $T_{l+1}=(v_0,v_{i},\dots,v_j,v_0)$; if the tour is $T_m$, we modify every two consecutive conflict tours $T_l=(v_0,v_h,\dots,v_i,v_0)$ and $T_{l+1}=(v_0,v_i,\dots,v_j,v_0)$ in the counterclockwise direction by letting $T^*_l=(v_0,v_h,\dots,v_{i},v_0)$ and $T_{l+1}=(v_0,v_{i+1},\dots,v_j,v_0)$. The directions of these two modifications are opposite. The demand of the modified tour $T_1$ will always become less, and hence the modification will always be feasible. For other tours, by the proof of Lemma~\ref{REUITP}, the modification is safe.

\textbf{Case~2.2: there are two tours having conflicts with $T_1$.} We can modify every two consecutive conflict tours $T_l=(v_0,v_h,\dots,v_i,v_0)$ and $T_{l+1}=(v_0,v_i,\dots,v_j,v_0)$ in the clockwise direction by letting $T^*_l=(v_0,v_h,\dots,v_{i-1},v_0)$ and $T_{l+1}=(v_0,v_{i},\dots,v_j,v_0)$. The modification is the same as that in Lemma~\ref{REUITP}. We only need to prove $T^*_1$ is feasible. The original demand of $T_1$ is at most $\ceil{k/2}+2$. Since $T_1$ has a conflict with the tour $T_m$ and each small customer has a demand of at most $\floor{k/2}$, the newly added demand of $T^*_1$ in the modification is at most $\floor{k/2}-1$. Moreover, since $T_1$ also has a conflict with the tour $T_2$, it reduces at least 1 demand according to the modification. Therefore, the demand of $T^*_1$ is at most $(\ceil{k/2}+2)+(\floor{k/2}-1)-1=k$. 
Hence, the tour $T_1$ will always be feasible. For other tours, by the proof of Lemma~\ref{REUITP}, the modification is safe.
\end{proof}

\begin{lemma}\label{rule3}
For unsplittable $k$-CVRP, if there exists a cycle $C$ with $\size{C}> k\geq 4$ and there is only one big customer in $C$ and its demand is $\floor{k/2}+1$, we can assign $m=\ceil{\frac{\size{C}}{\ceil{k/2}+1}}$ tours on $C$ with weight at most $2g\Delta_C+(1-g)w(C)$, where $g=\frac{m}{\size{C}}$.
\end{lemma}
\begin{proof}
First, we call EX-ITP with a capacity of $\ceil{k/2}+1$ on $C$. By Lemma~\ref{EXITP}, we know the weight is at most $2g\Delta_C+(1-g)w(C)$.

Next, we modify the conflict tours. 

Similar to AG-UITP, if a customer $v_1$ is served by three tours, say $T_1=(v_0,\dots,v_1,v_0)$, $T_2=(v_0,v_1,v_0)$, and $T_3=(v_0,v_1,\dots,v_0)$, we take $T_2$ as a trivial tour delivering all the demand of $v_1$, and then remove $v_1$ from both $T_1$ and $T_2$.
The modification is safe. Then, we assume that each customer is served by at most two tours.

Compared to Lemma~\ref{EXUITP}, the only difference is that there is one big customer with a demand of $\floor{k/2}+1$. Then, we only need to prove that the tours involving the big customer can be modified into feasible tours. 

We consider the following two cases.

\textbf{Case~1: the big customer is served by only one tour.} We can modify every two consecutive conflict tours $T_l=(v_0,v_h,\dots,v_i,v_0)$ and $T_{l+1}=(v_0,v_i,\dots,v_j,v_0)$ in the clockwise direction by letting $T^*_l=(v_0,v_h,\dots,v_{i-1},v_0)$ and $T_{l+1}=(v_0,v_{i},\dots,v_j,v_0)$. The modification is the same as that in Lemma~\ref{REUITP}. In this case, the big customer does not affect the modification. Hence, the modification is feasible.

\textbf{Case~2: the big customer is served by two tours.}  Assume that the big customer $v_1$ is served by the tours $T_1=(v_0,\dots,v_1,v_0)$ and $T_2=(v_0,v_1,\dots,v_0)$. The delivered demand of $v_1$ in $T_1$ (resp., $T_2$) is denoted by $x_1$ (resp., $x_2$). If $x_1\leq x_2$, we modify every two consecutive conflict tours $T_l=(v_0,v_h,\dots,v_i,v_0)$ and $T_{l+1}=(v_0,v_i,\dots,v_j,v_0)$ in the clockwise direction by letting $T^*_l=(v_0,v_h,\dots,v_{i-1},v_0)$ and $T_{l+1}=(v_0,v_{i},\dots,v_j,v_0)$. 
Otherwise, we have $x_1\geq x_2$, and then we modify every two consecutive conflict tours $T_l=(v_0,v_h,\dots,v_i,v_0)$ and $T_{l+1}=(v_0,v_i,\dots,v_j,v_0)$ in the counterclockwise direction by letting $T^*_l=(v_0,v_h,\dots,v_{i},v_0)$ and $T_{l+1}=(v_0,v_{i+1},\dots,v_j,v_0)$. The directions of these two modifications are opposite. Due to symmetry, we only analyze the case $x_1\leq x_2$. 
The big customer is entirely contained in the tour $T^*_2$ after modification. We only need to prove that $T^*_2$ is feasible. The demand of the original tour $T_2$ has a demand of at most $\ceil{k/2}+1$. The newly added demand is $x_1\leq \floor{k/2}-1$ since $x_1+x_2\leq \floor{k/2}+1$, $x_1\leq x_2$, and $k\geq 4$. Therefore, the demand of $T^*_2$ is at most $(\ceil{k/2}+1)+(\floor{k/2}-1)=k$ and hence feasible.
\end{proof}

\section{Improvements on Unsplittable 3-CVRP and 4-CVRP}\label{Sec-Unsplit-34}
In this section, we show that the techniques on cycle covers in the above two sections can be used to improve unsplittable 3-CVRP and 4-CVRP. The cycle covers that we use are based on the ideas in Section~\ref{SEC5}.

\subsection{Unsplittable 3-CVRP}\label{unsplit-3-cvrp}
For unsplittable 3-CVRP, we will improve the approximation ratio from $1.792$~\cite{gupta2023local} to $\frac{3}{2}=1.500$. Note that this is already the approximation ratio for splittable 3-CVRP. 
Thus, any further improvement on unsplittable 3-CVRP would also apply to splittable 3-CVRP.

By Assumption~\ref{ass3}, no customer has a demand greater than 2. Therefore, each customer's demand is either 1 or 2.
Our algorithm has three main steps. 

First, we find a cycle cover $\C^{**}$ on all customers $V$, where $\left|C_i\right|=\sum_{v_j\in C_i}d_j\geq 3$ for each $C_i\in \C^{**}$. The cycle cover $\C^{**}$ is computed in the following way: we construct a unit-demand instance $G'$ by replacing each customer $v_i$ with $d_i$ unit-demand customers as the same place. Then, we use the polynomial time algorithm in~\cite{hartvigsen1984extensions,schrijver2003combinatorial} to find a minimum weight cycle cover $\C^*$ in $G'[V']$, where $V'$ denotes the set of all customers in $G'$. Note that $\C^*$ may not be a
cycle cover in the original graph $G[V]$. For a vertex $v_i$, there will be $2d_i$ edges incident on it in $\C^*$. So, by shortcutting the cycles in $\C^*$, we can obtain a cycle cover $\C^{**}$ in $G[V]$ such that $w(\C^{**})\leq w(\C^*)$.

Second, we deal with all 2-demand customers: if a cycle $C_i\in \C^{**}$ contains no 1-demand customers or at least two 1-demand customers, assign a trivial tour for each 2-demand customer in $C_i$ and then obtain a cycle $C'_i$ from $C_i$ by shortcutting all 2-demand customers (having been delivered); if a cycle $C\in \C^{**}$ contains exactly one 1-demand customer $u$, then there is at least one 2-demand customer $u'$ in $C$, and we assign a tour $T$ visiting only the two customers $u$ and $u'$ and a trivial tour for each 2-demand customer in $C$ except $u'$. Now, we obtain a cycle cover $\C'=\{C'_i\}$ on all customers that have not been delivered yet. These customers are all 1-demand customers.

Third, we apply EX-ITP on $\C'$.

We analyze the quality of the solution. We partition the tours constructed in the algorithm into three parts.

The first part contains
all trivial tours visiting one 2-demand customer. Let $V_1$ denote this set of 2-demand customers and $\Delta_1= \sum_{v_i\in V_1} d_i w(v_0, v_i)=2\sum_{v_i\in V_1}w(v_0, v_i)$.
Thus, the total weight of the tours in the first part is $\Delta_1$.

Let $\C^{**}_2\subseteq \C^{**}$ be the set of cycles containing exact one 1-demand customer.
For each cycle $C\in \C^{**}_2$, the algorithm constructs a tour $T=(v_0,u,u',v_0)$ visiting one 1-demand customer $u$ and one 2-demand customer $u'$ in $C$ in the second step of the algorithm. The second part consists of all these tours, and we let $V_2$ denote the set of customers visited by these tours. Let $\Delta_2= \sum_{v_i\in V_2} d_i w(v_0, v_i)$.
The weight of $T_0$ is $w(v_0,u)+w(v_0,u')+w(u,u')$, where $w(u,u')$ is at most the half
of $w(C)$. Thus, the total weight of the tours in the second part is at most $\Delta_2+ \frac{1}{2}\sum _{C\in \C^{**}_2}w(C)$.

The third part is the tours computed by EX-ITP in the last step. Let $V_3$ denote the set of 1-demand customers visiting in the tours in the third part and $\Delta_3= \sum_{v_i\in V_3} d_i w(v_0, v_i)=\sum_{v_i\in V_3}w(v_0, v_i)$. Note that for each cycle $C'$ in $\C'$ there are at least two 1-demand customers and hence we have $\size{C'}\geq2$.
By Lemma~\ref{EXITP}, we know that the weight of the tours in the third part is at most $\Delta_3 + \frac{1}{2}w(\C')$ ($g\leq \frac{1}{2}$ in Lemma~\ref{EXITP}),
where $w(\C')\leq w(\C^{**}\setminus \C^{**}_2)$.

Recall that $w(\C^{**})\leq w(\C^*)$. By adding the three parts together, the weight of the solution is at most
\begin{align*}
\Delta_1 + \lrA{\Delta_2+ \frac{1}{2}w(\C^{**}_2)} + \lrA{\Delta_3 + \frac{1}{2}w(\C')}&\leq \Delta +\frac{1}{2}w(\C^{**}_2)+\frac{1}{2} w(\C^{**}\setminus \C^{**}_2)\\
&=\Delta +\frac{1}{2}w(\C^{**})\leq\Delta +\frac{1}{2}w(\C^{*}).    
\end{align*}
Then, by Lemmas~\ref{lb-delta} and~\ref{lb-cyclepacking}, we have
\[
\Delta+\frac{1}{2}w(\C^*)\leq\lrA{\ra+\frac{1}{2}}w(I^*)+\frac{1}{2}\min\{2(1-\ra),1\}w(I^*)\leq\frac{3}{2}w(I^*),
\]
where we have $\ra\geq\frac{1}{2}$ in the worst case. 

Therefore, we have the following result.
\begin{theorem}
For unsplittable 3-CVRP, there is a $\frac{3}{2}$-approximation algorithm.
\end{theorem}

We can see that the algorithm for for unsplittable 3-CVRP is more complicated than the algorithm for splittable 3-CVRP.

\subsection{Unsplittable 4-CVRP}\label{unsplit-4-cvrp}
For unsplittable 4-CVRP, we will improve the approximation ratio from $2.051$~\cite{gupta2023local} to $\frac{7}{4}=1.750$.
Assumption~\ref{ass3} ensures that each customer's demand is either 1, 2 or 3.

Our algorithm first computes a cycle cover $\C_{\bmod2}$, where ``$\bmod~2$'' means that for each cycle in the cycle cover, the total demand of customers in the cycle is even and at least 4.
We call this cycle cover as a mod-$2$-cycle cover. It is computed in the following way: we construct a unit-demand instance $G'$ by replacing each customer $v_i$ with $d_i$ unit-demand customers as the same place. First, we compute a minimum weight perfect matching $\M^*$ in $G'[V']$, where $V'$ denotes the set of all customers in $G'$, and then compute a minimum weight perfect matching $\M^{**}$ in $G'[V']\setminus \M^*$. Let $\C'_{\bmod2}=\M^*\cup\M^{**}$. Note that $\C'_{\bmod2}$ may not be a cycle cover in the original graph $G[V]$. 
Similar to the analysis in the previous subsection, by shortcutting all cycles in $\C'_{\bmod2}$, we can obtain a cycle cover $\C_{\bmod2}$ in $G[V]$. Note that we have $w(\C_{\bmod2})\leq w(\C'_{\bmod2})\leq w(\C_{4}^{*})$ by the triangle inequality and Lemma~\ref{lb-mod2}, and it holds that $\left|C\right|=\sum_{v_j\in C}d_j\geq 4$ and $\left|C\right|\bmod 2=0$ for each cycle $C\in \C_{\bmod2}$.

Then, our algorithm uses this mod-$2$-cycle cover $\C_{\bmod2}$ to obtain a feasible solution for unsplittable 4-CVRP.
For different types of cycles in $\C_{\bmod2}$, we have different operations. Consider a cycle $C\in\C_{\bmod2}$. Let $x$ be the sum of the demand of all 1-demand and 2-demand customers in $C$.
We classify the cycles in $\C_{\bmod2}$ into several types $\C_{<i>}$ according to the value of $x$. We use $V_i$ to denote the set of customers contained in a cycle in $\C_{<i>}$ and define $\Delta_i= \sum_{v_j\in V_i} d_j w(v_0, v_j)$. Note that we have $w(\C_{<i>})\leq2\Delta_i$ by the triangle inequality. Moreover, by the proof of Theorem~\ref{EXITP}, for any $g_1$ and $g_2$ with $0<g_1\leq g_2\leq1$, we have
\begin{equation}\label{freeq}
2g_1\Delta_i+(1-g_1)w(\C_{<i>})\leq2g_2\Delta_i+(1-g_2)w(\C_{<i>}).
\end{equation}
We may use the above inequality frequently. 

Now, we are ready to introduce our algorithm.

\textbf{Type 0:} The set of the cycles with $x=0$, denoted by $\C_{<0>}$. For each cycle $C\in\C_{<0>}$, there are no 1-demand or 2-demand customers, and we assign a trivial tour for each 3-demand customer in $C$. The total weight of these trivial tours is $\frac{2}{3}\Delta_0\leq \frac{3}{4}\Delta_0+\frac{5}{8}w(\C_{<0>})$.

\textbf{Type 1:} The set of the cycles with $x\bmod 3=0$ and $x\geq 3$, denoted by $\C_{<1>}$.
For each cycle $C\in\C_{<1>}$, we assign a trivial tour for each 3-demand customer in $C$ and then obtain a cycle $C'$ from $C$ by shortcutting all 3-demand customers.
The set of cycles $C'$ is denoted by $\C'_{<1>}$. Then, we apply EX-UITP on $\C'_{<1>}$.

We analyze the weight of the tours assigned in this step. Let $V'_1$ be the set of 1-demand and 2-demand customers in $V_1$ and $V''_1=V_1\setminus V_1'$ be the set of 3-demand customers in $V_1$. Let $\Delta'_1= \sum_{v_j\in V'_1} d_j w(v_0, v_j)$ and $\Delta''_1= \sum_{v_j\in V''_1}d_jw(v_0, v_j)$. Note that for each cycle $C'$ in $\C'_{<1>}$ we have $\size{C'}=x\geq3$.
For the tours computed by EX-UITP on $\C'_{<1>}$, by Lemma~\ref{EXUITP}, the total weight is at most $\frac{2}{3}\Delta'_1+ \frac{2}{3}w(\C'_{<1>})\leq\frac{3}{4}\Delta'_1+ \frac{5}{8}w(\C'_{<1>})\leq \frac{3}{4}\Delta'_1+ \frac{5}{8}w(\C_{<1>})$, where the first inequality follows from (\ref{freeq}) and the second from the triangle inequality.
For the trivial tours visiting the 3-demand customers, the total weight is at most $\frac{2}{3}\Delta''_1\leq\frac{3}{4}\Delta''_1$.
In total, the weight is at most $\frac{3}{4}\Delta'_1+\frac{3}{4}\Delta''_1+ \frac{5}{8}w(\C_{<1>})=\frac{3}{4}\Delta_1+\frac{5}{8}w(\C_{<1>})$.

\textbf{Type 2:} The set of the cycles with $x\bmod 3=1$ and $x\geq 16$, denoted by $\C_{<2>}$.
For each cycle $C\in\C_{<2>}$, we assign a trivial tour for each 3-demand customer in $C$ and then obtain a cycle $C'$ from $C$ by shortcutting all 3-demand customers.
The set of cycles $C'$ is denoted by $\C'_{<2>}$. Then, we apply EX-UITP on $\C'_{<2>}$.

By using an analysis similar to the one for the type 1 case, we can obtain that the weight of the tours in this step is at most $\frac{3}{4}\Delta_2+\frac{5}{8}w(\C_{<2>})$.

\textbf{Type 3:} The set of the cycles with $x\bmod 3=2$ and $x\geq 8$, denoted by $\C_{<3>}$.
For each cycle $C\in\C_{<3>}$, we assign a trivial tour for each 3-demand customer in $C$ and then obtain a cycle $C'$ from $C$ by shortcutting all 3-demand customers.
The set of cycles $C'$ is denoted by $\C'_{<3>}$. Then, we apply EX-UITP on $\C'_{<3>}$.

By using an analysis similar to the one for the type 1 case, we can obtain that the weight of the tours in this step is at most $\frac{3}{4}\Delta_3+ \frac{5}{8}w(\C_{<3>})$.

We still have several cases, where $x$ is a small integer.

\textbf{Type 4:} The set of the cycles with $x=1$, denoted by $\C_{<4>}$. Each cycle $C\in\C_{<4>}$  contains exactly one 1-demand customer $u$ and at least one 3-demand customer $u'$ since the total demand in the cycle is even.
We assign a tour $T$ visiting only the two customers $u$ and $u'$ and a trivial tour for each 3-demand customer in $C$ except $u'$. Then, we can obtain a cycle $C'$ from $C$ by shortcutting the 3-demand customers except $u'$. Let $V'_4$ be the set of customers contained in $T$ and let $\Delta'_4= \sum_{v_j\in V'_4} d_j w(v_0, v_j)$.
Note that $\size{C'}=d_u+d_{u'}=1+3=4$. By Lemma~\ref{rule1}, we can assign a tour $T$ with weight at most $\frac{1}{2}\Delta'_4+\frac{3}{4}w(C')\leq\frac{3}{4}\Delta'_4+\frac{5}{8}w(C')\leq\frac{3}{4}\Delta'_4+\frac{5}{8}w(C)$, where the first inequality follows from (\ref{freeq}) and the second from the triangle inequality.
Let $V''_4=C\setminus V'_4$ and $\Delta''_4=\sum_{v_j\in V''_4} d_j w(v_0, v_j)$.
For the trivial tours visiting the other 3-demand customers in $V''_4$, the total weight is at most $\frac{2}{3}\Delta''_4\leq\frac{3}{4}\Delta''_4$.
Thus, the weight of the tours in this step is at most $\frac{3}{4}\Delta_4+\frac{5}{8}w(\C_{<4>})$.

\textbf{Type 5:} The set of the cycles with $x=4$, denoted by $\C_{<5>}$. For each cycle $C\in\C_{<5>}$, we assign a trivial tour for each 3-demand customer in $C$ and then obtain a cycle $C'$ from $C$ by shortcutting all 3-demand customers. Then, we assign the last tour $T$ according to $C'$.
Let $V'_5$ be the set of customers contained in $T$ and let $\Delta'_5= \sum_{v_j\in V'_5} d_j w(v_0, v_j)$. By Lemma~\ref{rule1}, we can assign the tour $T$ with weight at most $\frac{1}{2}\Delta'_5+\frac{3}{4}w(C')\leq\frac{3}{4}\Delta'_5+\frac{5}{8}w(C')\leq\frac{3}{4}\Delta'_5+\frac{5}{8}w(C)$, where the first inequality follows from (\ref{freeq}) and the second from the triangle inequality. 
Let $V''_5=C\setminus V'_5$ and $\Delta''_5=\sum_{v_j\in V''_5}d_jw(v_0, v_j)$.
For the trivial tours visiting the 3-demand customers in $V''_5$, the total weight is at most $\frac{2}{3}\Delta''_5\leq\frac{3}{4}\Delta''_5$.
Thus, the weight of the tours in this step is at most $\frac{3}{4}\Delta_5+\frac{5}{8}w(\C_{<5>})$.

\textbf{Type 6:} The set of the cycles with $x=7$, denoted by $\C_{<6>}$. For each cycle $C\in\C_{<6>}$, we assign a trivial tour for each 3-demand customer in $C$ and then obtain a cycle $C'$ from $C$ by shortcutting all 3-demand customers. For $C'$, we assign two tours $T_1$ and $T_2$ with capacities 4 and 3 respectively.
Define $V'_6$, $\Delta'_6$, $V''_6$ and $\Delta''_6$ in the similar way as above.
By Lemma~\ref{rule2}, we can assign the tours $T_1$ and $T_2$ with total weight at most $\frac{4}{7}\Delta'_6+\frac{5}{7}w(C')\leq\frac{3}{4}\Delta'_6+\frac{5}{8}w(C')\leq\frac{3}{4}\Delta'_6+\frac{5}{8}w(C)$, where the first inequality follows from (\ref{freeq}) and the second from the triangle inequality.
For the trivial tours visiting the 3-demand customers in $V''_6$, the total weight is at most $\frac{2}{3}\Delta''_6\leq\frac{3}{4}\Delta''_6$.
Thus, the weight of the tours in this step is at most $\frac{3}{4}\Delta_6+\frac{5}{8}w(\C_{<6>})$.

\textbf{Type 7:} The set of the cycles with $x=10$, denoted by $\C_{<7>}$. For each cycle $C\in\C_{<7>}$, we assign a trivial tour for each 3-demand customer in $C$ and then obtain a cycle $C'$ from $C$ by shortcutting all 3-demand customers. For $C'$, we assign three tours $T_1, T_2$ and $T_3$ with capacities 4, 3 and 3, respectively. Define $V'_7$, $\Delta'_7$, $V''_7$ and $\Delta''_7$ in the similar way as above.
By Lemma~\ref{rule2}, we can assign the tours $T_1$, $T_2$ and $T_3$ with total weight at most $\frac{3}{5}\Delta'_7+\frac{7}{10}w(C')\leq \frac{3}{4}\Delta'_7+\frac{5}{8}w(C')\leq\frac{3}{4}\Delta'_7+\frac{5}{8}w(C)$, where the first inequality follows from (\ref{freeq}) and the second from the triangle inequality. For the trivial tours visiting the 3-demand customers, the total weight is at most $\frac{2}{3}\Delta''_7\leq\frac{3}{4}\Delta''_7$.
Thus, the weight of the tours in this step is at most $\frac{3}{4}\Delta_7+\frac{5}{8}w(\C_{<7>})$.

\textbf{Type 8:} The set of the cycles with $x=5$, denoted by $\C_{<8>}$.
Each cycle $C\in\C_{<8>}$  contains at least one 3-demand customer $u'$ because the total demand in the cycle $C$ is even. We assign a trivial tour for each 3-demand customer in $C$ except $u'$ and then obtain a cycle $C'$ by shortcutting the 3-demand customers having been delivered. Define $V'_8$, $\Delta'_8$, $V''_8$ and $\Delta''_8$ in the similar way as above. For $C'$, the total demand of which is 5+3=8, by Lemma~\ref{rule3}, we can obtain a solution on $C'$ with weight at most $\frac{3}{4}\Delta'_{8}+\frac{5}{8}w(C')\leq\frac{3}{4}\Delta'_{8}+\frac{5}{8}w(C')\leq\frac{3}{4}\Delta'_{8}+\frac{5}{8}w(C)$, where the first inequality follows from (\ref{freeq}) and the second from the triangle inequality. For the trivial tours visiting the 3-demand customers, the total weight is at most $\frac{2}{3}\Delta''_{8}\leq\frac{3}{4}\Delta''_{8}$.
Thus, the weight of the tours in this step is at most $\frac{3}{4}\Delta_{8}+\frac{5}{8}w(\C_{<8>})$.

\textbf{Type 9:} The set of the cycles with $x=13$, denoted by $\C_{<9>}$.
Each cycle $C\in\C_{<9>}$  contains at least one 3-demand customer $u'$ because the total demand in the cycle $C$ is even. We assign a trivial tour for each 3-demand customer in $C$ except $u'$ and then obtain a cycle $C'$ containing only all 1-demand and 2-demand customers and $u'$ by shortcutting the 3-demand customers having been delivered. Define $V'_9$, $\Delta'_9$, $V''_9$ and $\Delta''_9$ in the similar way as above. For $C'$, the total demand of which is 13+3=16, by Lemma~\ref{rule3}, we can obtain a solution on $C'$ with weight at most $\frac{3}{4}\Delta'_9+\frac{5}{8}w(C')\leq\frac{3}{4}\Delta'_9+\frac{5}{8}w(C)$ by the triangle inequality. For the trivial tours visiting the 3-demand customers, the total weight is at most $\frac{2}{3}\Delta''_9\leq\frac{3}{4}\Delta''_9$.
Thus, the weight of the tours in this step is at most $\frac{3}{4}\Delta_9+\frac{5}{8}w(\C_{<9>})$.

\textbf{Type 10:} The set of the cycles with $x=2$, denoted by $\C_{<10>}$.
Each cycle $C\in\C_{<10>}$  contains one 2-demand customer $u$ or two 1-demand customers $u_1$ and $u_2$ and at least two 3-demand customers $u'$ and $u''$ because the total demand in the cycle $C$ is even and at least 4. We assign a trivial tour for each 3-demand customer in $C$ except $u'$ and $u''$ and then obtain a cycle $C'$ containing only $\{u, u', u''\}$ (or $\{u_1, u_2, u', u''\}$) by shortcutting the 3-demand customers having been delivered. For $C'$, we assign three tours $T_1, T_2$ and $T_3$ with capacities 3, 3 and 2, respectively.
Define $V'_{10}$, $\Delta'_{10}$, $V''_{10}$ and $\Delta''_{10}$ in the similar way as above.
Note that $\size{C'}=8$. By considering all cases of the positions of these three customers on $C'$, it can be verified that we can assign the tours $T_1$, $T_2$ and $T_3$ with total weight at most $\frac{3}{4}\Delta'_{10}+\frac{5}{8}w(C')\leq\frac{3}{4}\Delta'_{10}+\frac{5}{8}w(C)$ by the triangle inequality.
For the trivial tours visiting the other 3-demand and 4-demand customers in $V''_{10}$, the total weight is at most $\frac{2}{3}\Delta''_{10}\leq\frac{3}{4}\Delta''_{10}$.
Thus, by summing all cycles in $\C_{<10>}$, the weight of the tours in this step is at most $\frac{3}{4}\Delta_{10}+\frac{5}{8}w(\C_{<10>})$.

Recall that $w(\C_{\bmod2})\leq w(\C^*_4)$. We have the following lemma.
\begin{lemma}
Given the mod-$2$-cycle cover $\C_{\bmod2}$, for unsplittable 4-CVRP, there is a polynomial-time algorithm that computes an itinerary $I$ with $w(I)\leq \frac{3}{4}\Delta+\frac{5}{8}w(\C_{\bmod2})\leq\frac{3}{4}\Delta+\frac{5}{8}w(\C^*_4)$.
\end{lemma}

By Lemmas~\ref{lb-delta} and~\ref{lb-cyclepacking}, we have
\[
\frac{3}{4}\Delta+\frac{5}{8}w(\C^*_4)\leq\frac{3}{4}\lrA{\ra+\frac{1}{2}}w(I^*)+\frac{5}{8}\min\{2(1-\ra),1\}w(I^*)\leq\frac{7}{4}w(I^*),
\]
where we have $\ra=\frac{1}{2}$ in the worst case. 

Therefore, we have the following result.
\begin{theorem}
For unsplittable 4-CVRP, there is a $\frac{7}{4}$-approximation algorithm.
\end{theorem}

As mentioned before, the $\frac{7}{4}$-approximation algorithm is based on the idea (a good mod-2-cycle cover) of the previous $\frac{5}{3}$-approximation algorithm for unit-demand 4-CVRP. It remains unknown whether this result can be improved using the idea (a good 4-path cover) of the previous $\frac{3}{2}$-approximation algorithm for unit-demand 4-CVRP.

\section{An Improvement for Unsplittable $k$-CVRP}\label{Sec-Split-Initial}
In this section, we improve unsplittable $k$-CVRP with larger $k$.
Given an $\alpha$-approximation algorithm for metric TSP, where $1\leq\alpha\leq\frac{3}{2}$, we will obtain an approximation ratio of $\alpha+1+\ln2-\frac{2\alpha}{k}-\Theta(\frac{1}{k})$. Our algorithm uses several techniques. The first one is the LP-based technique used in~\cite{uncvrp}.

\subsection{An LP-Based Algorithm}
When $k$ is a fixed integer, there are at most $n^k=n^{O(1)}$ different tours, denoted by $\T$.
The idea of the LP-based algorithm is to define a variable $x_T$ for each tour $T\in\T$, and then solve the following LP in $n^{O(1)}$ time.
\begin{alignat}{3}
\text{minimize} & \quad & \sum_{T\in\T} w(T)\cdot x_T \notag\\
\text{subject to} && \sum_{\substack{T\in \T:\\ v\in T}}x_T \geq\   & 1, \quad && \forall\  v \in V, \notag\\
&& x_T \geq\   & 0, \quad && \forall\  T \in \T. \notag
\end{alignat}

Recall that $k$-CVRP can be reduced to the minimum weight $k$-set cover problem. The above LP follows the same framework as the LP relaxation for the weighted $k$-set cover problem. Each vertex $v\in V$ is regarded as an element. Each tour $T\in\T$ is regarded as a set containing all vertices $v\in T$, with weight $w(T)$. 
The goal is to find a minimum weight collection of sets (tours) to cover all elements (vertices).

Next, we give a simple analysis of the LP-based algorithm. The details can be seen in~\cite{uncvrp}.

Using a constant $\gamma\geq0$, we make a randomized rounding independently for each tour with a probability of $\min\{1,\gamma\cdot x_T\}$.
The chosen tours will form a partial itinerary $I'$ with an expected weight of $\mathbb{E}[w(I')]\leq\gamma\cdot w(I^*)$. The customer in $I'$ is denoted by $V'$. Then, for the undelivered customers in $V''=V\setminus V'$, we use UITP to compute an itinerary $I''$ such that $\mathbb{E}[w(I'')]\leq e^{-\gamma}\cdot\frac{4}{k}\Delta+(1-\frac{2}{k})w(H'')$ where $H''$ is a Hamiltonian cycle on $V''\cup \{v_0\}$. Note that given a Hamiltonian cycle $H$ on $V\cup\{v_0\}$, we can obtain a Hamiltonian cycle $H''$ on $V''\cup \{v_0\}$ by shortcutting such that $w(H'')\leq w(H)$. We have
$
\mathbb{E}[I]=\mathbb{E}[I'\cup I'']\leq\gamma \cdot w(I^*)+e^{-\gamma}\cdot \frac{4}{k}\Delta+(1-\frac{2}{k})w(H).
$

Combining the LP-based method with the refined AG-UITP, we obtain an LP-UITP algorithm.
\begin{lemma}\label{LPUITP}
Given a Hamiltonian cycle $H$ on $V\cup\{v_0\}$, for unsplittable $k$-CVRP with fixed $k\geq 3$ and any constant $\gamma\geq0$, LP-UITP in $n^{O(k)}=n^{O(1)}$ time outputs a solution with weight at most
$
\gamma \cdot w(I^*)+e^{-\gamma}\cdot \frac{2}{\floor{ k/2}+1}\Delta+\lra{1-\frac{1}{\floor{ k/2}+1}}w(H).
$
\end{lemma}

Note that LP-UITP reduces to the refined AG-UITP when $\gamma=0$.

\subsection{Trade-off Between Two Results}
Next, we show that by using Lemma~\ref{LPUITP} and two initial Hamiltonian cycles, we can obtain the desired approximation ratio by doing a trade-off.

Our algorithms follows the ideas in Section~\ref{sec-split-initial}. 
We use two initial Hamiltonian cycles $H$ and $H_{CS}$ on $V\cup\{v_0\}$ to run LP-UITP and return the better found solution, where $H$ is an $\alpha$-approximate Hamiltonian cycle and $H_{CS}$ is the Hamiltonian cycle obtained by the Christofides-Serdyukov algorithm. 
Note that we do not consider the mod-$k$-cycle cover used in Section~\ref{sec-split-initial} since for unsplittable $k$-CVRP, the algorithms based on cycle covers become complicated to analyze (as we have shown in Section~\ref{Sec-Unsplit-34}).

The following theorem shows the approximation ratio of our algorithm.

\begin{theorem}\label{res-intitial-unsplit}
Given an $\alpha$-approximate Hamiltonian cycle $H$ on $V\cup\{v_0\}$ with $1\leq\alpha\leq\frac{3}{2}$, for unsplittable $k$-CVRP with fixed $k\geq3$, there is an approximation algorithm such that
\begin{itemize}
\item If $3\leq k\leq 5$, the approximation ratio is $\frac{2\floor{k/2}+1}{\floor{k/2}+1}+\ln\frac{k}{\floor{k/2}+1}$;
\item If $k=6$ and $\frac{7}{6}\leq\alpha\leq\frac{3}{2}$, the approximation ratio is $\frac{15}{8}+\ln\frac{4}{3}<2.163$;
\item If $k=7$ and $\frac{17}{12}\leq\alpha\leq\frac{3}{2}$, the approximation ratio is $\frac{33}{16}+\ln\frac{4}{3}<2.351$;
\item Otherwise, the approximation ratio is $\frac{(\alpha+1)\floor{k/2}+1}{\floor{k/2}+1}+\ln\frac{k-4(\alpha-1)}{\floor{k/2}+1}$.
\end{itemize}
\end{theorem}
\begin{proof}
Given two Hamiltonian cycles $H$ and $H_{CS}$, we can use LP-UITP to compute two solutions. We will select the better one.
In the worst case, we assume $\frac{2}{\floor{k/2}+1}\lra{\ra+\frac{k-2}{2}}\geq1$.
We will show that the approximation ratio will be better if $\frac{2}{\floor{k/2}+1}\lra{\ra+\frac{k-2}{2}}<1$.

Using the Hamiltonian cycle $H$, by Lemma~\ref{LPUITP}, LP-UITP computes an itinerary $I_1$ with weight at most
\begin{align*}
&\min_{\gamma\geq0}\lrC{\gamma \cdot w(I^*)+e^{-\gamma}\cdot \frac{2}{\floor{k/2}+1}\Delta+\lrA{1-\frac{1}{\floor{ k/2}+1}}w(H)}\\
&\leq\min_{\gamma\geq0}\lrC{\gamma+e^{-\gamma}\cdot \frac{2}{\floor{ k/2}+1}\lrA{\ra+\frac{k-2}{2}}+\lrA{1-\frac{1}{\floor{ k/2}+1}}\alpha }w(I^*)\\
&\leq\lrA{1+\ln\frac{2}{\floor{ k/2}+1}+\ln\lrA{\ra+\frac{k-2}{2}}+\lrA{1-\frac{1}{\floor{ k/2}+1}}\alpha}w(I^*),
\end{align*}
where the first inequality follows from Lemmas~\ref{lb-tsp} and \ref{lb-delta}, and the second is obtained by setting $\gamma=\ln\lra{\frac{2}{\floor{k/2}+1}\lra{\ra+\frac{k-2}{2}}}$ since $\frac{2}{\floor{k/2}+1}\lra{\ra+\frac{k-2}{2}}\geq 1$.
Thus, the approximation ratio is $1+\ln\frac{2}{\floor{ k/2}+1}+\ln\lra{\ra+\frac{k-2}{2}}+\lra{1-\frac{1}{\floor{ k/2}+1}}\alpha$.

Using the Hamiltonian cycle $H_{CS}$, by Lemma~\ref{LPUITP}, LP-UITP computes an itinerary $I_2$ with weight at most
\begin{align*}
&\min_{\gamma\geq0}\lrC{\gamma \cdot w(I^*)+e^{-\gamma}\cdot \frac{2}{\floor{ k/2}+1}\Delta+\lrA{1-\frac{1}{\floor{ k/2}+1}}w(H_{CS})}\\
&\leq\min_{\gamma\geq0}\lrC{\gamma+e^{-\gamma}\cdot \frac{2}{\floor{ k/2}+1}\lrA{\ra+\frac{k-2}{2}}+\lrA{1-\frac{1}{\floor{ k/2}+1}}\frac{3-\ra}{2} }w(I^*)\\
&\leq\lrA{1+\ln\frac{2}{\floor{ k/2}+1}+\ln\lrA{\ra+\frac{k-2}{2}}+\lrA{1-\frac{1}{\floor{ k/2}+1}}\frac{3-\ra}{2}}w(I^*).
\end{align*}
where the first inequality follows from Lemmas~\ref{lb-tsp}, \ref{lb-delta}, and \ref{lb-chris}, and the second is obtained by setting $\gamma=\ln\lra{\frac{2}{\floor{k/2}+1}\lra{\ra+\frac{k-2}{2}}}$ since $\frac{2}{\floor{k/2}+1}\lra{\ra+\frac{k-2}{2}}\geq 1$. Thus, the approximation ratio is $1+\ln\frac{2}{\floor{k/2}+1}+\ln\lra{\ra+\frac{k-2}{2}}+\lra{1-\frac{1}{\floor{ k/2}+1}}\frac{3-\ra}{2}$.

By returning the better found solution, the results can be briefly summarized as follows.

\textbf{Case~1: $3\leq k\leq 5$.} We have $\ra=1$ in the worst case, and then by setting $\gamma=\ln\lra{\frac{2}{\floor{k/2}+1}\lra{\ra+\frac{k-2}{2}}}=\ln\frac{k}{\floor{k/2}+1}$, the approximation ratio is
\[
1+\ln\frac{k}{\floor{ k/2}+1}+\lrA{1-\frac{1}{\floor{ k/2}+1}}=\frac{2\floor{k/2}+1}{\floor{k/2}+1}+\ln\frac{k}{\floor{k/2}+1}.
\]

\textbf{Case~2: $k=6$ and $\frac{7}{6}\leq\alpha\leq\frac{3}{2}$.} We have $\ra=\frac{2}{3}$, and then by setting $\gamma=\ln\lra{\frac{2}{\floor{k/2}+1}\lra{\ra+\frac{k-2}{2}}}=\ln\frac{4}{3}$, the approximation ratio is $\frac{15}{8}+\ln\frac{4}{3}<2.163$.

\textbf{Case~3: $k=7$ and $\frac{17}{12}\leq\alpha\leq\frac{3}{2}$.} We have $\chi=\frac{1}{6}$, and then by setting $\gamma=\ln\lra{\frac{2}{\floor{k/2}+1}\lra{\ra+\frac{k-2}{2}}}=\ln\frac{4}{3}$, the approximation ratio is $\frac{33}{16}+\ln\frac{4}{3}<2.351$.

\textbf{Case~4: $k=6$ and $1\leq\alpha\leq\frac{7}{6}$, $k=7$ and $1\leq\alpha\leq\frac{17}{12}$, or $k\geq 8$.} We have $\ra=3-2\alpha$ in the worst case, and then by setting $\gamma=\ln\lra{\frac{2}{\floor{k/2}+1}\lra{\ra+\frac{k-2}{2}}}=\ln\frac{k-4(\alpha-1)}{\floor{k/2}+1}$, the approximation ratio is
\[
1+\ln\frac{k-4(\alpha-1)}{\floor{ k/2}+1}+\lrA{1-\frac{1}{\floor{k/2}+1}}\alpha=\frac{(\alpha+1)\floor{k/2}+1}{\floor{k/2}+1}
+\ln\frac{k-4(\alpha-1)}{\floor{k/2}+1}.
\]

Now, we consider the case $\frac{2}{\floor{ k/2}+1}\lrA{\ra+\frac{k-2}{2}}<1$. It is sufficient to show that we can obtain a better approximation ratio in this case.

On one hand, setting $\gamma=0$, the solution $I_1$ has weight at most
\begin{align*}
&\frac{2}{\floor{ k/2}+1}\Delta+\lrA{1-\frac{1}{\floor{ k/2}+1}}w(H)\\
&\leq\lrA{\frac{2}{\floor{k/2}+1}\lrA{\ra+\frac{k-2}{2}}+\lrA{1-\frac{1}{\floor{k/2}+1}}\alpha} w(I^*)\\
&\leq\lrA{1+\lrA{1-\frac{1}{\floor{k/2}+1}}\alpha}w(I^*),
\end{align*}
where the first inequality follows from Lemmas~\ref{lb-tsp} and~\ref{lb-delta}, and the second from $\frac{2}{\floor{k/2}+1}\lrA{\ra+\frac{k-2}{2}}<1$ by definition. Thus, the current approximation ratio is better than the previous approximation ratio in Case 4. 
Moreover, when $k=6$ or $k=7$, the current approximation ratio is at most $1+(1-\frac{1}{4})\cdot\frac{3}{2}<2.125$, which is also better than the previous approximation ratios in Cases 2 and 3.

On the other hand, since $\frac{2}{\floor{k/2}+1}\lrA{\ra+\frac{k-2}{2}}<1$ by definition, we can obtain $0\leq \ra\leq \frac{\floor{k/2}-k+3}{2}\leq\frac{1}{2}$ when $3\leq k\leq5$. Then, setting $\gamma=0$, the solution $I_2$ has weight at most
\begin{align*}
&\frac{2}{\floor{ k/2}+1}\Delta+\lrA{1-\frac{1}{\floor{ k/2}+1}}w(H_{CS})\\
&\leq\lrA{\frac{2}{\floor{k/2}+1}\lrA{\ra+\frac{k-2}{2}}+\lrA{1-\frac{1}{\floor{k/2}+1}}\frac{3-\ra}{2}}w(I^*)\\
&\leq\lrA{\frac{5}{4}+\frac{k-9/4}{\floor{k/2}+1}}w(I^*),
\end{align*}
where the first inequality follows from Lemmas~\ref{lb-delta} and \ref{lb-chris}, and the second from $\ra\leq\frac{1}{2}$, $3\leq k\leq 5$, and $(\frac{2}{\floor{k/2}+1}-\frac{1}{2}(1-\frac{1}{\floor{k/2}+1}))\ra\leq(\frac{2}{\floor{k/2}+1}-\frac{1}{2}(1-\frac{1}{\floor{k/2}+1}))\cdot\frac{1}{2}$.
Therefore, for $k = 3$, $k = 4$, and $k = 5$, the current approximation ratios are 1.625, 1.834, and 2.167, respectively, which are also better than the previous approximation ratios of 1.906, 1.955, and 2.178 in Case 1.
\end{proof}

Note that the trade-off does not apply when $\alpha = 1$; in this case, the result can already be obtained using an optimal Hamiltonian cycle, and incorporating the Hamiltonian cycle $H_{CS}$ yields no further improvement.
However, for unsplittable $k$-CVRP with $k\geq6$, by Theorem~\ref{res-intitial-unsplit}, the approximation ratio is $\frac{2\floor{k/2}+1}{\floor{k/2}+1}+\ln\frac{k}{\floor{k/2}+1}$ which already achieves $\alpha+1+\ln2-\frac{2\alpha}{k}-\Theta(\frac{1}{k})$. 

We also remark that for $3\leq k\leq 5$, we have $\ra=1$ in the worst case, and then by Lemma~\ref{lb-cyclepacking}, we may use a good mod-$k$-cycle cover to obtain some further improvements. We will use unsplittable $5$-CVRP as an example to illustrate this in Section~\ref{Sec-Unsplit-5}.

\section{A Further Improvement for Unsplittable $k$-CVRP}\label{Sec-Unsplit-Final}
In the previous section, the approximation ratio of unsplittable $k$-CVRP is $\frac{5\floor{k/2}+2}{2\floor{k/2}+2}+\ln\frac{k-2}{\floor{k/2}+1}=\frac{5}{2}+\ln2-\Theta(\frac{1}{k})$ for $k\geq 8$ under $\alpha=\frac{3}{2}$ for metric TSP.
In this section, we further improve the approximation ratio to less than $\frac{5}{2}+\ln2-\sqrt{\frac{2}{k}}$. The main idea is to use the refined lower bounds based on $\MST$ and $\Delta$, provided in Section~\ref{sec-split-final}. However, the analysis is slightly different.

Our algorithm is simply to call LP-UITP on the Hamiltonian cycle $H_{CS}$ on $V\cup\{v_0\}$.
By Lemma~\ref{LPUITP}, for unsplittable $k$-CVRP with any constant $\gamma\geq0$ and $k\geq3$, LP-UITP computes an itinerary $I$ with
$
w(I)\leq\gamma \cdot w(I^*)+e^{-\gamma}\cdot \frac{2}{\floor{ k/2}+1}\Delta+\lra{1-\frac{1}{\floor{ k/2}+1}}w(H_{CS}).
$

We first show a simple analysis.

By setting $\gamma=\ln2$, we can obtain that $w(I)\leq\ln2\cdot w(I^*)+\frac{1}{\floor{k/2}+1}\Delta+(1-\frac{1}{\floor{k/2}+1})w(H_{CS})\leq\ln2\cdot w(I^*)+\frac{1}{\floor{k/2}+1}\Delta+(1-\frac{1}{2}\cdot\frac{1}{\floor{k/2}+1})w(H_{CS})\leq \ln2\cdot w(I^*)+\frac{2}{k}\Delta+(1-\frac{1}{k})w(H_{CS})$, where the last inequality follows from $w(H_{CS})\leq 2\Delta$ by the triangle inequality.
Note that the lower bounds of unit-demand $k$-CVRP are also the lower bounds of unsplittable $k$-CVRP.
Thus, by (\ref{Choris+ITP}) and the proof of Theorem~\ref{res-split}, we have $\frac{2}{k}\Delta+(1-\frac{1}{k})w(H_{CS})<(\frac{5}{2}-\sqrt{\frac{2}{k}})w(I^*)$. Then, we have $w(I)\leq\ln2\cdot w(I^*)+\frac{2}{k}\Delta+(1-\frac{1}{k})w(H_{CS})<(\frac{5}{2}+\ln2-\sqrt{\frac{2}{k}})w(I^*)$.
Therefore, a simple analysis already leads to an approximation ratio of less than $\frac{5}{2}+\ln2-\sqrt{\frac{2}{k}}$.

Next, by carefully calling LP-UITP, as shown in Algorithm~\ref{alg}, we use a more refined analysis to obtain a better result.

\begin{algorithm}[t!]
\caption{The algorithm for unsplittable $k$-CVRP}
\label{alg}
\small
\vspace*{2mm}
\textbf{Input:} An instance of unsplittable $k$-CVRP with any $k\geq 3$. \\
\textbf{Output:} A solution for unsplittable $k$-CVRP.

\begin{algorithmic}[1]
\State Call LP-UITP with $\gamma=0$ to obtain a solution $I_1$.
\State Let $C=\frac{\floor{k/2}}{\floor{ k/2}+1}\cdot\frac{1}{2}$.
\If{$\left\lfloor\frac{-(C+1)+\sqrt{(C+1)^2+4kC}}{2}\right\rfloor \neq \left\lfloor\frac{-(C+1)+\sqrt{(C+1)^2+4(k+1)C}}{2}\right\rfloor$}
\State Let $\gamma^*=\ln\lrA{\frac{k-2z^*+(z^*)^2x^*+z^*x^*}{\floor{k/2}+1}}$, where $z^*=\Floor{\frac{-(C+1)+\sqrt{(C+1)^2+4(k+1)C}}{2}}$ and $x^*=\frac{1}{C}-\frac{k-2z^*}{(z^*)^2+z^*}$.
\Else
\State Let $z_0=\frac{\sqrt{\floor{k/2}^2+8\floor{k/2}(\floor{k/2}+1)(k+1)}-\floor{k/2}}{4(\floor{k/2}+1)}$.
\State Let $z^*=\ceil{z_0}$ if $g(\ceil{z_0})\geq g(\floor{z_0})$, and $z^*=\floor{z_0}$ otherwise, where $g(z)=1+\ln\left(\frac{k-z+1}{\floor{k/2}+1}\right)+\frac{\floor{k/2}}{\floor{ k/2}+1}\cdot\frac{3-\frac{1}{z}}{2}$.  
\State Let $\gamma^*=\ln\lrA{\frac{k-z^*+1}{\floor{k/2}+1}}$.
\EndIf
\State Call LP-UITP with $\gamma=\gamma^*$ to obtain a solution $I_2$.
\State \Return the better solution from $\{I_1,I_2\}$, i.e., $\arg\min\{w(I_1),w(I_2)\}$.
\end{algorithmic}
\end{algorithm}

\begin{theorem}\label{res-unsplit}
Given fixed $k\geq 3$, for unsplittable $k$-CVRP, the approximation ratio of LP-UITP is shown as follows.
\begin{itemize}
\item If $\floor{\frac{-(C+1)+\sqrt{(C+1)^2+4kC}}{2}}\neq \floor{\frac{-(C+1)+\sqrt{(C+1)^2+4(k+1)C}}{2}}$, where $C=\frac{\floor{k/2}}{\floor{ k/2}+1}\cdot\frac{1}{2}$, the approximation ratio is at most $h(x^*)$, where $h(x)=1+\ln\lrA{\frac{k-2\floor{\frac{1}{x}}+\floor{\frac{1}{x}}^2x+\floor{\frac{1}{x}}x}{\floor{k/2}+1}}+\frac{\floor{k/2}}{\floor{ k/2}+1}\cdot\frac{3-x}{2}$, $x^*=\frac{1}{C}-\frac{k-2z^*}{(z^*)^2+z^*}$, and $z^*=\floor{\frac{1}{x^*}}=\Floor{\frac{-(C+1)+\sqrt{(C+1)^2+4(k+1)C}}{2}}$;
\item Otherwise, the approximation ratio is at most $\max\{g(\ceil{z_0}), g(\floor{z_0})\}$, where $g(z)=1+\ln\left(\frac{k-z+1}{\floor{k/2}+1}\right)+\frac{\floor{k/2}}{\floor{ k/2}+1}\cdot\frac{3-\frac{1}{z}}{2}$, and $z_0=\frac{\sqrt{\floor{k/2}^2+8\floor{k/2}(\floor{k/2}+1)(k+1)}-\floor{k/2}}{4(\floor{k/2}+1)}$.
\end{itemize}
\end{theorem}
\begin{proof}
The details are put in Section~\ref{proof-2}.
\end{proof}

We remark that when $3 \leq k \leq 6$ (resp., $k>6$), the approximation ratio in Theorem~\ref{res-unsplit} matches (resp., is strictly better than) the result under $\alpha=\frac{3}{2}$ in Theorem~\ref{res-intitial-unsplit}. The details can be found in Table~\ref{res-2}

\section{A Further Improvement for Unsplittable 5-CVRP}\label{Sec-Unsplit-5}
In this section, we further improve the approximation ratio from $\frac{5}{3}+\ln\frac{5}{3}< 2.178$ (in the previous section) to $2.157$ for unsplittable 5-CVRP. Different from our algorithms for unsplittable 3-CVRP and 4-CVRP, we will use LP-UITP for unsplittable 5-CVRP.
This is also the reason why we do not put our algorithms for unsplittable 3-CVRP, 4-CVRP and 5-CVRP together.

The previous approximation ratio is obtained by using the Hamiltonian cycle $H_{CS}$ under the case $\ra=1$. In this case, an $O(1)$-approximate mod-$k$-cycle cover has a better performance than $H_{CS}$ by Lemmas~\ref{lb-chris} and \ref{lb-cyclepacking}. Hence, for unsplittable 5-CVRP, we use two algorithms. In the first algorithm, we construct a mod-5-cycle cover via the 2-approximation algorithm in~\cite{GoemansW95} and then obtain a feasible solution from it.
The second algorithm calls LP-UITP on the Hamiltonian cycle $H_{CS}$.

Assumption~\ref{ass3} guarantees that each customer's demand is either 1, 2, 3 or 4.
Our first algorithm will first compute a cycle cover $\C_{\bmod 5}$, where ``$\bmod~5$'' means that for each cycle in the cycle cover, the total demand of customers in the cycle is divisible by 5.
We call this cycle cover a mod-$5$-cycle cover. 
It is computed in the following way: we construct a unit-demand instance $G'$ by replacing each customer $v_i$ with $d_i$ unit-demand customers as the same place. Then, we call the 2-approximation algorithm in~\cite{GoemansW95} to compute a mod-5-cycle cover $\C'_{\bmod5}$ in $G[V']$, where $V'$ denotes the set of all customers in $G'$, such that $w(\C'_{\bmod5})\leq 2w(\C^*_{\bmod5})$. Note that $\C'_{\bmod5}$ may not be a mod-5-cycle cover for the original graph $G[V]$. By shortcutting cycles in $\C'_{\bmod5}$, we can obtain a cycle cover $\C_{\bmod5}$ on $V$. Moreover, the total demand of all customers in each cycle in $\C_{\bmod5}$ is divisible by 5. We have that $w(\C_{\bmod5})\leq 2w(\C_{5}^{*})$ by the triangle inequality, and it holds that $\left|C\right|\bmod 5=\sum_{v_j\in C}d_j\bmod 5=0$ for each $C\in \C_{\bmod5}$.

Then, our algorithm uses this mod-$5$-cycle cover $\C_{\bmod5}$ to obtain a feasible solution for unsplittable 5-CVRP.
For different types of cycles in $\C_{\bmod5}$, we have different operations.
Consider a cycle $C$ in $\C_{\bmod5}$. Let $x$ be the sum of the demand of all 1-demand, 2-demand and 3-demand customers in $C$.
We classify the cycles in $\C_{\bmod5}$ into several types $\C_{<i>}$ according to the value of $x$. We use $V_i$ to denote the set of customers contained in a cycle in $\C_{<i>}$ and define $\Delta_i= \sum_{v_j\in V_i} d_j w(v_0, v_j)$.

\textbf{Type 0:} The set of the cycles with $x=0$, denoted by $\C_{<0>}$. For each cycle $C\in\C_{<0>}$, there are no 1-demand or 2-demand customers, and we assign a trivial tour for each 3-demand or 4-demand customer in $C$. The total weight of these trivial tours is at most $\frac{2}{3}\Delta_0\leq \frac{2}{3}\Delta_0+\frac{2}{3}w(\C_{<0>})$.

\textbf{Type 1:} The set of the cycles with $x\geq 6$, denoted by $\C_{<1>}$.
For each cycle $C\in\C_{<1>}$, we assign a trivial tour for each 3-demand or 4-demand customer in $C$ and then obtain a cycle $C'$ from $C$ by shortcutting all 3-demand and 4-demand customers.
The set of cycles $C'$ is denoted by $\C'_{<1>}$. Then, we apply EX-UITP on $\C'_{<1>}$.

We analyze the weight of the tours assigned in this step. Let $V'_1$ be the set of 1-demand and 2-demand customers in $V_1$ and $V''_1=V_1\setminus V_1'$ be the set of 3-demand and 4-demand customers in $V_1$. Let $\Delta'_1= \sum_{v_j\in V'_1}d_j w(v_0, v_j)$ and $\Delta''_1= \sum_{v_j\in V''_1} d_j w(v_0, v_j)$. 
For each cycle $C'$ in $\C'_{<1>}$ we have $\size{C'}=x\geq6$.
For the tours computed by EX-UITP on $\C'_{<1>}$, by Lemma~\ref{EXUITP}, the total weight is at most $\frac{2}{3}\Delta'_1+ \frac{2}{3}w(\C'_{<1>})\leq \frac{2}{3}\Delta'_1+ \frac{2}{3}w(\C_{<1>})$, where the inequality follows from the triangle inequality.
For the trivial tours visiting a 3-demand or 4-demand customer, the total weight is at most $\frac{2}{3}\Delta''_1$.
In total, the weight is at most $\frac{2}{3}\Delta'_1+\frac{2}{3}\Delta''_1+ \frac{2}{3}w(\C_{<1>})=\frac{2}{3}\Delta_1+ \frac{2}{3}w(\C_{<1>})$.

We still have several cases, where $x$ is a small integer.

\textbf{Type 2:} The set of the cycles with $x=1$, denoted by $\C_{<2>}$. Each cycle $C\in\C_{<2>}$ contains exactly one 1-demand customer $u$ and at least one 3-demand (or 4-demand) customer $u'$ since the total demand in the cycle is divisible by 5.
We assign a tour $T$ visiting only the two customers $u$ and $u'$ and a trivial tour for each 3-demand and 4-demand customer in $C$ except $u'$, and then obtain a cycle $C'$ from $C$ by shortcutting the 3-demand and 4-demand customers except $u'$. Let $V'_2$ be the set of customers contained in $T$ and let $\Delta'_2= \sum_{v_j\in V'_2} d_j w(v_0, v_j)$.
Note that $4\leq\size{C'}=d_u+d_{u'}\leq5$. By Lemma~\ref{rule1}, we can assign the tour $T$ with weight at most $\frac{2}{\size{C'}}\Delta'_2+(1-\frac{1}{\size{C'}})w(C')\leq\frac{2}{3}\Delta'_2+\frac{2}{3}w(C')\leq\frac{2}{3}\Delta'_2+\frac{2}{3}w(C)$, where the first inequality follows from $3<\size{C'}\leq 5$ and (\ref{freeq}) and the second from the triangle inequality.
Let $V''_2=C\setminus V'_2$ and $\Delta''_2=\sum_{v_j\in V''_2} d_j w(v_0, v_j)$.
For the trivial tours visiting the other 3-demand and 4-demand customers in $V''_2$, the total weight is at most $\frac{2}{3}\Delta''_2$.
Thus, the weight of the tours in this step is at most $\frac{2}{3}\Delta_2+\frac{2}{3}w(\C_{<2>})$.

\textbf{Type 3:} The set of the cycles with $x=2$, denoted by $\C_{<3>}$. Each cycle $C\in\C_{<3>}$ contains exactly two 1-demand customers $u_1$ and $u_2$ or exactly one 2-demand customer $u$, and at least one 3-demand (or 4-demand) customer $u'$ since the total demand in the cycle is divisible by 5.

If the demand of $u'$ is 3, we assign a single tour $T$ visiting $\{u_1,u_2,u'\}$ (or $\{u,u'\}$) and a trivial tour for each 3-demand and 4-demand customer in $C$ except $u'$, and then obtain a cycle $C'$ from $C$ by shortcutting the 3-demand and 4-demand customers except $u'$. 
Let $V'_3$ be the set of customers contained in $T$ and let $\Delta'_3= \sum_{v_j\in V'_3} d_j w(v_0, v_j)$.
Note that $\size{C'}=5$. By Lemma~\ref{rule1}, we can assign the tour $T$ with weight at most $\frac{2}{5}\Delta'_3+\frac{3}{5}w(C')\leq\frac{2}{3}\Delta'_3+\frac{2}{3}w(C')\leq\frac{2}{3}\Delta'_2+\frac{2}{3}w(C)$, where the first inequality follows from (\ref{freeq}) and the second from the triangle inequality.
Let $V''_3=C\setminus V'_3$ and $\Delta''_3=\sum_{v_j\in V''_3} d_j w(v_0, v_j)$.
For the trivial tours visiting the other 3-demand and 4-demand customers in $V''_3$, the total weight is at most $\frac{2}{3}\Delta''_3$.
Thus, the weight of the tours in this step is at most $\frac{2}{3}\Delta_3+\frac{2}{3}w(\C_{<3>})$.

If the demand of $u'$ is 4, we assign two tours $T_1$ and $T_2$ with capacities $2$ and $4$ visiting $\{u_1,u_2,u'\}$ (or $\{u,u'\}$) and a trivial tour for each 3-demand and 4-demand customer in $C$ except $u'$, and then obtain a cycle $C'$ from $C$ by shortcutting the 3-demand and 4-demand customers except $u'$.
Define $V'_3$, $\Delta'_3$, $V''_3$ and $\Delta''_3$ in the similar way as above.
Note that $\size{C'}=6$. By considering all cases of the positions of these three (or two) customers on $C'$, it can be verified that we can assign the tours $T_1$ and $T_2$ with total weight at most $\frac{2}{3}\Delta'_3+\frac{2}{3}w(C')\leq \frac{2}{3}\Delta'_3+\frac{2}{3}w(C)$ by the triangle inequality.
For the trivial tours visiting the other 3-demand and 4-demand customers in $V''_3$, the total weight is at most $\frac{2}{3}\Delta''_3$.
Thus, the weight of the tours in this step is still at most $\frac{2}{3}\Delta_3+\frac{2}{3}w(\C_{<3>})$.

\textbf{Type 4:} The set of the cycles with $3\leq x\leq 5$, denoted by $\C_{<4>}$. For each cycle $C\in\C_{<4>}$, we assign a trivial tour for each 3-demand or 4-demand customer in $C$ and then obtain a cycle $C'$ from $C$ by shortcutting all 3-demand and 4-demand customers. Then, we assign the last tour $T$ according to $C'$.
Let $V'_4$ be the set of customers contained in $T$ and let $\Delta'_4= \sum_{v_j\in V'_4} d_j w(v_0, v_j)$. By Lemma~\ref{rule1}, we can assign the tour $T$ with weight at most $\frac{2}{x}\Delta'_4+(1-\frac{1}{x})w(C')\leq\frac{2}{3}\Delta'_4+\frac{2}{3}w(C')\leq\frac{2}{3}\Delta'_2+\frac{2}{3}w(C)$, where the first inequality follows from $3\leq x\leq 5$ and \ref{freeq}) and the second from the triangle inequality. Let $V''_4=C\setminus V'_4$ and $\Delta''_4=\sum_{v_j\in V''_4}d_jw(v_0, v_j)$.
For the trivial tours visiting the 3-demand and 4-demand customers in $V''_4$, the total weight is at most $\frac{2}{3}\Delta''_4$.
Thus, the weight of the tours in this step is at most $\frac{2}{3}\Delta_4+\frac{2}{3}w(\C_{<4>})$.

Recall that $w(\C_{\bmod5})\leq 2w(\C^*_{\bmod5})$. We have the following result.
\begin{lemma}\label{alg-mod5}
Given the mod-$5$-cycle cover $\C_{\bmod5}$, for unsplittable 5-CVRP, there is a polynomial-time algorithm that computes an itinerary $I$ with $w(I)\leq \frac{2}{3}\Delta+\frac{2}{3}w(\C_{\bmod5})\leq \frac{2}{3}\Delta+\frac{4}{3}w(\C^*_{\bmod5})$.
\end{lemma}

\begin{theorem}
For unsplittable 5-CVRP, there is a $2.157$-approximation algorithm.
\end{theorem}
\begin{proof}
Using the mod-$5$-cycle cover algorithm, by Lemma~\ref{alg-mod5}, we can obtain an itinerary $I_1$ with
$
w(I_1)\leq \frac{2}{3}\Delta+\frac{4}{3}w(\C_{\bmod 5}^{*}).
$
By Lemmas~\ref{lb-delta} and \ref{lb-constrained-cyclepacking}, we have
\[
\frac{2}{3}\Delta+\frac{4}{3}w(\C_{\bmod 5}^{*})
\leq \frac{2}{3}\lrA{\ra+\frac{3}{2}}w(I^*)+\frac{4}{3}(2-2\ra)w(I^*)
=\frac{11-6\ra}{3}w(I^*).
\]

Using LP-UITP and the Hamiltonian cycle $H_{CS}$ on $V\cup\{v_0\}$, by Lemma~\ref{LPUITP}, we can obtain an itinerary $I_2$ with
$
w(I_2)\leq\min_{\gamma\geq0}\lrc{\gamma\cdot w(I^*)+e^{-\gamma}\cdot \frac{2}{3}\Delta+\frac{2}{3}w(H_{CS})}.
$
We will set $\gamma=\ln\frac{3+2\ra}{3}$ (to be defined later).
Since $\Delta\leq(\ra+\frac{3}{2})w(I^*)$ by Lemma~\ref{lb-delta} and $w(H_{CS})\leq \frac{3-\ra}{2}$ by Lemma~\ref{lb-chris}, we have
\[
\gamma\cdot w(I^*)+e^{-\gamma}\cdot \frac{2}{3}\Delta+\frac{2}{3}w(H_{CS})\leq \lrA{\frac{6-\ra}{3}+\ln\frac{3+2\ra}{3}}w(I^*).
\]

We will select the better one between these two itineraries. This yields an approximation ratio of
\[
\max_{\substack{0\leq \ra\leq 1}}\min\lrC{\frac{11-6\ra}{3},\frac{6-\ra}{3}+\ln\frac{3+2\ra}{3}}=\frac{11-6x_0}{3}<2.157,
\]
where $\ra=x_0>0.755$ in the worst case, the unique root of the equitation $\frac{5}{3}(x-1)+\ln\frac{3+2x}{3}=0$.
Note that it suffices to set $\gamma=\ln\frac{3+2x_0}{3}$.
\end{proof}

\section{Proof of Lemma~\ref{core}}\label{proof-1}
\begingroup
\def\thelemma{\ref{core}}
\begin{lemma}
When $k\geq3$, for any cycle $C=(v_0,v_1,...,v_k,v_0)\in I^*$, we have
\[
w(T_C)\leq\lrA{1-\max_{1\leq i\leq m}\frac{1}{2}x_{i}}w(C)\quad\text{and}\quad\Delta_C\leq\lrA{\frac{k+2}{2}-\sum_{i=1}^{m}ix_{i}}w(C),
\]
where it holds that $0\leq x_i\leq 1$ for each $1\leq i\leq m$, $\sum_{i=1}^{m}x_i=1$, and $m=\Ceil{\frac{k+1}{2}}$.
Moreover, we have
\[
\frac{2}{k}\Delta_C+\lrA{1-\frac{1}{k}}w(T_C)\leq\lrA{2-\frac{2l^2+k-1}{2kl}}w(C),
\]
where $l=\ceil{\frac{\sqrt{2k-1}-1}{2}}$.
\end{lemma}
\addtocounter{lemma}{-1}
\endgroup
\begin{proof}
When $k$ is odd (resp., even), the number of edges in the cycle $C$ is even (resp., odd). Due to different structural properties, these two cases have to be handled separately. We first consider the case that $k$ is odd. Another case will be considered in a similar way. For convince, we let $m=\ceil{\frac{k+1}{2}}$.

\textbf{Case~1: $k$ is odd.} Note that $m=\ceil{\frac{k+1}{2}}=\frac{k+1}{2}$.
We define parameters
\[
a_{i}=\frac{w(v_{i-1},v_{i})+w(v_{k+1-i},v_{(k+2-i)\bmod (k+1)})}{w(C)},
\]
where $1\leq i\leq m$. Then, We have $\sum_{i=1}^{m}a_{i}=1$ and $a_i\geq0$ for each $1\leq i\leq m$.
We may assume $w(C)\neq0$; otherwise, Lemma~\ref{core} holds trivially. See Figure~\ref{fig01} for an illustration of the definition on $a_i$. Note that $a_1$ measures the proportion of the weights of the home-edges in $C$, so it can be regarded as a generalization of $\ra$, but it works on a single cycle.

\begin{figure}[ht]
\centering
\begin{tikzpicture}
\tikzstyle{vertex}=[black,circle,fill,minimum size=5,inner sep=0pt]
\begin{scope}[every node/.style={vertex}]
 \node (l) at (0,0) {};
 \node (x1) at (0.6,0.5) {};
 \node (y1) at (0.6,-0.5) {};
 \node (x2) at (0.6+2,0.5) {};
 \node (y2) at (0.6+2,-0.5) {};
 \node (x3) at (0.6+4,0.5) {};
 \node (y3) at (0.6+4,-0.5) {};
 \node (x4) at (0.6+6,0.5) {};
 \node (y4) at (0.6+6,-0.5) {};
 \node (r) at (1.2+6,0) {};
\end{scope}

\node at (0.6+3,0) {$a_i$};

\node[left] at (0,0) {$v_0$};
\node[above] at (0.6,0.5) {$v_1$};
\node[below] at (0.6,-0.5) {$v_k$};
\node[above] at (0.6+2,0.5) {$v_{i-1}$};
\node[below] at (0.6+2,-0.5) {$v_{k+2-i}$};
\node[above] at (0.6+4,0.5) {$v_i$};
\node[below] at (0.6+4,-0.5) {$v_{k+1-i}$};
\node[above] at (0.6+6,0.5) {$v_{m-1}$};
\node[below] at (0.6+6,-0.5) {$v_{m+1}$};
\node[right] at (1.2+6,0) {$v_m~( m=\frac{k+1}{2})$};
\draw[very thick,-] (l.east) to (x1.west);
\draw[very thick,dotted,-] (x1.east) to (x2.west);
\draw[very thick,-,blue] (x2.east) to (x3.west);
\draw[very thick,dotted,-] (x3.east) to (x4.west);
\draw[very thick,-] (x4.east) to (r.west);
\draw[very thick,-] (l.east) to (y1.west);
\draw[very thick,dotted,-] (y1.east) to (y2.west);
\draw[very thick,-,blue] (y2.east) to (y3.west);
\draw[very thick,dotted,-] (y3.east) to (y4.west);
\draw[very thick,-] (y4.east) to (r.west);
\end{tikzpicture}
\caption{An illustration of the cycle $C=(v_0,v_1,\dots, v_{k},v_0)$ for the case of odd $k$, where $a_i$ ($1\leq i\leq m$) measures the proportion of the weights of the blue edges in $C$.}
\label{fig01}
\end{figure}

Next, we are ready to give refined lower bounds based on $\MST$ and $\Delta$.

\begin{lemma}~\label{odd-tree}
For any cycle $C=(v_0,v_1,...,v_k,v_0)\in I^*$ with odd $k$, we have $w(T_C)\leq(1-\max_{\substack{1\leq i\leq m}}\frac{1}{2}a_{i})w(C)$.
\end{lemma}
\begin{proof}
By the definition of $T_C$, we have that $w(T_C)=w(C)-\max_{0\leq i\leq k}w(v_{i},v_{(i+1)\bmod (k+1)})$. Therefore, we have
\[
\begin{split}
w(T_C)=&\ w(C)-\max_{0\leq i\leq k}w(v_{i},v_{(i+1)\bmod (k+1)})\\
\leq&\ w(C)-\max_{\substack{1\leq i\leq m}}\frac{1}{2}(w(v_{i-1},v_{i})+w(v_{k+1-i},v_{(k+2-i)\bmod (k+1)}))\\
=&\ \lrA{1-\max_{\substack{1\leq i\leq m}}\frac{1}{2}a_{i}}w(C),
\end{split}
\]
where the inequality can be obtained directly.
\end{proof}

\begin{lemma}~\label{odd-delta}
For any cycle $C=(v_0,v_1,...,v_k,v_0)\in I^*$ with odd $k$, we have $\Delta_C\leq (\frac{k+2}{2}-\sum_{i=1}^{m}ia_{i})w(C)$.
\end{lemma}
\begin{proof}
First, by the definition of $\Delta_C$, we have
\[
\Delta_C=\sum_{i=1}^{k}w(v_0, v_i)=\sum_{i=1}^{m-1}\lra{w(v_0,v_{i})+w(v_0,v_{k+1-i})}+w(v_0,v_{m}).
\]

Then, by the triangle inequality, we have
\[
w(v_0,v_{i})\leq\sum_{j=1}^{i}w(v_{j-1},v_{j}) \quad \mbox{and}\quad
w(v_0,v_{k+1-i})\leq \sum_{j=1}^{i}w(v_{k+1-j},v_{(k+2-j)\bmod(k+1)}).
\]
See Figure~\ref{fig02} for an illustration of these two inequalities. 
Moreover, we also have $w(v_0,v_{m})\leq\frac{1}{2}w(C)$.

\begin{figure}[ht]
\centering
\begin{tikzpicture}
\tikzstyle{vertex}=[black,circle,fill,minimum size=5,inner sep=0pt]
\begin{scope}[every node/.style={vertex}]
 \node (l) at (0,0) {};
 \node (x1) at (0.6,0.5) {};
 \node (y1) at (0.6,-0.5) {};
 \node (x2) at (0.6+2,0.5) {};
 \node (y2) at (0.6+2,-0.5) {};
 \node (x3) at (0.6+4,0.5) {};
 \node (y3) at (0.6+4,-0.5) {};
 \node (x4) at (0.6+6,0.5) {};
 \node (y4) at (0.6+6,-0.5) {};
 \node (r) at (1.2+6,0) {};
\end{scope}

\node[left] at (0,0) {$v_0$};
\node[above] at (0.6,0.5) {$v_1$};
\node[below] at (0.6,-0.5) {$v_k$};
\node[above] at (0.6+2,0.5) {$v_{i-1}$};
\node[below] at (0.6+2,-0.5) {$v_{k+2-i}$};
\node[above] at (0.6+4,0.5) {$v_i$};
\node[below] at (0.6+4,-0.5) {$v_{k+1-i}$};
\node[above] at (0.6+6,0.5) {$v_{m-1}$};
\node[below] at (0.6+6,-0.5) {$v_{m+1}$};
\node[right] at (1.2+6,0) {$v_m~(m=\frac{k+1}{2})$};
\draw[very thick,->,blue] (l.east) to (x1.west);
\draw[very thick,dotted,->,blue] (x1.east) to (x2.west);
\draw[very thick,->,blue] (x2.east) to (x3.west);
\draw[very thick,dotted,-] (x3.east) to (x4.west);
\draw[very thick,-] (x4.east) to (r.west);

\draw[very thick,<-,blue] (l.east) to (y1.west);
\draw[very thick,dotted,<-,blue] (y1.east) to (y2.west);
\draw[very thick,<-,blue] (y2.east) to (y3.west);
\draw[very thick,dotted,-] (y3.east) to (y4.west);
\draw[very thick,-] (y4.east) to (r.west);

\draw[very thick,<-,red] (l.east) to (x3.west);
\draw[very thick,->,red] (l.east) to (y3.west);
\end{tikzpicture}
\caption{An illustration of the two inequalities: there are two cycles (denoted by the directed edges): $(v_0,v_1,\dots,v_i,v_0)$ and $(v_0,v_{k+1-i},\dots,v_k,v_0)$; for each cycle, the weight of the red edge is at most the sum of the weights of the blue edges by the triangle inequality.}
\label{fig02}
\end{figure}

Recall that $w(v_{j-1},v_{j})+w(v_{k+1-j},v_{(k+2-j)\bmod(k+1)})=a_jw(C)$. Hence, we have
\begin{align*}
\Delta_C&\leq\sum_{i=1}^{m-1}\sum_{j=1}^{i}a_jw(C)+\frac{1}{2}w(C)\\
&=\sum_{i=1}^{m}\lra{m-i}a_{i}w(C)+\frac{1}{2}w(C)=\lrA{\frac{k+2}{2}-\sum_{i=1}^{m}ia_{i}}w(C),
\end{align*}
where the last equality follows from $\sum_{i=1}^{m}a_i=1$ and $m=\frac{k+1}{2}$.
\end{proof}

\begin{lemma}\label{evencycle}
For any cycle $C=(v_0,v_1,...,v_k,v_0)\in I^*$ with odd $k$, we have
\[
\frac{2}{k}\Delta_C+\lrA{1-\frac{1}{k}}w(T_C)\leq \max_{\substack{a_{1}\geq a_{2}\geq\dots\geq a_{m}\geq0\\a_{1}+a_{2}+\dots+a_{m}=1}}
\lrC{\frac{2k+1}{k}-\sum_{i=1}^{m}c^o_{i}a_{i}}w(C),
\]
where $c^o_1=\frac{k+3}{2k}$ and $c^o_{i}=\frac{2i}{k}$ for each $2\leq i\leq m$.
\end{lemma}
\begin{proof}
By Lemmas~\ref{odd-tree} and \ref{odd-delta}, we have
\[
\begin{split}
&\frac{2}{k}\Delta_C+\lrA{1-\frac{1}{k}}w(T_C)\\
&\leq\max_{\substack{a_{1}, a_{2},\dots, a_{m}\geq0\\a_{1}+a_{2}+\dots+a_{m}=1}}
\lrC{\frac{2}{k}\lrA{\frac{k+2}{2}-\sum_{i=1}^{m}ia_{i}}+\frac{k-1}{k}\lrA{1-\max_{\substack{1\leq i\leq m}}\frac{1}{2}a_{i}}}w(C)\\
&=\max_{\substack{a_{1}, a_{2},\dots, a_{m}\geq0\\a_{1}+a_{2}+\dots+a_{m}=1}}
\lrC{\frac{2k+1}{k}-\frac{2}{k}\sum_{i=1}^{m}ia_{i}-\max_{\substack{1\leq i\leq m}}\frac{k-1}{2k}a_{i}}w(C).\\
\end{split}
\]

Then, it suffices to prove that $a_{1}\geq a_{2}\geq\dots\geq a_{m}$, as this also implies $\max_{\substack{1\leq i\leq m}}a_{i}=a_1$.

Assume that there exists $a_{p}<a_{q}$ for some $p<q$. If we exchange their values, the value $\max_{\substack{1\leq i\leq m}}a_{i}$ does not change. However, since the coefficients of $a_p$ and $a_q$ satisfy $0>\frac{-2p}{k}>\frac{-2q}{k}$, we obtain a bigger solution which causes a contradiction.
\end{proof}

Next, we need to consider the case that $k$ is even. The analysis is almost the same.

\textbf{Case~2: $k$ is even.} Note that $m=\ceil{\frac{k+1}{2}}=\frac{k+2}{2}$.
We define parameters
\[
b_{i}=\frac{w(v_{i-1},v_{i})+w(v_{k+1-i},v_{(k+2-i)\bmod (k+1)})}{w(C)},
\]
where $1\leq i\leq m-1$. We let $b_{m}=\frac{w(v_{k/2},v_{k/2+1})}{w(C)}$. Then, we have $\sum_{i=1}^{m}b_{i}=1$. See Figure~\ref{fig03} for an illustration.

\begin{figure}[ht]
\centering
\begin{tikzpicture}
\tikzstyle{vertex}=[black,circle,fill,minimum size=5,inner sep=0pt]
\begin{scope}[every node/.style={vertex}]
 \node (l) at (0,0) {};
 \node (x1) at (0.6,0.5) {};
 \node (y1) at (0.6,-0.5) {};
 \node (x2) at (0.6+2,0.5) {};
 \node (y2) at (0.6+2,-0.5) {};
 \node (x3) at (0.6+4,0.5) {};
 \node (y3) at (0.6+4,-0.5) {};
 \node (x4) at (0.6+6,0.5) {};
 \node (y4) at (0.6+6,-0.5) {};
\end{scope}

\node at (0.6+3,0) {$b_i$};

\node[left] at (0,0) {$v_0$};
\node[above] at (0.6,0.5) {$v_1$};
\node[below] at (0.6,-0.5) {$v_k$};
\node[above] at (0.6+2,0.5) {$v_{i-1}$};
\node[below] at (0.6+2,-0.5) {$v_{k+2-i}$};
\node[above] at (0.6+4,0.5) {$v_i$};
\node[below] at (0.6+4,-0.5) {$v_{k+1-i}$};
\node[above] at (0.6+6,0.5) {$v_{m-1}$};
\node[below] at (0.6+6,-0.5) {$v_{m}$};

\node[right] at (0.8+6,0) {$b_m~(m=\frac{k+2}{2})$};

\draw[very thick,-,red] (x4.south) to (y4.north);

\draw[very thick,-] (l.east) to (x1.west);
\draw[very thick,dotted,-] (x1.east) to (x2.west);
\draw[very thick,-,blue] (x2.east) to (x3.west);
\draw[very thick,dotted,-] (x3.east) to (x4.west);

\draw[very thick,-] (l.east) to (y1.west);
\draw[very thick,dotted,-] (y1.east) to (y2.west);
\draw[very thick,-,blue] (y2.east) to (y3.west);
\draw[very thick,dotted,-] (y3.east) to (y4.west);

\end{tikzpicture}
\caption{An illustration of the cycle $C=(v_0,v_1,\dots, v_{k},v_0)$ for the case of even $k$, where $b_i$ ($1\leq i<m$) measures the proportion of the weights of the blue edges in $C$, and $b_m$ measures the proportion of the weight of the red edge in $C$.}
\label{fig03}
\end{figure}

\begin{lemma}~\label{even-tree}
For any cycle $C=(v_0,v_1,...,v_k,v_0)\in I^*$ with even $k$, we have $w(T_C)\leq(1-\max_{1\leq i\leq m}\frac{1}{2}b_{i})w(C)$.
\end{lemma}
\begin{proof}
Recall that $w(T_C)=w(C)-\max_{0\leq i\leq k}w(v_{i},v_{(i+1)\bmod (k+1)})$. Then, we have
\[
\begin{split}
w(T_C) =&\ w(C)-\max_{0\leq i\leq k}w(v_{i},v_{(i+1)\bmod (k+1)})\\
\leq&\ w(C)-\max_{\substack{1\leq i<m}}\lrC{\frac{1}{2}(w(v_{i-1},v_{i})+w(v_{k+1-i},v_{(k+2-i)\bmod (k+1)})),\ w(v_{k/2},v_{k/2+1})}\\
=&\ \lrA{1-\max_{\substack{1\leq i< m}}\lrC{\frac{1}{2}b_{i}, b_{m}}}w(C)\leq\lrA{1-\max_{\substack{1\leq i\leq m}}\frac{1}{2}b_{i}}w(C),
\end{split}
\]
where the inequalities can be obtained directly.
\end{proof}

\begin{lemma}~\label{even-delta}
For any cycle $C=(v_0,v_1,...,v_k,v_0)\in I^*$ with even $k$, we have $\Delta_C\leq(\frac{k+2}{2}-\sum_{i=1}^{m}ib_{i})w(C)$.
\end{lemma}
\begin{proof}
First, by the definition of $\Delta_C$, we have
\[
\Delta_C=\sum_{i=1}^{k}w(v_0, v_i)=\sum_{i=1}^{m-1}\lra{w(v_0,v_{i})+w(v_0,v_{k+1-i})}.
\]

Then, by the proof of Lemma~\ref{odd-delta}, we have
\[
w(v_0,v_{i})\leq\sum_{j=1}^{i}w(v_{j-1},v_{j}) \quad \mbox{and}\quad
w(v_0,v_{k+1-i})\leq \sum_{j=1}^{i}w(v_{k+1-j},v_{(k+2-j)\bmod(k+1)}).
\]
See Figure~\ref{fig04} for an illustration of these two inequalities. 

\begin{figure}[ht]
\centering
\begin{tikzpicture}
\tikzstyle{vertex}=[black,circle,fill,minimum size=5,inner sep=0pt]
\begin{scope}[every node/.style={vertex}]
 \node (l) at (0,0) {};
 \node (x1) at (0.6,0.5) {};
 \node (y1) at (0.6,-0.5) {};
 \node (x2) at (0.6+2,0.5) {};
 \node (y2) at (0.6+2,-0.5) {};
 \node (x3) at (0.6+4,0.5) {};
 \node (y3) at (0.6+4,-0.5) {};
 \node (x4) at (0.6+6,0.5) {};
 \node (y4) at (0.6+6,-0.5) {};
\end{scope}

\node[left] at (0,0) {$v_0$};
\node[above] at (0.6,0.5) {$v_1$};
\node[below] at (0.6,-0.5) {$v_k$};
\node[above] at (0.6+2,0.5) {$v_{i-1}$};
\node[below] at (0.6+2,-0.5) {$v_{k+2-i}$};
\node[above] at (0.6+4,0.5) {$v_i$};
\node[below] at (0.6+4,-0.5) {$v_{k+1-i}$};
\node[above] at (0.6+6,0.5) {$v_{m-1}$};
\node[below] at (0.6+6,-0.5) {$v_{m}$};

\node[right] at (0.8+6,0) {$(m=\frac{k+2}{2})$};

\draw[very thick,-] (x4.south) to (y4.north);

\draw[very thick,->,blue] (l.east) to (x1.west);
\draw[very thick,dotted,->,blue] (x1.east) to (x2.west);
\draw[very thick,->,blue] (x2.east) to (x3.west);
\draw[very thick,dotted,-] (x3.east) to (x4.west);

\draw[very thick,<-,blue] (l.east) to (y1.west);
\draw[very thick,dotted,<-,blue] (y1.east) to (y2.west);
\draw[very thick,<-,blue] (y2.east) to (y3.west);
\draw[very thick,dotted,-] (y3.east) to (y4.west);

\draw[very thick,<-,red] (l.east) to (x3.west);
\draw[very thick,->,red] (l.east) to (y3.west);
\end{tikzpicture}
\caption{An illustration of the two inequalities: there are two cycles (denoted by directed edges): $(v_0,v_1,\dots,v_i,v_0)$ and $(v_0,v_{k+1-i},\dots,v_k,v_0)$; for each cycle, the weight of the red edge is at most the sum of the weights of the blue edges by the triangle inequality.}
\label{fig04}
\end{figure}

Recall that $w(v_{j-1},v_{j})+w(v_{k+1-j},v_{(k+2-j)\bmod(k+1)})=b_jw(C)$. Hence, we have
\[
\Delta_C\leq\sum_{i=1}^{m-1}\sum_{j=1}^{i}b_jw(C)=\sum_{i=1}^{m}\lra{m-i}b_{i}w(C)=\lrA{\frac{k+2}{2}-\sum_{i=1}^{m}ib_{i}}w(C),
\]
where the last equality follows from $\sum_{i=1}^{m}b_i=1$ and $m=\frac{k+2}{2}$.
\end{proof}

Both Lemmas~\ref{even-tree} and \ref{even-delta} have the same form as Lemmas~\ref{odd-tree} and \ref{odd-delta}. Hence, by the proof of Lemma~\ref{evencycle}, we can derive the following lemma.

\begin{lemma}\label{oddcycle}
For any cycle $C=(v_0,v_1,...,v_k,v_0)\in I^*$ with even $k$, we have
\[
\frac{2}{k}\Delta_C+\lrA{1-\frac{1}{k}}w(T_C)\leq \max_{\substack{b_{1}\geq b_{2}\geq\dots\geq b_{m}\geq0\\b_{1}+b_{2}+\dots+b_{m}=1}}
\lrA{\frac{2k+1}{k}-\sum_{i=1}^{m}c^e_{i}b_{i}}w(C),
\]
where $c^e_1=\frac{k+3}{2k}$ and $c^e_{i}=\frac{2i}{k}$ for each $2\leq i\leq m$.
\end{lemma}

Lemmas~\ref{evencycle} and \ref{oddcycle} can be combined as follows.
\begin{lemma}\label{cycle}
When $k\geq3$, for any cycle $C=(v_0,v_1,...,v_k,v_0)\in I^*$, we have
\[
\frac{2}{k}\Delta_C+\lrA{1-\frac{1}{k}}w(T_C)\leq \max_{\substack{x_{1}\geq x_{2}\geq\dots\geq x_{m}\geq0\\x_{1}+x_{2}+\dots+x_{m}=1}}
\lrA{\frac{2k+1}{k}-\sum_{i=1}^{m}c_{i}x_{i}}w(C),
\]
where $c_1=\frac{k+3}{2k}$ and $c_{i}=\frac{2i}{k}$ for each $2\leq i\leq m$.
\end{lemma}

For the sake of analysis, we generalize some definitions and define variables $\{x_{i}\mid i\in\mathbb{Z}_{\geq 1}\}$ and coefficients $\{c_{i}\mid i\in\mathbb{Z}_{\geq 1}\}$ where $c_1=\frac{k+3}{2k}$ and $c_{i}=\frac{2i}{k}$ for each $ i\geq 2$, then we have
\[
\min_{\substack{x_{1}\geq x_{2}\geq\dots\geq0\\x_{1}+x_{2}+\dots=1}}
\sum_{i=1}^{\infty}c_{i}x_{i}
\leq\min_{\substack{x_{1}\geq x_{2}\geq\dots\geq x_{m}\geq0\\x_{1}+x_{2}+\dots+x_{m}=1}}
\sum_{i=1}^{m}c_{i}x_{i},
\]
since the latter is a special case of the former. Then, we focus on analyzing the former LP.

\begin{lemma}\label{lp}
For the following LP,
\[
\min_{\substack{x_{1}\geq x_{2}\geq\dots\geq0\\x_{1}+x_{2}+\dots=1}}\sum_{i=1}^{\infty}c_{i}x_{i},
\]
where $c_1=\frac{k+3}{2k}$ and $c_{i}=\frac{2i}{k}$ for each $i\geq 2$. There exists an optimal solution, denoted by $\mbox{SOL}(x^*_1, x^*_2,\dots)$, satisfying that $x_{1}^{*}=x_{2}^{*}=\dots=x_{l}^{*}=\frac{1}{l}$ and $x_{i}^{*}=0$ for all $i>l$, where $l=\ceil{\frac{\sqrt{2k-1}-1}{2}}$. Moreover, we have $\mbox{SOL}(x^*_1, x^*_2,\dots)=\frac{2l^2+2l+k-1}{2kl}$.
\end{lemma}
\begin{proof}
Since $\lim_{\substack{i\rightarrow +\infty}}c_i=+\infty$, there exists an optimal solution where the number of nonzero variables is limited.
Hence, we may consider the optimal solution $\mbox{SOL}(x^*_1, x^*_2,\dots)$ with the number of nonzero variables minimized. 
Let $x_{l}^{*}$ be the last nonzero variable.
We define $\overline{c_{i}}=\frac{1}{i}\sum_{i=1}^{i}c_{i}$ for $i>0$.
We will show that $x_{1}^{*}=x_{2}^{*}=\dots=x_{l}^{*}=\frac{1}{l}$ and $l=\arg\min_{l'\in \mathbb{Z}_{\geq 1}}\{c_{l'+1}\geq\overline{c_{l'}}\}$.

Since the number of nonzero variables is minimized, we have $c_2\geq \overline{c_1}=c_1$ if and only if $l=1$. Hence, we know that $l=\arg\min_{l'\in \mathbb{Z}_{\geq 1}}\{c_{l'+1}\geq\overline{c_{l'}}\}$ holds when $c_2\geq \overline{c_1}$. Next, we consider $c_2<\overline{c_1}$ and in this case we have $l>1$.

\begin{claim}\label{cla1}
For all $1<i\leq l$, we have $c_i<\overline{c_{i-1}}$.
\end{claim}
\begin{proof}[Claim Proof]
Suppose there exists $n_0\in\mathbb{Z}_{>1}$ such that $1<n_0\leq l$ and $c_{n_0}\geq \overline{c_{n_0-1}}$, then we let $x_{i}^{**}=x_{i}^{*}+\frac{1}{n_0-1}\sum_{i'=n_0}^{l}x_{i'}^{*}$ for $1\leq i<n_0$ and $x_{i}^{**}=0$ for $n_0\leq i\leq l$. Then, we have
\begin{align*}
\sum_{i=1}^{n_0-1}c_{i}x_{i}^{**}=\ &\sum_{i=1}^{n_0-1}c_{i}\lrA{x_{i}^{*}+\frac{1}{n_0-1}\sum_{i'=n_0}^{l}x_{i'}^{*}}\\
=\ &\sum_{i=1}^{n_0-1}c_{i}x_{i}^{*}+\sum_{i'=n_0}^{l}\sum_{i=1}^{n_0-1}\frac{c_i}{n_0-1}x_{i'}^{*}\\
=\ &\sum_{i=1}^{n_0-1}c_{i}x_{i}^{*}+\sum_{i'=n_0}^{l}\overline{c_{n_0-1}}x_{i'}^{*}\\
\leq\ &\sum_{i=1}^{n_0-1}c_{i}x_{i}^{*}+\sum_{i'=n_0}^{l}c_{n_0}x_{i'}^{*}\\
\leq\ &\sum_{i=1}^{n_0-1}c_{i}x_{i}^{*}+\sum_{i'=n_0}^{l}c_{i'}x_{i'}^{*}=\sum_{i=1}^{l}c_{i}x_{i}^{*},
\end{align*}
where the first inequality follows from $\overline{c_{n_0-1}}\leq c_{n_0}$ and the second from $c_{n_0}\leq c_{i'}$ for $i'\geq n_0>1$ since $c_{i}=\frac{2i}{k}$ when $i>1$.
Note that the number of nonzero variables in the new solution is $n_0-1<l$, which causes a contradiction.
\end{proof}

\begin{claim}\label{cla2}
For all $1<i\leq l$, we have $x_{i}^{*}=x_{i-1}^{*}$.
\end{claim}
\begin{proof}[Claim Proof]
Suppose there exists (minimized) $n_0\in\mathbb{Z}_{>1}$ such that $1<n_0\leq l$ and $x_{1}^{*}=\dots=x_{n_0-1}^{*}> x_{n_0}^{*}$, we simply let $x_{1}^{**}=\dots=x_{n_0}^{**}=\frac{1}{n_0}\sum_{i=1}^{n_0}x_{i}^{*}$. Note that the new solution still forms a feasible solution since $x_{n_0}^{**}>x_{n_0}^{*}\geq x_{n_0+1}^{*}\geq\cdots$.
We have
\begin{align*}
\sum_{i=1}^{n_0}c_{i}x_{i}^{**}=\ &\sum_{i=1}^{n_0-1}c_{i}x_{i}^{**}+c_{n_0}x_{n_0}^{**}\\
=\ &(n_0-1)\overline{c_{n_0-1}}x_{n_0}^{**}+c_{n_0}x_{n_0}^{**}\\
=\ &(n_0-1)\overline{c_{n_0-1}}x_{n_0}^{**}+c_{n_0}(x_{n_0}^{**}-x_{n_0}^{*})+c_{n_0}x_{n_0}^{*}\\
<\ &(n_0-1)\overline{c_{n_0-1}}x_{n_0}^{**}+\overline{c_{n_0-1}}(x_{n_0}^{**}-x_{n_0}^{*})+c_{n_0}x_{n_0}^{*}\\
=\ &\overline{c_{n_0-1}}(n_0x_{n_0-1}^{**}-x_{n_0}^{*})+c_{n_0}x_{n_0}^{*}\\
=\ &\overline{c_{n_0-1}}\sum_{i=1}^{n_0-1}x_{i}^{*}+c_{n_0}x_{n_0}^{*}=\sum_{i=1}^{n_0}c_{i}x_{i}^{*},
\end{align*}
where the inequality follows from $c_{n_0}<\overline{c_{n_0-1}}$ by Claim~\ref{cla1}.
The new solution is better than the optimal solution, a contradiction.
\end{proof}

\begin{claim}\label{cla3}
$c_{l+1}\geq\overline{c_{l}}$.
\end{claim}
\begin{proof}[Claim Proof]
Otherwise, if $c_{l+1}<\overline{c_{l}}$, we let $x_{1}^{**}=\dots=x_{l+1}^{**}=\frac{1}{l+1}\sum_{i=1}^{l}x_{i}^{*}=x^{**}$. Note that $x_{1}^{*}=\dots=x_{l}^{*}$ by Claim~\ref{cla2}, and then we have
\[
\sum_{i=1}^{l+1}c_{i}x^{**}_i=\sum_{i=1}^{l+1}c_{i}x^{**}=l\overline{c_l}x^{**}+c_{l+1}x^{**}<\overline{c_l}(l+1)x^{**}=\overline{c_l}\sum_{i=1}^{l}x_{i}^{*}=\sum_{i=1}^{l}c_ix_{i}^{*}.
\]
The new solution is better than the optimal solution, a contradiction.
\end{proof}

By Claims~\ref{cla1} and~\ref{cla3}, we have
\[
l=\arg\min_{l'\in \mathbb{Z}_{\geq 1}}\lrC{c_{l'+1}\geq\overline{c_{l'}}}=\arg\min_{l'\in\mathbb{Z}_{\geq 1}}
\lrC{l'^2+l'-\frac{k-1}{2}\geq0}=\Ceil{\frac{\sqrt{2k-1}-1}{2}}.
\]
Therefore, we have
\[
\mbox{SOL}(x^*_1, x^*_2,\dots)=\sum_{i=1}^{l}c_ix^*_i=\frac{1}{l}\sum_{i=1}^{l}c_i=\overline{c_l}=\frac{2l^2+2l+k-1}{2kl},
\]
where $l=\ceil{\frac{\sqrt{2k-1}-1}{2}}$.
\end{proof}

Finally, Lemma~\ref{core} follows directly from Lemmas~\ref{cycle} and \ref{lp}.
\end{proof}

\section{Proof of Theorem~\ref{res-unsplit}}\label{proof-2}
\begingroup
\def\thetheorem{\ref{res-unsplit}}
\begin{theorem}
Given fixed $k\geq 3$, for unsplittable $k$-CVRP, the approximation ratio of LP-UITP is shown as follows.
\begin{itemize}
\item If $\floor{\frac{-(C+1)+\sqrt{(C+1)^2+4kC}}{2}}\neq \floor{\frac{-(C+1)+\sqrt{(C+1)^2+4(k+1)C}}{2}}$, where $C=\frac{\floor{k/2}}{\floor{ k/2}+1}\cdot\frac{1}{2}$, the approximation ratio is at most $h(x^*)$, where $h(x)=1+\ln\lrA{\frac{k-2\floor{\frac{1}{x}}+\floor{\frac{1}{x}}^2x+\floor{\frac{1}{x}}x}{\floor{k/2}+1}}+\frac{\floor{k/2}}{\floor{ k/2}+1}\cdot\frac{3-x}{2}$, $x^*=\frac{1}{C}-\frac{k-2z^*}{(z^*)^2+z^*}$, and $z^*=\floor{\frac{1}{x^*}}=\Floor{\frac{-(C+1)+\sqrt{(C+1)^2+4(k+1)C}}{2}}$;
\item Otherwise, the approximation ratio is at most $\max\{g(\ceil{z_0}), g(\floor{z_0})\}$, where $g(z)=1+\ln\left(\frac{k-z+1}{\floor{k/2}+1}\right)+\frac{\floor{k/2}}{\floor{ k/2}+1}\cdot\frac{3-\frac{1}{z}}{2}$, and $z_0=\frac{\sqrt{\floor{k/2}^2+8\floor{k/2}(\floor{k/2}+1)(k+1)}-\floor{k/2}}{4(\floor{k/2}+1)}$.
\end{itemize}
\end{theorem}
\addtocounter{theorem}{-1}
\endgroup
\begin{proof}
Recall that the solution obtained by LP-UITP has a weight of at most
\[
\gamma \cdot w(I^*)+e^{-\gamma}\cdot \frac{2}{\floor{ k/2}+1}\Delta+\lrA{1-\frac{1}{\floor{ k/2}+1}}w(H_{CS}),
\]
where $w(H_{CS})\leq \frac{1}{2}w(I^*)+\MST$ by Lemmas~\ref{lb-tsp} and \ref{chris}.

Recall that $I^*$ is an optimal itinerary where every tour delivers exactly $k$ of the demand.
By generalizing the results in Lemmas~\ref{odd-delta}, \ref{even-delta}, \ref{odd-tree}, and \ref{even-tree}, we can obtain the following result.

\begin{claim}\label{unsplittable-bounds}
For unsplittable $k$-CVRP, it holds that
\[
\Delta\leq\lrA{\frac{k+2}{2}-\sum_{i=1}^{m}ix_i}w(I^*)\quad \mbox{and}\quad \MST\leq\lrA{1-\max_{1\leq i\leq m}\frac{1}{2}x_i}w(I^*),
\]
where $\sum_{i=1}^{m}x_i=1$ and $x_i\geq0$ for each $1\leq i\leq m$.
\end{claim}
\begin{proof}[Claim Proof]
We may assume w.l.o.g.\ that $w(I^*)\neq 0$.
Given the instance $G$ of unsplittable $k$-CVRP, we obtain an instance $G'$ of unit-demand $k$-CVRP by replacing each customer $v_i$ with $d_i$ unit-demand customers as the same place. 
Then, $I^*$ corresponds to a feasible solution in $G'$, denoted as $I'$, where $I'$ is a set of $(k+1)$-cycles intersecting only at $v_0$. We can obtain $w(I^*)=w(I')=\sum_{C\in I'}w(C)$.  Next, we focus on the graph $G'$.

Consider an arbitrary cycle $C\in I'$ in $G'$.
If $w(C)\neq 0$, we define $\lambda_C\coloneqq \frac{w(C)}{w(I')}$; otherwise, we define $\lambda_C=0$. Then, we have $\sum_{C\in I'}\lambda_C=1$.
Moreover, we define $\Delta_C=\sum_{v\in C}w(v_0,v)$ and $w(T_C)=w(C)-\max_{0\leq i\leq k}w(v_i,v_{(i+1)\bmod (k+1)})$.
By Lemmas~\ref{odd-delta}, \ref{even-delta}, \ref{odd-tree}, and \ref{even-tree}, we have
\[
\Delta_C\leq\lrA{\frac{k+2}{2}-\sum_{i=1}^{m}ix^C_i}w(C)\quad \mbox{and}\quad w(T_C)\leq\lrA{1-\max_{1\leq i\leq m}\frac{1}{2}x^C_i}w(C),
\]
where $\sum_{i=1}^{m}x^C_i=1$ and $x^C_i\geq0$ for each $1\leq i\leq m$.

Let $x_i=\sum_{C\in I'}\lambda_C\cdot x^C_i$. We have 
\begin{align*}
\sum_{C\in I'}\Delta_C&\leq\sum_{C\in I'}\lrA{\frac{k+2}{2}-\sum_{i=1}^{m}ix^C_i}w(C)\\
&=\sum_{C\in I'}\lrA{\frac{k+2}{2}-\sum_{i=1}^{m}ix^C_i}\lambda_C\cdot w(I')=\lrA{\frac{k+2}{2}-\sum_{i=1}^{m}ix_i} w(I'),
\end{align*}
where the first equality follows from $w(C)=\lambda_C\cdot w(I')$ and the second equality follows from $\sum_{C\in I'}\lambda_C=1$ and $x_i=\sum_{C\in I'}\lambda_C\cdot x^C_i$. Similarly, we have
\begin{align*}
\sum_{C\in I'}w(T_C)&=\sum_{C\in I'}\lrA{1-\max_{1\leq i\leq m}\frac{1}{2}x^C_i}w(C)\\
&=\sum_{C\in I'}\lrA{1-\max_{1\leq i\leq m}\frac{1}{2}x^C_i}\lambda_C\cdot w(I')=\lrA{1-\max_{1\leq i\leq m}\frac{1}{2}x_i}w(I'),
\end{align*}
where $\sum_{i=1}^{m}x_i=1$ and $x_i\geq0$ for each $1\leq i\leq m$.

Note that the value $\Delta$ in $G$ equals to $\sum_{C\in I'}\Delta_C$, and the weight of the minimum weight spanning tree in $G$ is the same as that in $G'$, i.e., $\MST\leq \sum_{C\in I'}w(T_C)$. Moreover, since $w(I')=w(I^*)$, we have 
\[
\Delta\leq\lrA{\frac{k+2}{2}-\sum_{i=1}^{m}ix_i}w(I^*)\quad \mbox{and}\quad \MST\leq\lrA{1-\max_{1\leq i\leq m}\frac{1}{2}x_i}w(I^*),
\]
where $\sum_{i=1}^{m}x_i=1$ and $x_i\geq0$ for each $1\leq i\leq m$.
\end{proof}

By Claim~\ref{unsplittable-bounds}, the approximation ratio of LP-UITP satisfies
\[
\max_{\substack{x_1, x_2,\dots, x_{m}\geq 0\\ x_1+x_2+\dots+x_{m}=1}}
\min_{\substack{\gamma\geq 0}}
\lrC{\gamma+e^{-\gamma}\cdot\frac{2}{\floor{k/2}+1}\lrA{\frac{k+2}{2}-\sum_{i=1}^{m}ix_i}+\frac{\floor{k/2}}{\floor{ k/2}+1}\lrA{\frac{3}{2}-\max_{1\leq i\leq m}\frac{1}{2}x_i}}.
\]

By the proof in Lemma~\ref{evencycle}, it holds $x_1\geq x_2\geq\dots\geq x_m$ in the worst case. Hence, we have $\max_{1\leq i\leq m}x_i=x_1$, and the approximation ratio satisfies
\[
\max_{\substack{x_1\geq x_2\geq\dots\geq x_{m}\geq 0\\ x_1+x_2+\dots+x_{m}=1}}
\min_{\substack{\gamma\geq 0}}
\lrC{\gamma+e^{-\gamma}\cdot\frac{2}{\floor{k/2}+1}\lrA{\frac{k+2}{2}-\sum_{i=1}^{m}ix_i}+\frac{\floor{k/2}}{\floor{ k/2}+1}\cdot\frac{3-x_1}{2}}.
\]

Taking that $x_1$ is a constant, and by a similar argument to the proof of Lemma~\ref{lp}, we can show that when $\frac{k+2}{2}-\sum_{i=1}^{m}ix_i$ is maximized, the solution satisfies $x_1=x_2=\dots=x_{l}\geq x_{l+1}=1-lx_1$ and $x_{l+2}=\cdots=x_{m}=0$, where $l=\floor{\frac{1}{x_1}}\leq m$. 
Note that we cannot obtain $x_{l}=x_{l+1}$ directly.
So, we have 
\begin{align*}
\sum_{i=1}^{m}ix_i=\sum_{i=1}^{l}ix_1+(l+1)x_{l+1}&=\frac{l(l+1)}{2}x_1+(l+1)(1-lx_1)\\
&=\frac{-l^2x_1-lx_1+2l+2}{2}.
\end{align*}

Then, let $l=\floor{\frac{1}{x}}$. The approximation ratio is at most
\[
\min_{\gamma\geq 0}\max_{0<x\leq 1}f(\gamma,x),
\]
where
\begin{equation}\label{eqfx}
f(\gamma,x)=\gamma+e^{-\gamma}\cdot\frac{k+l^2x+lx-2l}{\floor{k/2}+1}+\frac{\floor{k/2}}{\floor{ k/2}+1}\cdot\frac{3-x}{2}.
\end{equation}

Then, we explore some properties of the function $f(\gamma, x)$.

For any $z\in\mathbb{Z}_{\geq1}$, it can be verified that 
\[
\lim_{x\rightarrow\lrA{\frac{1}{z}}^-}f(\gamma, x)=\lim_{x\rightarrow\lrA{\frac{1}{z}}^+}f(\gamma, x).
\]
Moreover, under $x\in[\frac{1}{z+1},\frac{1}{z}]$, we have
\[
k-z\leq\frac{k-2l+l^2x+lx}{\floor{k/2}+1}\leq k-z+1.
\]
Since $k\in\mathbb{Z}_{\geq1}$, we known that
\begin{equation}\label{eqeq1}
\frac{k-2l+l^2x+lx}{\floor{k/2}+1}\leq1\quad\Longleftrightarrow\quad x\leq \frac{1}{k-\floor{k/2}}.
\end{equation}

Note that
\[
\min_{\gamma\geq 0}\max_{0<x\leq 1}f(\gamma,x)\leq\max\lrC{\min_{\substack{\gamma\geq 0}}\max_{0<x\leq \frac{1}{k-\floor{k/2}}}f(\gamma, x),\quad\min_{\substack{\gamma\geq 0}}\max_{\frac{1}{k-\floor{k/2}}<x\leq 1}f(\gamma, x)}
\]

Next, we consider two cases: $0<x\leq\frac{1}{k-\floor{k/2}}$ and $\frac{1}{k-\floor{k/2}}<x\leq 1$.

\textbf{Case~1: $\min_{\substack{\gamma\geq 0}}\max_{0<x\leq \frac{1}{k-\floor{k/2}}}f(\gamma, x)\geq\min_{\substack{\gamma\geq 0}}\max_{\frac{1}{k-\floor{k/2}}<x\leq 1}f(\gamma, x)$.} By setting $\gamma=0$, we have
\[
\max_{0<x\leq \frac{1}{k-\floor{k/2}}}f(0, x)=\frac{k-2l+l^2x+lx}{\floor{k/2}+1}+\frac{\floor{k/2}}{\floor{ k/2}+1}\cdot\frac{3-x}{2}.
\]
Since $0<x\leq \frac{1}{k-\floor{k/2}}$, we have $l=\floor{\frac{1}{x}}\geq k-\floor{k/2}$. For any $x\in(\frac{1}{z+1}, \frac{1}{z}]$ with $z\in\mathbb{Z}_{\geq k-\floor{k/2}}$, since $k\in\mathbb{Z}_{\geq 3}$, it can be verified that
\[
\frac{\partial f}{\partial x}=\frac{l^2+l}{\floor{k/2}+1}-\frac{\floor{k/2}}{\floor{ k/2}+1}\cdot\frac{1}{2}\geq 0.
\]
Hence, we have
\begin{align*}
\max_{0<x\leq \frac{1}{k-\floor{k/2}}}f(0, x)&\leq f\lrA{0,\frac{1}{k-\floor{k/2}}}\\
&=\frac{k-(k-\floor{k/2})+1}{\floor{k/2}+1}+\frac{\floor{k/2}}{\floor{ k/2}+1}\cdot\frac{3-\frac{1}{k-\floor{k/2}}}{2}.\\
&=\frac{5}{2}-\frac{3+\frac{\floor{k/2}}{k-\floor{k/2}}}{2(\floor{k/2}+1)}.
\end{align*}
Therefore, in this case, LP-UITP with $\gamma=0$ achieves an approximation ratio of at most $\frac{5}{2}-\frac{3+\frac{\floor{k/2}}{k-\floor{k/2}}}{2(\floor{k/2}+1)}$.

\textbf{Case~2: $\min_{\substack{\gamma\geq 0}}\max_{0<x\leq \frac{1}{k-\floor{k/2}}}f(\gamma, x)\leq \min_{\substack{\gamma\geq 0}}\max_{\frac{1}{k-\floor{k/2}}<x\leq 1}f(\gamma, x)$.} 
When $\frac{1}{k-\floor{k/2}}<x\leq 1$, we have $\frac{k-2l+l^2x+lx}{\floor{k/2}+1}>1$ by (\ref{eqeq1}). Moreover, it can be verified that
\[
\left. \frac{\partial f}{\partial \gamma} \right|_{\gamma=\ln\lrA{\frac{k-2l+l^2x+lx}{\floor{k/2}+1}}}=0.
\]
Let
\begin{equation}\label{eqhx}
\begin{split}
h(x)&\coloneqq f\lrA{\ln\lrA{\frac{k-2l+l^2x+lx}{\floor{k/2}+1}}, x}\\
&=1+\ln\lrA{\frac{k-2l+l^2x+lx}{\floor{k/2}+1}}+\frac{\floor{k/2}}{\floor{ k/2}+1}\cdot\frac{3-x}{2}.
\end{split}
\end{equation}

Hence, we can obtain 
\begin{align*}
\min_{\gamma\geq0}\max_{\frac{1}{k-\floor{k/2}}\leq x\leq 1}f(\gamma, x)&\leq \max_{x\leq 1}h(x).
\end{align*}

For any $x\in(\frac{1}{z+1}, \frac{1}{z})$ with $z\in\mathbb{Z}_{\geq 1}$, we have $l=\floor{\frac{1}{x}}=z$. By (\ref{eqhx}), we have
\begin{equation}\label{efdh}
h'(x)=\frac{z^2+z}{k-2z+z^2x+zx}-\frac{\floor{k/2}}{\floor{ k/2}+1}\cdot\frac{1}{2}.
\end{equation}
Then, it can be verified that $h(x)$ is a concave function and
\[
\lim_{x\rightarrow\lrA{\frac{1}{z}}^-}h(x)=\lim_{x\rightarrow\lrA{\frac{1}{z}}^+}h(x).
\]
Therefore, $h(x)$ attains its maximum value only if $h'(x)=0$ with $x\in(\frac{1}{z+1},\frac{1}{z})$ for some $z\in\mathbb{Z}_{\geq1}$ or $\frac{1}{x}$ is an integer. See Figure~\ref{example} for an illustration.

\begin{figure}[!t]
    \centering
    \begin{subfigure}[t]{0.66\textwidth}  
        \centering
        \includegraphics[width=\textwidth]{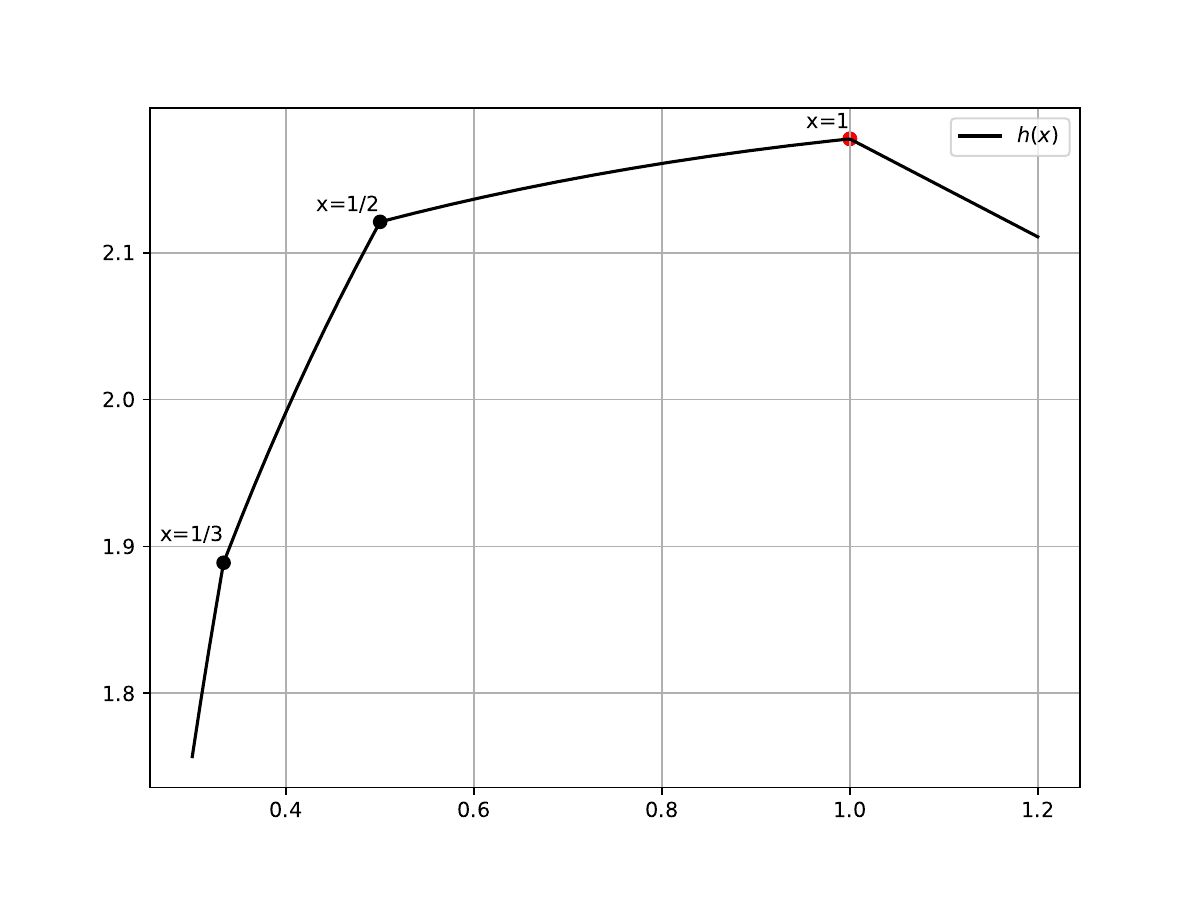}
        \caption{The function $h(x)$ attains its maximum value when $\frac{1}{x}$ is an integer.}
        \label{fig:left}
    \end{subfigure}
    
    \vspace{0.2cm}  
    
    \begin{subfigure}[t]{0.66\textwidth}  
        \centering
        \includegraphics[width=\textwidth]{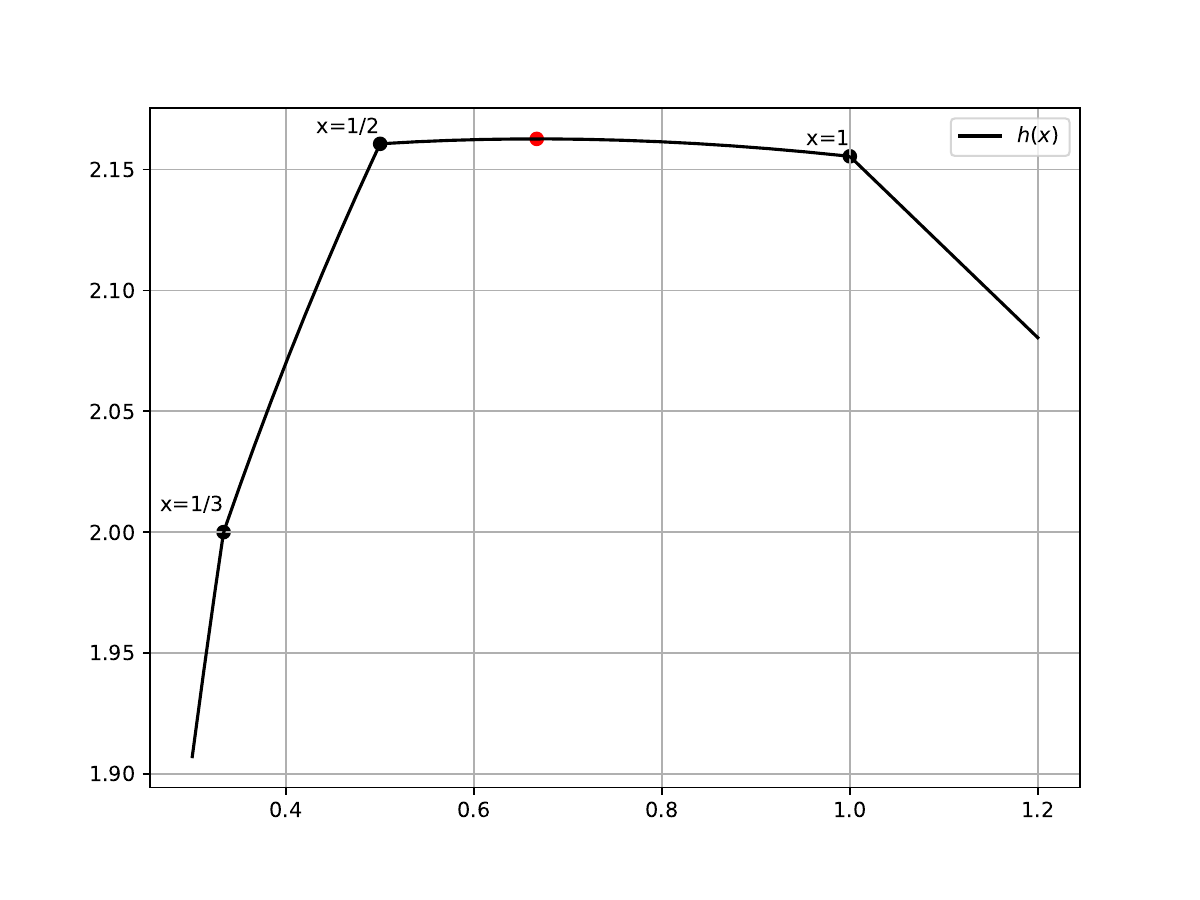}
        \caption{The function $h(x)$ attains its maximum value when $h'(x)=0$.}
        \label{fig:right}
    \end{subfigure}
    \caption{An illustration of the two cases when $h(x)$ attains its maximum value, where the point with respect to the maximum value is marked in red.}
    \label{example}
\end{figure}

We first consider the former case.

\textbf{Case~2.1: $h(x)$ attains its maximum value when $h'(x)=0$ with $x\in(\frac{1}{z+1},\frac{1}{z})$ for some $z\in\mathbb{Z}_{\geq1}$.} In this case, we have $h'(\lrA{\frac{1}{z+1}}^+)>0$ and $h'(\lrA{\frac{1}{z}}^-)<0$. Let $C=\frac{\floor{k/2}}{\floor{ k/2}+1}\cdot\frac{1}{2}$. Then, by (\ref{efdh}), we have $\frac{z^2+z}{k-z}>C$ and $\frac{z^2+z}{k-z+1}<C$, which implies that the interval $(\frac{-(C+1)+\sqrt{(C+1)^2+4kC}}{2}, \frac{-(C+1)+\sqrt{(C+1)^2+4(k+1)C}}{2})$ contains an integer. That is 
\[
\Floor{\frac{-(C+1)+\sqrt{(C+1)^2+4kC}}{2}}\neq \Floor{\frac{-(C+1)+\sqrt{(C+1)^2+4(k+1)C}}{2}}.
\]

In this case, the interval contains only one integer since it can be verified that for any $k\geq 3$, the length of the interval is at most 1, i.e., 
\[
\frac{\sqrt{(C+1)^2+4(k+1)C}-\sqrt{(C+1)^2+4kC}}{2}\leq 1.
\]

Therefore, let $z^*=\Floor{\frac{-(C+1)+\sqrt{(C+1)^2+4(k+1)C}}{2}}$, and then we have $h'(x)=0$ for some $x\in(\frac{1}{z^*+1},\frac{1}{z^*})$.
Let $h'(x^*)=0$, and then we can obtain 
\[
x^*=\frac{1}{C}-\frac{k-2z^*}{(z^*)^2+z^*}.
\]

Note that $\floor{\frac{1}{x^*}}=z^*$.
We can obtain 
\begin{align*}
\min_{\gamma\geq0}\max_{\frac{1}{k-\floor{k/2}}\leq x\leq 1}f(\gamma, x)&\leq \max_{x\leq 1}h(x)\\
&=h(x^*)=f\lrA{\ln\lrA{\frac{k-2z^*+(z^*)^2x^*+z^*x^*}{\floor{k/2}+1}}, x^*}.
\end{align*}
For any $k\in\mathbb{Z}_{\geq3}$, since $x^*>\frac{1}{z^*+1}$, we have
\begin{align*}
\frac{k-2z^*+(z^*)^2x^*+z^*x^*}{\floor{k/2}+1}&\geq\frac{k-z^*}{\floor{k/2}+1}=\frac{k-\Floor{\frac{-(C+1)+\sqrt{(C+1)^2+4(k+1)C}}{2}}}{\floor{k/2}+1}\geq 1.
\end{align*}

Thus, it is feasible to set $\gamma^*=\ln\lrA{\frac{k-2z^*+(z^*)^2x^*+z^*x^*}{\floor{k/2}+1}}$. 
Then, LP-UITP with $\gamma=\gamma^*$ achieves an approximation ratio of at most
\[
h(x^*)=1+\ln\lrA{\frac{k-2z^*+(z^*)^2x^*+z^*x^*}{\floor{k/2}+1}}+\frac{\floor{k/2}}{\floor{k/2}+1}\cdot\frac{3-x^*}{2}.
\]

\textbf{Case~2.2: $h(x)$ attains its maximum value when $\frac{1}{x}$ is an integer.} In this case, we have
\[
\Floor{\frac{-(C+1)+\sqrt{(C+1)^2+4kC}}{2}}=\Floor{\frac{-(C+1)+\sqrt{(C+1)^2+4(k+1)C}}{2}}.
\]

Let
\begin{equation}\label{eqgx}
\begin{split}
g(z)&\coloneqq 1+\ln\left(\frac{k-2z+z^2\cdot\frac{1}{z}+z\cdot\frac{1}{z}}{\floor{k/2}+1}\right)-\frac{\floor{k/2}}{\floor{ k/2}+1}\cdot\frac{3-\frac{1}{z}}{2}\\
&=1+\ln\left(\frac{k-z+1}{\floor{k/2}+1}\right)+\frac{\floor{k/2}}{\floor{ k/2}+1}\cdot\frac{3-\frac{1}{z}}{2}.
\end{split}
\end{equation}

For any $z\in\mathbb{Z}_{\geq 1}$, we have $g(z)=h(\frac{1}{z})$. Hence, the approximation ratio is given by
\[
\max_{z\in\mathbb{Z}_{\geq 1}}g(z).
\]

Note that 
\[
g'(z)=\frac{-1}{k+1-z}+\frac{\floor{k/2}}{2(\floor{k/2}+1)}\cdot \frac{1}{z^2}.
\]
When $z=z_0=\frac{\sqrt{\floor{k/2}^2+8\floor{k/2}(\floor{k/2}+1)(k+1)}-\floor{k/2}}{4(\floor{k/2}+1)}$, we have $g'(z)=0$. Hence, we have 
\[
\max_{z\in\mathbb{Z}_{\geq1}} g(z)=\max\{g(\ceil{z_o}), g(\floor{z_o})\}.
\]

Next, we show how to set $\gamma$ to call LP-UITP.

Let $z^*=\arg\max_{z\in\mathbb{Z}_{\geq1}} g(z)$, and $x^*=\frac{1}{z^*}$. Note that we have $z^*=\ceil{z_0}$ if $g(\ceil{z_0})\geq g(\floor{z_0})$, and $z^*=\floor{z_0}$ otherwise. Since $\floor{\frac{1}{x^*}}=z^*$, we know that
\begin{align*}
\min_{\gamma\geq0}\max_{\frac{1}{k-\floor{k/2}}\leq x\leq 1}f(\gamma, x)&\leq \max_{x\leq 1}h(x)\\
&=\max_{z\in\mathbb{Z}_{\geq 1}}g(z)\\
&=g(z^*)\\
&=h(x^*)\\
&=f\lrA{\ln\lrA{\frac{k-2z^*+(z^*)^2x^*+z^*x^*}{\floor{k/2}+1}}, x^*}\\
&=f\lrA{\ln\lrA{\frac{k-z^*+1}{\floor{k/2}+1}}, x^*}.
\end{align*}
For any $k\in\mathbb{Z}_{\geq3}$, we can verify that
\begin{align*}
\frac{k-z^*+1}{\floor{k/2}+1}&\geq\frac{k-\ceil{z_0}+1}{\floor{k/2}+1}=\frac{k-\Ceil{\frac{\sqrt{\floor{k/2}^2+8\floor{k/2}(\floor{k/2}+1)(k+1)}-\floor{k/2}}{4(\floor{k/2}+1)}}+1}{\floor{k/2}+1}\geq 1.
\end{align*}

Thus, it is feasible to set $\gamma^*=\ln\lrA{\frac{k-z^*+1}{\floor{k/2}+1}}$. 
Note that $\gamma^*=\ln\left(\frac{k-\ceil{z_0}+1}{\floor{k/2}+1}\right)$ if $g(\ceil{z_0})\geq g(\floor{z_0})$, and $\gamma^*=\ln\left(\frac{k-\floor{z_0}+1}{\floor{k/2}+1}\right)$ otherwise. Then, LP-UITP with $\gamma=\gamma^*$ achieves an approximation ratio of at most $\max\{g(\ceil{z_0}), g(\floor{z_0})\}$.

It can be verified that the approximation ratio obtained in Case~2 is worse than that obtained in Case~1.
Therefore, Algorithm~\ref{alg} can achieve the result from this theorem.
\end{proof}

\section{Concluding Remarks}
In this paper, we consider $k$-CVRP in general metrics and improve the approximation ratio for $k$ less than a sufficiently large value, approximately $k\leq 1.7\times 10^7$.
Although most of our algorithms, such as EX-ITP, are simple or extensions of the classic methods, the analysis is technically involved. We need to carefully analyze the structure of the solutions and most of our analysis is based on pure combinatorial analysis. We think that analyzing better theoretical bounds for simple and classic algorithms is an interesting and important task in algorithm research.
We also believe that our methods have potential to be applied to more routing problems.

\section*{Acknowledgments}
A preliminary version of this paper~\cite{DBLP:conf/mfcs/zhao} was presented at the 50th International Symposium on Mathematical Foundations of Computer Science (MFCS 2025).

\bibliographystyle{plain}
\bibliography{main}

\appendix

\section{The Best-Known Approximation Ratio}
\subsection{The Approximation Ratio of Bompadre \emph{et al.}'s Algorithm}\label{bomoadre}
\begin{lemma}[Bompadre \emph{et al.}~\cite{BompadreDO06}, Section~3.2.3]
For splittable $k$-CVRP and unit-demand $k$-CVRP with $k\geq3$, given an $\alpha$-approximation algorithm for metric TSP with $\alpha\geq1$, there is an approximation algorithm with an approximation ratio of $1-\frac{\beta}{\beta+1}+(1-\frac{1}{k})\alpha=\alpha+1-\frac{\alpha}{k}-\Omega(\frac{1}{k^3})$, where $\beta$ is the positive root of the following quadratic equation
\[
0=(2k^2-2k)\beta^2+\lrA{2-\alpha+\frac{\alpha}{k}-\frac{2}{k+1}+2k^2-2k}\beta+\lrA{1-\alpha+\frac{\alpha}{k}-\frac{2}{k+1}}.
\]
\end{lemma}

\subsection{Calculating the Approximation Ratio for the Splittable Case}\label{best-splittable}
In this section, we show that by using the main result in~\cite{blauth2022improving}, the approximation ratio of splittable $k$-CVRP is at least $\frac{5}{2}-\frac{1.005}{3000}-\frac{1.5}{k}$.

\subsubsection{The Main Result of the Algorithm}
Fix a constant $\varepsilon>0$. Blauth \emph{et al.}~\cite{blauth2022improving} distinguishes two kinds of instances.
\begin{itemize}
    \item If $(1-\varepsilon) w(I^*)\geq \frac{2}{k}\Delta$, the instance is \emph{simple}.
    \item Otherwise, the instance is \emph{difficult}.
\end{itemize}

\begin{lemma}[\cite{blauth2022improving}]\label{good}
For hard instances, there is a function $f: \mathbb{R}_{>0}\rightarrow\mathbb{R}_{>0}$ with $\lim_{\substack{\varepsilon\rightarrow 0}}f(\varepsilon)=0$ and a polynomial-time algorithm to obtain a Hamiltonian cycle $H$ on $V\cup\{v_0\}$ such that $w(H)\leq (1+f(\varepsilon)) w(I^*)$.
\end{lemma}
The function $f$ is put in Appendix~\ref{thefunction}.
The above lemma shows that we can find a good Hamiltonian cycle in the difficult instances when $\varepsilon$ is small.

Given a Hamiltonian cycle $H$ on $V\cup\{v_0\}$, by Lemma~\ref{AGITP}, AG-ITP computes a solution with weight at most $\frac{2}{k}\Delta+\frac{k-1}{k}w(H)\leq \frac{2}{k}\Delta+w(H)$.

\subsubsection{A Simple Analysis}
By calculation, we show that using a weaker result $\frac{2}{k}\Delta+w(H)$ of AG-ITP, the approximation ratio in~\cite{blauth2022improving} is at least $\frac{5}{2}-\frac{1.005}{3000}$.

\textbf{Case~1: easy instances.} We have $(1-\varepsilon) w(I^*)\geq \frac{2}{k}\Delta$ by definition. Using an $\alpha$-approximate Hamiltonian cycle, by Lemma~\ref{lb-tsp}, we have
\[
\frac{2}{k}\Delta+\alpha w(H^*)\leq (1-\varepsilon) w(I^*)+\alpha w(I^*)=(\alpha+1-\varepsilon) w(I^*).
\]

\textbf{Case~2: hard instances.} We have $w(I^*)\geq \frac{2}{k}\Delta$ by Lemma~\ref{lb-delta}.
Using the Hamiltonian cycle $H$ in Lemma~\ref{good}, we have
\[
\frac{2}{k}\Delta+w(H)\leq  w(I^*)+(1+f(\varepsilon)) w(I^*)=(2+f(\varepsilon)) w(I^*).
\]
Therefore, the approximation ratio of splittable $k$-CVRP is
\[
\min_{\varepsilon>0}\max\{\alpha+1-\varepsilon,2+f(\varepsilon)\}=\alpha+1-\varepsilon^*,
\]
where $\varepsilon^*$ is the maximum value satisfying $f(\varepsilon)+\varepsilon\leq \alpha-1$.

When $\alpha=\frac{3}{2}$, by calculation\footnote{The code is available at \url{https://github.com/JingyangZhao/CVRP}.}, we have
\[
\frac{1.005}{3000}-6\times10^{-9}<\varepsilon^*<\frac{1.005}{3000}-5\times10^{-9}.
\]
Thus, the approximation ratio is at least $\frac{5}{2}-\frac{1.005}{3000}$.

\subsubsection{A Tighter Analysis}
By calculation, we show that using the stronger result $\frac{2}{k}\Delta+\frac{k-1}{k}w(H)$ of AG-ITP, the approximation ratio in~\cite{blauth2022improving} is at least $\frac{5}{2}-\frac{1.005}{3000}-\frac{1.5}{k}$.

\textbf{Case~1: easy instances.} We have $(1-\varepsilon) w(I^*)\geq \frac{2}{k}\Delta$ by definition. Using an $\alpha$-approximate Hamiltonian cycle, by Lemma~\ref{lb-tsp}, we have
\[
\frac{2}{k}\Delta+\frac{k-1}{k}\alpha w(H^*)\leq (1-\varepsilon) w(I^*)+\frac{k-1}{k}\alpha w(I^*)=\lrA{\alpha+1-\varepsilon-\frac{\alpha}{k}}w(I^*).
\]

\textbf{Case~2: hard instances.} We have $w(I^*)\geq \frac{2}{k}\Delta$ by Lemma~\ref{lb-delta}.
Using the Hamiltonian cycle $H$ in Lemma~\ref{good}, we have
\[
\frac{2}{k}\Delta+\frac{k-1}{k}w(H)\leq  w(I^*)+\frac{k-1}{k}(1+f(\varepsilon)) w(I^*)=\lrA{1+\frac{k-1}{k}(1+f(\varepsilon)}w(I^*).
\]
Therefore, the approximation ratio of splittable $k$-CVRP is
\[
\min_{\varepsilon>0}\max\lrC{\alpha+1-\varepsilon-\frac{\alpha}{k},1+\frac{k-1}{k}(1+f(\varepsilon))}=\alpha+1-\varepsilon^{**}-\frac{\alpha}{k},
\]
where $\varepsilon^{**}$ is the maximum value satisfying $\frac{k-1}{k}f(\varepsilon)+\varepsilon\leq \frac{k-1}{k}(\alpha-1)$.

Recall that $\varepsilon^*$ is the maximum value satisfying $f(\varepsilon)+\varepsilon\leq \alpha-1$. Then, we have $f(\varepsilon^*)+\varepsilon^*= \alpha-1$ and $\frac{k-1}{k}f(\varepsilon^*)+\varepsilon^*= \frac{k-1}{k}(\alpha-1-\varepsilon^*)+\varepsilon^*> \frac{k-1}{k}(\alpha-1)$. 

Thus, we have $\varepsilon^{**}<\varepsilon^*<\frac{1.005}{3000}$ when $\alpha=\frac{3}{2}$. Then, we have
\[
\alpha+1-\varepsilon^{**}-\frac{\alpha}{k}>\alpha+1-\varepsilon^*-\frac{\alpha}{k}>\frac{5}{2}-\frac{1.005}{3000}-\frac{1.5}{k}.
\]
Therefore, the approximation ratio in~\cite{blauth2022improving} is at least $\frac{5}{2}-\frac{1.005}{3000}-\frac{1.5}{k}$.

\subsection{Calculating the Approximation Ratio for the Unsplittable Case}\label{best-unsplittable}
In this section, we show that by using the main results in~\cite{blauth2022improving,uncvrp}, the approximation ratio of unsplittable $k$-CVRP is at least $\frac{5}{2}+\ln2+\ln(1-\frac{1.005}{3000})-\frac{3}{k}$ for even $k$. Note that, for odd $k$, the approximation ratio, obtained by doubling the vehicle capacity and all customers' demand, is then at least $\frac{5}{2}+\ln2+\ln(1-\frac{1.005}{3000})-\frac{1.5}{k}$. 

Since Lemma~\ref{good} applies to the unsplittable case~\cite{blauth2022improving}, we also consider easy and difficult instances defined in the previous section.

Given a Hamiltonian cycle $H$ on $V\cup\{v_0\}$, AG-UITP~\cite{altinkemer1987heuristics} computes a solution with weight at most $\frac{4}{k}\Delta+\frac{k-2}{k}w(H)$ for even $k$.

\subsubsection{A Simple Analysis}
Using a weaker result $\frac{4}{k}\Delta+w(H)$ of AG-UITP, the LP-based algorithm in~\cite{uncvrp} has weight of at most
\[
\min_{\gamma\geq0}\lrC{\gamma\cdot  w(I^*)+e^{-\gamma}\cdot \frac{4}{k}\Delta+w(H)}.
\]

We show that its approximation ratio is at least $\frac{5}{2}+\ln2+\ln(1-\frac{1.005}{3000})$.

\textbf{Case~1: easy instances.} We have $(1-\varepsilon) w(I^*)\geq \frac{2}{k}\Delta$ by definition. Using an $\alpha$-approximate Hamiltonian cycle, by Lemma~\ref{lb-tsp}, we have
\begin{align*}
\min_{\gamma\geq0}\lrC{\gamma\cdot  w(I^*)+e^{-\gamma}\cdot \frac{4}{k}\Delta+\alpha w(H^*)}&\leq \min_{\gamma\geq0}\lrC{\gamma+e^{-\gamma}\cdot2(1-\varepsilon)+\alpha} w(I^*)\\
&\leq (\alpha+1+\ln2+\ln(1-\varepsilon))w(I^*).
\end{align*}

\textbf{Case~2: hard instances.} We have $w(I^*)\geq \frac{2}{k}\Delta$ by Lemma~\ref{lb-delta}.
Using the Hamiltonian cycle $H$ in Lemma~\ref{good}, we have
\begin{align*}
\min_{\gamma\geq0}\lrC{\gamma\cdot   w(I^*)+e^{-\gamma}\cdot \frac{4}{k}\Delta+w(H)}&\leq \min_{\gamma\geq0}\lrC{\gamma+e^{-\gamma}\cdot2+(1+f(\varepsilon))}w(I^*)\\
&\leq (2+\ln2+f(\varepsilon))w(I^*).
\end{align*}

Therefore, the approximation ratio is
\[
\min_{\varepsilon>0}\max\{\alpha+1+\ln2+\ln(1-\varepsilon),2+\ln2+f(\varepsilon)\}=\alpha+1+\ln2+\ln(1-\varepsilon^*),
\]
where $\varepsilon^*$ is the maximum value satisfying $f(\varepsilon)-\ln(1-\varepsilon)\leq \alpha-1$.

When $\alpha=\frac{3}{2}$, by calculation\footnote{The code is available at \url{https://github.com/JingyangZhao/CVRP}.}, we also have
\[
\frac{1.005}{3000}-6\times10^{-9}<\varepsilon^*<\frac{1.005}{3000}-5\times10^{-9}.
\]
Thus, the approximation ratio is at least $\frac{5}{2}+\ln2+\ln(1-\frac{1.005}{3000})$.

\subsubsection{A Tighter Analysis}
We show that using the stronger result $\frac{4}{k}\Delta+\frac{k-2}{k}w(H)$ of AG-ITP, the approximation ratio is at least $\frac{5}{2}+\ln2+\ln(1-\frac{1.005}{3000})-\frac{3}{k}$ for even $k$.

\textbf{Case~1: easy instances.} We have $(1-\varepsilon) w(I^*)\geq \frac{2}{k}\Delta$ by definition. Using an $\alpha$-approximate Hamiltonian cycle, by Lemma~\ref{lb-tsp}, we have
\begin{align*}
\min_{\gamma\geq0}\lrC{\gamma\cdot  w(I^*)+e^{-\gamma}\cdot \frac{4}{k}\Delta+\frac{k-2}{k}\alpha w(H^*)}&\leq\min_{\gamma\geq0}\lrC{\gamma+e^{-\gamma}\cdot2(1-\varepsilon)+\frac{k-2}{k}\alpha}w(I^*)\\
&\leq \lrA{\alpha+1+\ln2+\ln(1-\varepsilon)-\frac{2\alpha}{k}}w(I^*).
\end{align*}

\textbf{Case~2: hard instances.} We have $w(I^*)\geq \frac{2}{k}\Delta$ by Lemma~\ref{lb-delta}.
Using the Hamiltonian cycle $H$ in Lemma~\ref{good}, we have
\begin{align*}
\min_{\gamma\geq0}\lrC{\gamma\cdot  w(I^*)+e^{-\gamma}\cdot \frac{4}{k}\Delta+\frac{k-2}{k}w(H)}&\leq\min_{\gamma\geq0}\lrC{\gamma+e^{-\gamma}\cdot2+\frac{k-2}{k}(1+f(\varepsilon))}w(I^*)\\
&\leq \lrA{1+\ln2+\frac{k-2}{k}(1+f(\varepsilon))}w(I^*).
\end{align*}

Therefore, the approximation ratio is
\begin{align*}
&\min_{\varepsilon>0}\max\lrC{\alpha+1+\ln2+\ln(1-\varepsilon)-\frac{2\alpha}{k},1+\ln2+\frac{k-2}{k}(1+f(\varepsilon))}\\
&=\alpha+1+\ln2+\ln(1-\varepsilon^{**})-\frac{2\alpha}{k},
\end{align*}
where $\varepsilon^{**}$ is the maximum value satisfying $\frac{k-2}{k}f(\varepsilon)-\ln(1-\varepsilon)\leq \frac{k-2}{k}(\alpha-1)$.

Recall that $\varepsilon^*$ is the maximum value satisfying $f(\varepsilon)-\ln(1-\varepsilon)\leq \alpha-1$. Then, we have $f(\varepsilon^*)-\ln(1-\varepsilon^*)=\alpha-1$ and $\frac{k-2}{k}f(\varepsilon^*)-\ln(1-\varepsilon^*)=\frac{k-2}{k}(\alpha-1+\ln(1-\varepsilon^*))-\ln(1-\varepsilon^*) > \frac{k-2}{k}(\alpha-1)$. 

Thus, we have $\varepsilon^{**}<\varepsilon^*<\frac{1.005}{3000}$ when $\alpha=\frac{3}{2}$. Then, we have
\[
\alpha+1+\ln2+\ln(1-\varepsilon^{**})-\frac{2\alpha}{k}>\frac{5}{2}+\ln2+\ln(1-\varepsilon^*)-\frac{3}{k}> \frac{5}{2}+\ln2+\ln\lrA{1-\frac{1.005}{3000}}-\frac{3}{k}.
\]
Therefore, the approximation ratio in~\cite{uncvrp} is at least $\frac{5}{2}+\ln2+\ln(1-\frac{1.005}{3000})-\frac{3}{k}$ for even $k$.

\subsection{The Function}\label{thefunction}
Let
\[
\zeta=\frac{3\rho+\tau-4\tau\cdot\rho}{1-\rho}+\frac{\varepsilon}{\tau\cdot\rho}\cdot\lrA{1-\tau\cdot\rho-\frac{3\rho+\tau-4\tau\cdot\rho}{1-\rho}}.
\]

There are two functions, $f$ and $f'$, defined in~\cite{blauth2022improving}. The function $f$, obtained via an LP-based algorithm, yields better results. Thus, we use only $f$, which is defined as follows:
\[
f(\varepsilon)=\min_{\substack{0<\theta\leq1-\tau \\ 0<\tau,\rho\leq \frac{1}{6}}} \lrC{\frac{1+\zeta}{\theta}+\frac{1-\tau-\theta}{\theta\cdot(1-\tau)}+\frac{3\varepsilon}{1-\theta}+\frac{3\rho}{(1-\rho)\cdot(1-\tau)}}-1.
\]
Note that $\theta$ was set to $1-\tau$ in~\cite{blauth2022improving}, but it is slightly better to set it as
\[
\theta=\min\lrC{\frac{1}{\sqrt{\frac{3\varepsilon}{2+\zeta}}+1}, 1-\tau}.
\]
\end{document}